%% file: ms.tex
\RequirePackage[l2tabu,orthodox]{nag}

\documentclass[openright,appendixprefix,headsepline,footsepline,footinclude=false,fontsize=11pt,paper=a4,listof=totoc,bibliography=totoc,BCOR=12mm,DIV=12]{scrbook} 

\usepackage[
    footsepline=0.25pt, 
]{scrlayer-scrpage}
\ModifyLayer[addvoffset=.6ex]{scrheadings.foot.odd}
\ModifyLayer[addvoffset=.6ex]{scrheadings.foot.even}
\ModifyLayer[addvoffset=.6ex]{scrheadings.foot.oneside}
\ModifyLayer[addvoffset=.6ex]{plain.scrheadings.foot.odd}
\ModifyLayer[addvoffset=.6ex]{plain.scrheadings.foot.even}
\ModifyLayer[addvoffset=.6ex]{plain.scrheadings.foot.oneside}

\input{settings}
\input{macros}

\newcommand*{\getUniversity}{University of Oxford}
\newcommand*{\getCollege}{Wolfson College}
\newcommand*{\getFaculty}{Mathematical Institute}
\newcommand*{\getDegree}{MSc in Mathematics and Foundations of Computer Science}
\newcommand*{\getTitle}{Decision Procedures for\\ Guarded Logics}
\newcommand*{\getCandidateNr}{1026380}
\newcommand*{\getAuthor}{Kevin Kappelmann}

\newcommand*{\getSupervisor}{Prof. Michael Benedikt}
\newcommand*{\getSubmissionDate}{31st August 2018}

\begin{document}
\setboolean{anon}{false}
\setboolean{publish}{true}

\ifthenelse{\boolean{publish}}{}{\let\cleardoublepage\par}

\ifthenelse{\boolean{publish}}{\input{pages/cover}}{}

\frontmatter

\ifthenelse{\boolean{publish}}{
\ifthenelse{\boolean{anon}}{}{\input{pages/dedication}\input{pages/acknowledgements}}
\input{pages/abstract}
}
{}
\microtypesetup{protrusion=false}
\tableofcontents
\microtypesetup{protrusion=true}

\mainmatter

\input{chapters/intro}

\input{chapters/prelims}

\input{chapters/one_pass}

\input{chapters/gtgd}

\input{chapters/dgtgd}

\input{chapters/comparison}

\input{chapters/conclusion}

\appendix
\microtypesetup{protrusion=false}
\microtypesetup{protrusion=true}

\counterwithin{thm}{chapter}
\counterwithin{lem}{chapter}
\counterwithin{cor}{chapter}
\input{pages/appendix}
\backmatter
\bibliographystyle{plainnat}
\bibliography{bibliography}
\printindex[not]
\printindex[ind]

\end{document}

%% file: settings.tex
\PassOptionsToPackage{table,svgnames,dvipsnames}{xcolor}
\usepackage[T1]{fontenc}
\usepackage[utf8]{inputenc}
\usepackage[english]{babel}
\usepackage[autostyle]{csquotes}
\usepackage{xargs} 
\usepackage{xparse}
\usepackage{setspace}

\input{tikz_settings}
\usepackage{float}
\usepackage{standalone}
\usepackage{caption}
\usepackage{subcaption}
\usepackage{xcolor}

\usepackage{xpatch}
\xapptocmd{\appendix}{%
  \addtocontents{toc}{%
    \RedeclareSectionCommand[
      tocdynnumwidth,
      tocentrynumberformat=\tocappendixnumber
    ]{chapter}%
  }%
}{}{\PatchFailed}
\newcommand\tocappendixnumber[1]{\chapapp~#1}

\BeforeTOCHead[toc]{{\cleardoublepage\pdfbookmark[0]{\contentsname}{toc}}}

\usepackage{ifthen}
\newboolean{anon}
\newboolean{publish}

\usepackage{acronym}
\usepackage[nonewpage]{imakeidx}

\makeindex[intoc,columns=2,name=not,title={Notation}]
\makeindex[intoc,columns=2,name=ind,title={Index}]

\makeatletter
\AtBeginDocument{%
\LetLtxMacro\oldindex\index%
  \RenewDocumentCommand{\index}{om}{%
    \begingroup%
    \phantomsection%
    \IfValueTF{#1}{%
      \imki@wrindexentry{#1}{#2|hyperlink{\@currentHref}}{\thepage}%
    }{%
      \imki@wrindexentry{ind}{#2|hyperlink{\@currentHref}}{\thepage}%
    }%
    \endgroup%
    }%
}%
\makeatother

\newcommandx{\iindex}[2][1=ind]{\index[#1]{#2}#2}
\newcommandx{\indexsee}[3][1=ind]{\oldindex[#1]{#2|see{#3}}}
\newcommandx{\iindexsee}[3][1=ind]{\indexsee[#1]{#2}{#3}#2}

\usepackage{url}
\definecolor{pastelblue}{RGB}{0,72,205}
\usepackage[hidelinks]{hyperref}
\hypersetup{
 colorlinks,
 linktoc=page,
 linkcolor=pastelblue,
 citecolor=pastelblue,
 urlcolor=pastelblue
}

\setkomafont{disposition}{\normalfont\bfseries} 
\usepackage{xspace}
\usepackage{booktabs}
\usepackage{paralist}
\usepackage{listings}
\usepackage{cmap}           
\usepackage{enumitem}
\usepackage[final]{microtype}

\usepackage{latexsym}
\usepackage{amsmath}
\usepackage{mathtools}
\usepackage{amssymb}
\usepackage{stmaryrd}

\input{proof_settings}

\usepackage[authoryear]{natbib}

\usepackage[colorinlistoftodos,prependcaption,textsize=tiny]{todonotes}



%% file: tikz_settings.tex
\usepackage{tikz}
\usetikzlibrary{calc}

\definecolor{brown}{RGB}{146,73,0}

\tikzset{->,>=latex,every node/.style={shape=rectangle,draw,},
scope/.style={draw=none},
lnode/.style={draw=none},
trigger/.style={red,line width=1.7pt},
guardok/.style={blue,dashed},
guardnok/.style={brown,dotted},
pics/drawv/.style args={#1/#2/#3/#4}{
	code = {
	\draw[-,dashed] let \p1=(#1.east),\p2=(#2.west),\p3=(#3.north),\p4=(#4.south) in
	({(\x1+\x2)/2},\y3) -- ({(\x1+\x2)/2},\y4);
	}
},
pics/drawvr/.style args={#1/#2}{
	code = {
	\pic{drawv=#1/#2/#2/#2};
	}
},
pics/drawh/.style args={#1/#2/#3/#4}{
	code = {
	\draw[-,dashed] let \p1=(#1.south),\p2=(#2.north),\p3=(#3.west),\p4=(#4.east) in
	(\x3,{(\y1+\y2)/2}) -- (\x4,{(\y1+\y2)/2});
	}
}
}

%% file: proof_settings.tex
\usepackage{mathpartir}
\usepackage{amsthm}
\usepackage{thmtools}
\usepackage{chngcntr}
\declaretheorem[
    name=Theorem,
    Refname={Theorem,Theorems},
    numberwithin=section]{thm}
\declaretheorem[
    name=Lemma,
    Refname={Lemma,Lemmas},
    sibling=thm]{lem}
\declaretheorem[
    name=Proposition,
    Refname={Proposition,Propositions},
    sibling=thm]{prop}
\declaretheorem[
    name=Corollary,
    Refname={Corollary,Corollaries},
    sibling=thm]{cor}
\declaretheorem[
    name=Claim,
    Refname={Claim,Claims},
    parent=thm]{claim}
\let\oldproof\proof
\let\oldendproof\endproof

\newenvironment{subproof}[1][Subproof]{\begingroup\oldproof[#1]}{\oldendproof \endgroup}

\theoremstyle{definition}
\declaretheorem[
    name=Definition,
    Refname={Definition,Definitions},
    sibling=thm]{defn}
\declaretheorem[
    name=Example,
    Refname={Example,Examples},
    sibling=thm]{exmpl}
\theoremstyle{remark}
\declaretheorem[
    name=Remark,
    Refname={Remark,Remarks},
    sibling=thm]{rmk}

\newlist{thmlist}{enumerate}{1}
\setlist[thmlist,1]{label=\textup{(\alph*)}, ref=\thethm.\textup{(\alph*)},noitemsep}
\addtotheorempostheadhook[thm]{\crefalias{thmlisti}{thm}}
\addtotheorempostheadhook[lem]{\crefalias{thmlisti}{lem}}
\addtotheorempostheadhook[prop]{\crefalias{thmlisti}{prop}}
\addtotheorempostheadhook[cor]{\crefalias{thmlisti}{cor}}
\addtotheorempostheadhook[defn]{\crefalias{thmlisti}{defn}}
\addtotheorempostheadhook[exmpl]{\crefalias{thmlisti}{exmpl}}
\addtotheorempostheadhook[rmk]{\crefalias{thmlisti}{rmk}}
\addtotheorempostheadhook[claim]{\crefalias{thmlisti}{claim}}

\newlist{propylist}{enumerate}{1}
\setlist[propylist]{label=\textup{(\alph*)}, ref=\thethm.\textup{(\alph*)},noitemsep}
\addtotheorempostheadhook[defn]{\crefalias{propylisti}{propy}}

\newlist{assmlist}{enumerate}{1}
\setlist[assmlist,1]{label={\textup{(\roman*)}}, ref=\textup{(\roman*)},noitemsep} 

\newlist{invarlist}{enumerate}{1}
\setlist[invarlist,1]{label={\textup{(\arabic*)}}, ref=\textup{(\arabic*)},noitemsep} 

\usepackage[capitalize,sort,compress,noabbrev]{cleveref}
\crefdefaultlabelformat{#2\textup{#1}#3}
\Crefname{thm}{Theorem}{Theorems}
\Crefname{lem}{Lemma}{Lemmas}
\Crefname{prop}{Proposition}{Propositions}
\Crefname{cor}{Corollary}{Corollary}
\Crefname{defn}{Definition}{Definitions}
\Crefname{exmpl}{Example}{Examples}
\Crefname{rmk}{Remark}{Remarks}
\Crefname{propy}{Property}{Properties}
\Crefname{claim}{Claim}{Claim}
\Crefname{assmlisti}{Assumption}{Assumptions}
\Crefname{invarlisti}{Invariant}{Invariants}

\usepackage{thm-restate}
\newcommand{\restatablespacefix}{\vspace{-\topsep}}

%% file: macros.tex
\newcommand{\kw}[1]{{\mathsf{#1}}\xspace}

\let\otodo\todo
\newcommandx{\unsurenote}[2][1=]{\otodo[linecolor=orange,backgroundcolor=orange!25,bordercolor=orange,#1]{UNSURE: #2}}
\newcommandx{\todonote}[1]{\otodo[linecolor=red,backgroundcolor=red!25,bordercolor=red]{TODO: #1}}
\renewcommand{\todo}[1]{\textcolor{red}{\textbf{TODO}: #1}}

\newcommand{\calA}{\mathcal{A}}
\newcommand{\calB}{\mathcal{B}}
\newcommand{\calC}{\mathcal{C}}
\newcommand{\calD}{\mathcal{D}}

\newcommand{\calF}{\mathcal{F}}

\newcommand{\calM}{\mathcal{M}}

\newcommand{\calO}{\mathcal{O}}
\newcommand{\calP}{\mathcal{P}}

\newcommand{\calR}{\mathcal{R}}

\newcommand{\calT}{\mathcal{T}}

\newcommand{\calV}{\mathcal{V}}

\newcommand{\nth}[1]{#1^{\text{th}}}
\newcommand{\varset}{\calV}
\newcommand{\constset}{\calC}

\newcommand{\funcset}{\calF}
\newcommand{\termset}{\calT}
\newcommand{\dom}{\kw{dom}}
\newcommand{\vars}{\kw{vars}}
\newcommand{\consts}{\kw{consts}}
\newcommand{\hnf}{\kw{HNF}}
\newcommand{\shnf}{\kw{SHNF}}
\newcommand{\skol}{\kw{Sk}}
\newcommand{\vnf}{\kw{VNF}}
\newcommand{\rel}{\calR}
\newcommand{\relset}{\calR}
\newcommand{\db}{\calD}

\newcommand{\width}{\kw{width}}
\newcommand{\hwidth}{\kw{hwidth}}
\newcommand{\bwidth}{\kw{bwidth}}
\newcommand{\comp}{\textup{(COMPOSITION)}\xspace}
\newcommand{\orig}{\textup{(ORIGINAL)}\xspace}
\newcommand{\evolve}{\textup{(EVOLVE)}\xspace}
\newcommand{\devolve}{\textup{(D-EVOLVE)}\xspace}

\newcommand{\func}[1]{#1_{+}}
\newcommand{\nfunc}[1]{#1_{-}}
\newcommand{\lvs}{\kw{lvs}}
\newcommand{\sclo}{\kw{SSat}}
\newcommand{\gclo}{\kw{GSat}}
\newcommand{\gcloex}{\kw{GSat}_\exists}
\newcommand{\dgclo}{\kw{DGSat}}
\newcommand{\dgcloex}{\kw{DGSat}_\exists}
\newcommand{\size}{\kw{size}}

\newcommand{\cut}{\kw{cut}}
\newcommand{\ifc}{\kw{IFC}}
\renewcommand{\phi}{\varphi}
\newcommand{\wrt}{with respect to\xspace}
\newcommand{\ifftext}{if and only if\xspace}

\newcommand{\ptime}{\kw{PTIME}}
\newcommand{\exptime}{\kw{EXPTIME}}
\newcommand{\texptime}{\kw{2EXPTIME}}
\newcommand{\conexptime}{\kw{coNEXPTIME}}

\newcommand{\tgd}{TGD\xspace}
\newcommand{\tgds}{TGDs\xspace}
\newcommand{\dtgd}{DisTGD\xspace}
\newcommand{\dtgds}{DisTGDs\xspace}
\newcommand{\gtgd}{GTGD\xspace}
\newcommand{\gtgds}{GTGDs\xspace}
\newcommand{\dgtgd}{DisGTGD\xspace}
\newcommand{\dgtgds}{DisGTGDs\xspace}
\newcommand{\gf}{GF\xspace}

\newcommand{\body}{\beta}
\newcommand{\head}{\eta}
\newcommand{\cq}{CQ\xspace}
\newcommand{\cqs}{CQs\xspace}
\newcommand{\bcq}{BCQ\xspace}
\newcommand{\bcqs}{BCQs\xspace}
\newcommand{\ucq}{UCQ\xspace}

\newcommand{\ubcq}{UBCQ\xspace}
\newcommand{\ubcqs}{UBCQs\xspace}
\newcommand{\qfqaprob}{Quantifier-Free Query Answering Problem\xspace}
\newcommand{\qfqaprobgtgds}{\qfqaprob under \gtgds}
\newcommand{\qfqaprobdgtgds}{\qfqaprob under \dgtgds}
\newcommand{\qfqagtgds}{Quantifier-Free Query Answering under \gtgds}
\newcommand{\qfqadgtgds}{Quantifier-Free Query Answering under \dgtgds}

%% file: pages/cover.tex
\begin{titlepage}
  \oddsidemargin=\evensidemargin\relax
  \textwidth=\dimexpr\paperwidth-2\evensidemargin-2in\relax
  \hsize=\textwidth\relax

\begin{center}
  \null

  {\textsc{\Huge\getUniversity{}}\\
  \vspace{5mm}
   \Large\getFaculty{}
  }

  \vspace{10mm}
  \IfFileExists{logos/beltcrest.pdf}{%
    \hspace{1mm}
    \includegraphics[height=50mm]{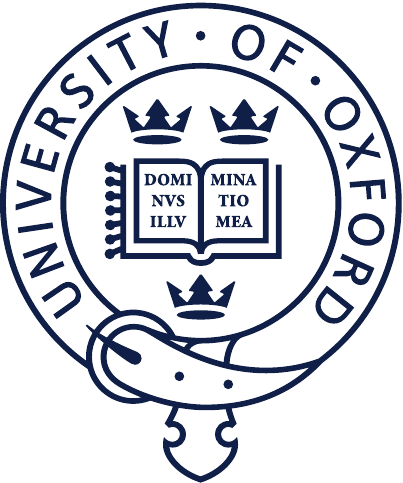}
  }{%
    \vspace*{49mm}
  }

  \vspace*{10mm}
  {
	\begin{spacing}{1.1}
		\Huge\bfseries \getTitle{}
	\end{spacing}
  }

  \vspace*{10mm}
  \ifthenelse{\boolean{anon}}
{
  {\Large Candidate Number: \getCandidateNr{}}\\
  \vspace{30mm}
}
{
  {\Large \getAuthor{}}\\
  \vspace{3mm}
  {\large \getCollege{}}\\
  \vspace{3mm}
  {\large supervised by \getSupervisor{}}\\
  \vspace{20mm}
}

  {\large A thesis submitted in partial fulfilment for the degree of}\\
  \vspace{3mm}
  {\large \textit{\getDegree{}}}\\
  \vspace{10mm}

\ifthenelse{\boolean{anon}}
{
  {\large Trinity 2018}\\
  \vspace{3mm}
  {\large \getSubmissionDate{}}
}
{
  {\large Trinity 2018}
}
\end{center}

\end{titlepage}

%% file: pages/dedication.tex
\thispagestyle{empty}

\vspace*{\stretch{1}}
\begin{center}
\raggedleft
\textit{
To my father, Hans-Jürgen Kappelmann $\dagger 16.01.2018$}
\end{center}
\vspace{\stretch{3}}

\cleardoublepage{}

%% file: pages/acknowledgements.tex
\thispagestyle{empty}
\vspace*{20mm}

\begin{center}
\textbf{{\usekomafont{section}Acknowledgments}}
\end{center}

\vspace{10mm}
I want to thank my supervisor Prof. Michael Benedikt for all
his support, especially his great suggestions for structuring this thesis,
appropriately abstracting and presenting the results, and 
introducing me to an interesting research topic.

I also want to thank all my friends that visited me 
during my overwhelming year in Oxford:
Alex, Bene, Cedric, Domi B., Jonny, Liz, Lukas, Netti, Niklas, Philipp, Silja, Tom, and Vali.
Words cannot describe how thankful I am to have you as part of my life.
Moreover, I want to thank all my new friends that I met in Oxford, 
especially Domi W., Philip, and Vasilii,
and all my extremely talented MFoCS coursemates for our invaluable
discussions, exchange of ideas, and, most importantly, collective 
suffering and consumption of biscuits in the common room.

Finally, I want to thank my mother, Erika, for always believing in me and 
her unconditional support and love.

\cleardoublepage{}

%% file: pages/abstract.tex
\thispagestyle{empty}
\vspace*{20mm}

\begin{center}
\textbf{{\usekomafont{section}\abstractname}}
\end{center}

\vspace{10mm}

An important class of decidable first-order logic fragments are those
satisfying a guardedness condition, such as the guarded fragment (\gf). 
Usually, decidability for these logics is closely linked
to the tree-like model property -- the fact that satisfying models
can be taken to have tree-like form. 
Decision procedures for the guarded fragment based on the tree-like model 
property are difficult to implement. 
An alternative approach, based on restricting first-order resolution, 
has been proposed, and this shows more promise from the point of view of implementation. 
In this work, we connect the tree-like model property of the guarded fragment
with the resolution-based approach.
We derive efficient resolution-based rewriting algorithms that 
solve the Quantifier-Free Query Answering Problem
under Guarded Tuple Generating Dependencies (\gtgds) and
Disjunctive Guarded Tuple Generating Dependencies (\dgtgds). 
The Query Answering Problem for these classes subsumes many 
cases of \gf satisfiability.
Our algorithms, in addition to making the connection to the tree-like model 
property clear, give a natural account of the selection and 
ordering strategies used by resolution procedures for the guarded fragment.
We also believe that our rewriting algorithm for the special case of \gtgds 
may prove itself valuable in practice as it does not require 
any Skolemisation step and its theoretical runtime outperforms those 
of known \gf resolution procedures in case of fixed dependencies.
Moreover, we show a novel normalisation procedure for the widely used chase procedure
in case of (disjunctive) \gtgds, which could be useful for future studies.

\cleardoublepage{}

%% file: chapters/intro.tex
\chapter{Introduction}

\input{pages/introduction/motivation}
\input{pages/introduction/contributions}
\input{pages/introduction/thesis_outline}

%% file: pages/introduction/motivation.tex
\section{Motivation}\label{sec:motivation}

The satisfiability problem for first-order logic is known to be
undecidable, as shown by \cite{church1936unsolvable} and \cite{turing1937computable}.
Nevertheless, some expressive fragments of first-order logic have been shown to be decidable.
One very rich family of decidable fragments is the family of so-called guarded logics.
Many guarded languages enjoy desirable characteristics, such as the finite model property, Craig interpolation, 
and  the tree-like model property \citep{gradel1999restraining}.
The historically first representative of the guarded family is the guarded fragment~(\gf), introduced by \cite{andreka1998modal}.
Its satisfiability problem has been solved in a variety of ways,
and arguably the most natural one is to use the tree-like model property of \gf:
\begin{itemize}
\item One can use the tree-like model property to reduce to the monadic second-order logic over trees, 
which was shown decidable by \cite{rabin1969decidability}.
\item A refinement of above argument is to directly generate an automaton that accepts 
codes of tree-like models.
One then checks for non-emptiness of the language of the automaton.
\end{itemize}
Both techniques, however, fall short in different ways.
The first approach is clearly infeasible in practice as
it relies on Rabin's decision procedure for the monadic second-order logic over trees, which has non-elementary complexity;
that is, its complexity is not bounded by any power of exponentials.
The second approach has complexity bounded by two exponentials, which is tight in the worst-case.
But the building of the automaton always hits the worst-case time complexity.

An approach that seems superficially quite different is to use resolution, 
as outlined by \cite{de1998resolution} and \cite{ganzinger1999superposition}. 
This approach gives rise to a deterministic algorithm that has the possibility of stopping much earlier
than the worst-case bounds on typical cases.
However, it is not clear how it connects to the tree-like model property exploited in the techniques described above. 
The existing completeness proofs rely on general results on resolution with selection \citep{ganzinger1999superposition},
or a resolution game for non-liftable orders \citep{de1998resolution}, techniques that seem very distinct from the tree-like model property.

The goal of this thesis is to give a new resolution based algorithm
for a special case of \gf satisfiability. We will prove completeness
of the method via analysis of the tree-like model property.

%% file: pages/introduction/contributions.tex
\section{Contributions}\label{sec:contribution}

We work in this thesis not directly with the guarded fragment, but
with some fragments that make the connection with the tree-like model property clearer.
We use Disjunctive Guarded Tuple Generating Dependencies (\dgtgds), which
represent a guarded version of disjunctive, existential Datalog.
Our main contribution is a completeness proof of a resolution-based rewriting algorithm from \gtgds to Datalog and from \dgtgds to disjunctive Datalog that directly uses the tree-like model property.
This language does not comprise all of \gf, but it captures a substantial fragment
that enjoys a wide range of applications in database theory (e.g.\ deductive databases and data exchange)
and artificial intelligence (e.g.\ knowledge representation and reasoning) \citep{gottlob2012complexity,bourhis2016guarded}.

To this end, we will contribute to an important technique used for query evaluation and containment, 
called the chase, that has found its way into numerous studies, such as \citep{cali2010querying,baget2011walking,benedikt2015querying}, to name but a few.
More specifically, we show in~\cref{chap:one_pass} that the chase not only admits a tree-like property for \gtgds,
but can also be tamed to admit a one-pass behaviour in which chase inferences are more closely related to the tree structure.
This result is then generalised to show an analagous one-pass property for disjunctive chases under \dgtgds.

We then use this close connection 
to formulate an algorithm-independent property that implies completeness of any procedure satisfying the property.
This is done independently for \gtgds in \cref{sec:gtgd_property_for_completeness}
and then naturally extended to \dgtgds in \cref{sec:dgtgd_property_for_completeness}.
Due to its algorithm-independent formulation, these properties can be used for work beyond this thesis.

Following that, we derive two theoretically optimal \gtgd rewriting algorithms in \cref{sec:gtgd_simple_saturation,sec:gtgd_guarded_saturation}.
Notably, their theoretical runtime outperforms those of known \gf resolution procedures 
in the case of a fixed set of dependencies.
We then generalise our ideas in \cref{sec:dgtgd_guarded_saturation} to derive an efficient \dgtgd rewriting algorithm.
The completeness of all algorithms is shown by means of our abstract completeness properties, which explains the connection to the tree-like model property.
Moreover, our algorithms give a natural account of the selection and ordering
strategies used by resolution procedures for the guarded fragment.

%% file: pages/introduction/thesis_outline.tex
\section{Thesis Outline}

We start this thesis by briefly recalling the required background for our upcoming algorithms and proofs.
In particular, we give a brief account of first-order logic
and introduce our considered fragments, namely the guarded fragment in \cref{sec:gf}, and the relational model and \dgtgds in \cref{sec:rel_model}.
We then formulate the problem statement and
explain our general solution approach in \cref{sec:problem_statement}.
Our approach will split the problem into a rewriting part, 
and a simpler, second step, which can be solved by well-known saturation processes described in \cref{sec:atomic_rewriting_dec_proc}.
In \cref{sec:proof_machinery}, we then add an essential proof technique to our toolbox -- the chase.

It will be a common pattern throughout this thesis that we perform an initial
investigation for the simpler language of \gtgds and then
extend our results to the more expressive disjunctive case.
That being said, our main work begins in \cref{chap:one_pass},
where we first prove a normalisation result for the chase under \gtgds,
which we call the one-pass property.
We then naturally extend this property to the disjunctive chase to conclude the chapter.

In \cref{chap:gtgd}, we present an abstract property that implies completeness of \gtgd rewriting algorithms.
The completeness proof based on these properties will make use of the one-pass property shown in the previous chapter.
We then present a simple resolution-based algorithm in \cref{sec:gtgd_simple_saturation} that satisfies the formulated completeness property;
however, this simple algorithm experiences a serious unification issue, and we thus derive a second, improved version in \cref{sec:gtgd_guarded_saturation} that solves this issue.

\cref{chap:qfqadgtgds} then generalises and extends our results to the disjunctive case in an analogous way;
that is, we again first formulate an abstract completeness property and then derive a resolution-based \dgtgd rewriting algorithm that satisfies this property.

We then compare our algorithms with the known resolution procedures for the guarded fragment in \cref{chap:comparison}.
In particular, we examine which fragments of \gf we can cover with our chosen framework.

To conclude the thesis, we give a brief summary of our results as well as 
some final thoughts about possible future developments and related work.

%% file: chapters/prelims.tex
\chapter{Preliminaries}\label{chap:prelims}
In this chapter, we first introduce some background knowledge
and define basic notions used throughout this thesis.
We then set our problem statement and explain our general solution idea in \cref{sec:problem_statement}.
Our approach will split the problem into two parts,
the second of which can be solved by very simple, well-known saturation processes that are described in \cref{sec:atomic_rewriting_dec_proc}.
Finally, we introduce some required machinery 
that we will use for the correctness proofs of our upcoming algorithms.

\input{pages/prelims/first_order_logic/first_order_logic}

\input{pages/prelims/gf}
\input{pages/prelims/relational_model}

\input{pages/prelims/query_answering}
\input{pages/prelims/atomic_rewriting_dec_proc/atomic_rewriting_dec_proc}

\input{pages/prelims/proof_machinery/proof_machinery}

%% file: pages/prelims/first_order_logic/first_order_logic.tex
\section{First-Order Logic}\label{sec:prelims_first_order_logic}
First-order logic is a formal language that allows,
for example, the formalisation of reasoning,
axiomatisation of mathematical systems,
and specification of software and hardware systems.
All problems considered in this thesis are formulated
in first-order logic, and we assume a basic knowledge
of its syntax and semantics.
Introductions can be found, for example,
in \citep{abiteboul1995foundations,hodges1997shorter}.
In this section, we briefly introduce some syntactic conventions
and recall some (perhaps less well-known) concepts and results, 
namely the notion of a homomorphism, Skolemisation, and Herbrand's Theorem.

To speak about the syntax of first-order logic, 
we fix an infinite set of \emph{variables} $x_1,x_2,\dotso$, denoted by~\iindex[not]{$\varset$}, and
a \emph{signature} $\sigma=(\constset,\funcset,\rel)$, where
\begin{itemize}
\item \iindex[not]{$\constset$} is a set of \emph{constants} $c_1,c_2,\dotso$,
\item $\funcset$ is a set of \emph{function symbols} $f_1,f_2,\dotso$ with associated arities $a_{f_i}\geq 1$, and
\item \iindex[not]{$\rel$} is a set of \emph{relation symbols} $R_1,R_2,\dotso$ with associated arities $a_{R_i}\geq 1$.
\end{itemize} 
We denote the set of terms by \iindex[not]{$\termset$}
and often abbreviate a sequence of terms $t_1,\dotsc,t_n$ as~\iindex[not]{$\vec{t}$};
we also allow ourselves to treat $\vec t$ as a set when convenient.
Given a formula $\phi$, we denote the set of free variables of $\phi$ by \iindex[not]{$\vars(\phi)$}
and the set of used constants of $\phi$ by \iindex[not]{$\consts(\phi)$}.
Similarly, we define $\vars(t),\consts(t)$ for the set of used variables and constants of a term $t$.
A formula without free variables is also called a \emph{sentence},
and a formula $\phi=R(t_1,\dotsc,t_n)$, where $R\in\relset$, is called an \emph{\iindex{atom}}.
Unless otherwise noted, we write~$\phi(\vec x)$ to indicate that 
the set of free variables of $\phi$ is among~$\vec x$. That is, we write $\vars(\phi)\subseteq\vec x$, 
and analogously, $\phi(\vec c)$ to indicate that the set of constants of $\phi$ is among~$\vec c$.
These definitions are lifted in the natural way to sets of formulas $S$, for example $\vars(S)\coloneqq\bigcup_{\phi\in S}\vars(\phi)$.

Given a signature $\sigma$ and a sentence $\phi$, in order to speak about
the semantics of $\phi$, we also have to fix a \emph{structure} $\calA$ consisting of
\begin{itemize}
\item a set of elements \iindex[not]{$\dom(\calA)$}, called the \emph{domain of $\calA$},
\item an assignment $c^\calA\in\dom(\calA)$ for any $c\in\constset$, 
\item an assignment $f^\calA:\dom(\calA)^a\rightarrow \dom(\calA)$ for any $f\in\funcset$ of arity $a$, and
\item an assignment $R^\calA\subseteq \dom(\calA)^a$ for any $R\in\relset$ of arity $a$.
\end{itemize} 
\begin{defn}[Satisfiability]
Given a structure $\calA$ and sentence $\phi$, we say that \emph{$\calA$ satisfies $\phi$} or \emph{$\calA$ models $\phi$}, written \iindex[not]{$\calA\models\phi$}, 
if the interpretation of $\phi$ under $\calA$ evaluates to ``true'' using the usual semantics.\footnote{see, for example, \citep[Section~2.4]{abiteboul1995foundations} for a formal definition.}
We say that $\phi$ is \emph{satisfiable} if there is some structure $\calA$ that models $\phi$, and we call $\phi$ \emph{valid} if every structure is a model for $\phi$.
These definitions are lifted to set of sentences $S$ in the usual way,
for example, we write $\calA\models S$ if $\calA\models\phi$ for every $\phi\in S$.
Lastly, given two sets of sentences $S$ and $S'$, we say that \emph{$S$ entails $S'$}, written $S\models S'$, if every model $\calA$ of $S$ is also a model of $S'$.
\end{defn} 
Given two $\sigma$-structures $\calA,\calB$, there are different notions
that allow us to compare the structures.
One of them is that of a homomorphism.
Intuitively, a homomorphism between $\calA$ and $\calB$ is a map 
that preserves constants, functions, and relations from $\calA$ to $\calB$.
More formally:
\begin{defn}[Homomorphism]
Let $\calA,\calB$ be two $\sigma$-structures.
A \emph{\iindex{homomorphism} from $\calA$ to $\calB$} is a map 
$h:\dom(\calA)\rightarrow\dom(\calB)$ such that:
\begin{thmlist}
\item For any constant $c$ of $\sigma$,
we have $h(c^\calA)=c^\calB$.
\item For any function $f$ of $\sigma$ and elements $\vec a\in\dom(\calA)$,
we have $h(f^\calA(\vec a))=f^\calB(h(\vec a))$.
\item For any relation $R$ of $\sigma$ and elements $\vec a\in\dom(\calA)$,
we have $R^\calA(\vec a)\implies R^\calB(h(\vec a))$.
\end{thmlist} 
For convenience, we write $h:\calA\rightarrow \calB$ if we have a homomorphism from $\calA$ to $\calB$.
\end{defn} 

Next, we recall the idea of \emph{\iindex{Skolemisation}}, 
a technique commonly applied in automated theorem proving.
Skolemisation is used to remove existential variables of a given first-order formula,
and hence gives rise to a normal form transformation in which all quantifications are universal.
The basic idea is to replace any existential variable $y$ that is in the scope of 
universal quantifiers binding variables $\vec x$ by a functional term $f_y(\vec x)$,
where $f_y$ is a fresh function symbol. 
The fresh functions $f_y$ are commonly called \emph{Skolem functions}.
\begin{exmpl}
The Skolemisation of 
\begin{equation*}
	\forall x_1\exists y_1\forall x_2\exists y_2[R(x_1,y_2)\land P(y_1, x_1, x_2)]
\end{equation*} 
is
\begin{equation*}
	\forall x_1\forall x_2[R(x_1,f_{y_2}(x_1,x_2))\land P(f_{y_1}(x_1), x_1, x_2)].
\end{equation*} 
\end{exmpl} 
\begin{thm}[{\citep[Chapter~3]{hodges1997shorter}}]\label{thm:Skolem_equisat}
A formula $\phi$ is satisfiable \ifftext its Skolemisation $\phi'$ is satisfiable.
\end{thm} 

Lastly, we state a simple form of a fundamental result of mathematical logic: Herbrand's Theorem.
It allows a reduction of first-order logic validity to validity of propositional logic.
We will make use of it in our procedure given in \cref{sec:atomic_rewriting_dec_proc_dtgd},
but we will not require it for any main result of this thesis.
\begin{thm}[\cite{herbrand1930recherches}]\label{thm:herbrand}
A set of clauses $S$ is unsatisfiable \ifftext
there is a finite unsatisfiable set $S'$ of ground instances of clauses in $S$.
\end{thm} 

\input{pages/prelims/first_order_logic/unification}
\input{pages/prelims/first_order_logic/resolution}

%% file: pages/prelims/first_order_logic/unification.tex
\subsection{Unification}
An important step in our algorithms will be the unification of atoms.
First-order unification has its origin in the seminal paper of \cite{robinson1965machine}
that introduced the resolution calculus of first-order logic.
In this section, we introduce basic definitions and results from the literature
that we will later make us of.

\begin{defn}[Substitution]
A \emph{\iindex{substitution}} is a partial function $\rho:\varset\rightarrow\termset$.
\end{defn} 
Given a substitution $\rho$, we let $\dom(\rho)$ denote the domain of $\rho$ 
and assume $x\rho\coloneqq\rho(x)=x$ for ${x\notin\dom(\rho)}$.
Given two substitutions $\rho,\rho'$, we write $\rho\rho'$ for the composition
of the substitutions; that is, $\rho\rho'\coloneqq\rho'\circ\rho$.
Moreover, we write \index[not]{$[t_i/x_j]$}$[t_1/x_1,\dotsc,t_n/x_n]$ for a substitution~$\rho$ 
defined by ${x_i\rho\coloneqq t_i}$ for $1\leq i\leq n$.
Lastly, given a formula $\phi(\vec x)$, 
we write $\phi\rho$ for the formula obtained by replacing any free 
variable $x_i$ in $\phi$ by $\rho(x_i)$.
\begin{exmpl}
We have
\begin{align*}
	R(x_1,x_2)[y/x_1,f(c)/x_2]&=R(y,f(c)),\\
	R(x_1,x_2)[x_2/x_1]&=R(x_2,x_2),\\
	R(x_1,x_2)[x_2/x_1,x_1/x_2]&=R(x_2,x_1),\\
	R(x_1,x_2)[x_2/x_1][x_1/x_2]&=R(x_1,x_1).
\end{align*} 
\end{exmpl} 
\begin{defn}[Renaming]
A substitution $\rho$ is a \emph{\iindex{renaming}} if $\rho$ is injective on $\dom(\rho)$ and $x\rho$ is a variable for any $x\in\dom(\rho)$.
\end{defn} 
\begin{lem}\label{lem:renaming_iff_inverse_subst}
A substitution $\rho$ is a renaming \ifftext it has an inverse substitution.
\end{lem} 
\begin{proof}
If $\rho$ is a renaming, then $\rho^{-1}$ is a partial function on $\termset$ since $\rho$ is injective, 
and since $x\rho$ is a variable for any $x\in\dom(\rho)$, it is a substitution.
If $\rho$ has an inverse substitution~$\rho^{-1}$, then $\rho$ must be injective, 
and for any $x\in\dom(\rho)$, $x\rho$ must be a variable 
as $\rho^{-1}$ is defined on $x\rho$ and $\rho^{-1}$ is a substitution.
\end{proof} 
\begin{defn}[Unifier]
A \emph{\iindex{unifier}} of atoms $A_1,\dotsc,A_n$ and $B_1,\dotsc,B_n$ is a substitution $\theta$ satisfying $A_i\theta=B_i\theta$ for $1\leq i\leq n$.
\end{defn} 
In theory, there might be many different unifiers of atoms $A_1,\dotsc,A_n$ and $B_1,\dotsc,B_n$.
Practically, however, one does not want to consider the vast set of all possible unifiers.
Luckily, it often suffices to only consider one special unifier:
\begin{defn}[Most general unifier]
A unifier $\theta$ of atoms $A_1,\dotsc,A_n$ and $B_1,\dotsc,B_n$ is called a \emph{\iindex{most general unifier}} (\iindexsee{mgu}{most general unifier})
if for any unifier $\rho$ of $A_1,\dotsc,A_n$ and $B_1,\dotsc,B_n$, there is a substitution $\delta$ such that
$\rho=\theta\delta$.
\end{defn} 

\begin{lem}\label{lem:mgu_uniqueness}
For any unifiable atoms $A_1,\dotsc,A_n$ and $B_1,\dotsc,B_n$,
the mgu is unique up to renaming.
\end{lem} 
\begin{proof}
Let $\theta,\theta'$ be two mgus of the given atoms.
Then there are substitutions $\delta,\delta'$ such that
$\theta=\theta'\delta'$ and $\theta'=\theta\delta$.
Thus, $\theta=\theta\delta\delta'$ and $\theta'=\theta'\delta'\delta$, 
i.e.\ $\delta\delta'=id$ and $\delta'\delta=id$.
Hence, $\delta,\delta'$ are mutual inverse bijections, and thus renamings by \cref{lem:renaming_iff_inverse_subst}.
\end{proof} 

The problem of obtaining an mgu for given atoms $A_1,\dotsc,A_n$ and 
$B_1,\dotsc,B_n$ is rather old
and many different algorithms, some even linear, exist.
\begin{lem}\label{lem:mgu_complexity}
The mgu of atoms $A_1,\dotsc,A_n$ and $B_1,\dotsc,B_n$
can be obtained in time 
$\calO\bigl(\sum_{i=1}^n\size(A_i)+\size(B_i)\bigr)$,
where $\size(A)$ is the encoding size of an atom $A$.
\end{lem} 
\begin{proof}
A linear procedure can be found, for example,
in \citep{Paterson:1976:LU:800113.803646}.
\end{proof} 

Some of our algorithms will only have to deal with function-free atoms,
for which the following lemma will be a useful helper.
\begin{lem}\label{lem:mgu_function_free}
Assume $A_1,\dotsc,A_n$ and $B_1,\dotsc,B_n$ are unifiable and function-free,
and let $\theta$ be an mgu. 
Then all $A_i\theta=B_i\theta$ are function-free.
\end{lem} 
\begin{proof}
Assume otherwise. 
Pick some variable $x$.
Since all given atoms are function-free and unifiable, the substitution $\theta'$ defined by
$v\theta'\coloneqq x$ for any $v\in\varset$ is a unifier of given atoms and all $A_i\theta'=B_i\theta'$ are function-free.
Thus, there is a substitution $\delta$ such that $\theta'=\theta\delta$.
Pick some $A_i$ such that $A_i\theta$ contains a functional subterm $f(\vec t\,)$.
Then $f(\vec t\,)\delta=f(\vec t\delta)$ and hence, $A_i\theta'=A_i\theta\delta$ is functional,
a contradiction.
\end{proof} 

%% file: pages/prelims/first_order_logic/resolution.tex
\subsection{Resolution}

In this section, we very briefly recall first-order logic resolution.
A detailed introduction can be found in \citep{bachmair2001resolution}.

Resolution, as introduced by \cite{robinson1965machine}, is a complete 
theorem proving method for the unsatisfiability 
problem of first-order logic, and undeniably, it builds 
the foundation of some of the most successful automated theorem proving procedures.
Resolution operates on a formula in conjunctive normal form, usually
represented by a set of \index{clause}\emph{clauses}.
A clause ${C=\{\lnot L_1,\dotsc,\lnot L_n,L_1',\dotsc,L_m'\}}$
corresponds to a disjunction of (negated) atoms, also called \emph{\iindex{literals}}.
The most important inference rule of the resolution calculus is the binary resolution rule:
\begin{equation*}
\inferrule
  {C\cup\{L\}\\D\cup\{\lnot L'\}}
  {(C\cup D)\theta}
\end{equation*} 
where $C,D$ are clauses and $\theta$ is an mgu of $L$ and $L'$.

In our work, we will usually work with clauses that contain at least 
one negative literal.
Note that any such clause ${C=\{\lnot L_1,\dotsc,\lnot L_n,L_1',\dotsc,L_m'\}}$ 
can be written as an inference rule $\bigwedge_{i=1}^nL_i\rightarrow\bigvee_{i=1}^m L_i'$.

There are many extensions of the standard resolution calculus, such as
ordered resolution, paramodulation, and superposition, that increase
the expressiveness or efficiency of the calculus, but we will not have to directly 
deal with the rules of these extensions.

%% file: pages/prelims/gf.tex
\section{The Guarded Fragment}\label{sec:gf}
In this section, we briefly introduce the \emph{\iindex{guarded fragment}} (\iindexsee{\gf}{guarded fragment}).
The core idea of guarded logics is to ``guard'' certain operations by using atoms in the language.
In the case of \gf, the operations being guarded are the universal and existential quantification of variables.

\begin{defn}
The formulas of \gf are defined inductively as follows:
\begin{propylist}
\item The atomic formulas of \gf are $\top$, $\bot$, 
any function-free atom $R(\vec x,\vec c)$, 
and any equality atom $t_1\approx t_1$ for $t_1,t_2\in\varset\cup\constset$.
\item If $\phi,\psi$ are in \gf, then $\lnot\phi$, $\phi\land\psi$, and $\phi\lor\psi$ are in \gf.
\item If $\phi$ and $G$ are in \gf and $G$ is an atom with $\vars(\phi)\subseteq\vars(G)$, then $\exists\vec x(G\land \phi)$ and $\forall\vec x(G\rightarrow \phi)$ are in \gf for any finite $\vec x\subseteq\varset$.
The atom $G$ is called the \emph{guard of $\phi$}.
\end{propylist} 
\end{defn} 
It is noteworthy that the original definition of \gf in \citep{andreka1998modal}
does not allow constants in its formulas.
This feature has been added by \cite{gradel1999restraining},
who at the same time also proved the following result:
\begin{lem}[\citep{gradel1999restraining}]
For any \gf sentence~$\phi$, one can construct in polynomial time 
a \gf sentence $\phi'$ without constants that is satisfiable \ifftext $\phi$ is satisfiable.
\end{lem} 
\begin{proof}
Let $\vec c=\{c_1,\dotsc,c_k\}=\consts(\phi)$.
For every $m$-ary relation symbol $R$ of $\phi$,
introduce a new $(k+m)$-ary relation symbol $R^*$.
Further, let $P$ be a $k$-ary fresh relation symbol.
Then the sentence 
\begin{equation*}
\phi'\coloneqq\exists\vec c\,\bigl(P(\vec c)\land\phi[R^*(\vec c,\vec x)/R(\vec x)]\bigr)
\end{equation*} 
is in \gf, where by abuse of notation, $\phi[R^*(\vec c,\vec x)/R(\vec x)]$ is the formula obtained by replacing any atom of the form $R(\vec x)$ in $\phi$ by the corresponding atom $R^*(\vec c,\vec x)$.
The constants $\vec c$ are now variables in $\phi'$, and it is clear that $\phi$ and $\phi'$
are satisfiable over the same domains.
\end{proof} 
In particular, the resolution procedures given by \cite{de1998resolution} and \cite{ganzinger1999superposition} use a constant-free definition of \gf.
\begin{exmpl}
The following sentences are in \gf:
\begin{align*}
	&\forall x[x\approx x\rightarrow R(x)],
	&R(c_1,c_2)\land(\lnot P(c_3)\lor \lnot S(c_1,c_2)),\\
	&\forall x_1,x_2\bigl[R(x_1,x_2)\rightarrow \lnot R(x_2,x_1)\bigr],
	&\exists y_1[G(y_1)\land \lnot R(y_1)],\\
	&\forall x_1,x_2\bigl[R(x_1,x_2)\rightarrow \exists y[R(x_2,y)\land \top]\bigr],
	&\exists y\bigl[G(y)\land[\forall x(R(x,y)\rightarrow x\approx y)]\bigr].
\end{align*} 
A typical sentence that is not expressible in \gf is the simple sentence $\forall x_1,x_2\, R(x_1,x_2)$ or the transitivity axiom $\forall x_1,x_2,x_3[R(x_1,x_2)\land R(x_2,x_3)\rightarrow R(x_1,x_3)]$.
\end{exmpl} 

%% file: pages/prelims/relational_model.tex
\section{The Relational Model, Queries, and Dependencies}\label{sec:rel_model}
As mentioned before, we will not directly work with general \gf formulas,
but instead consider problems that are closely related to database querying,
for which we fix a function-free first-order signature $\sigma=(\constset,\emptyset,\relset)$.
In the spirit of database theory, we introduce some special terminology in this section.

The set $\relset$ of relational symbols is called a \emph{\iindex{relational schema}}.
A \emph{\iindex{fact}} is a function-free atom $R(\vec c)$ without free variables.
An \emph{instance of a relation} $R\in\rel$ is a set of facts of the form $R(\vec c)$.
An \emph{\iindex{instance} $I$ of a relational schema $\relset$} is a union of instances for each $R\in\rel$.
We also call such an instance a \emph{\iindex{database}} \iindex[not]{$\db$}.
Note that every instance $I$ can be considered as a $\sigma$-structure with 
domain $\constset$ mapping any constant $c$ to the element of the same name.
Since $I$ is just a set of facts, we write $\consts(I)$ for the set of constants that occur in at least one fact of $I$;
that is, $\consts(I)\coloneqq\{c\mid\exists F\in I:c\in\consts(F)\}$.

Given a database $\db$, we naturally want to query it.
One simple query fragment is given by the class of \emph{conjunctive queries}.
Although these queries are very simple, they constitute the vast
majority of relational database queries arising in practice \citep{abiteboul1995foundations}.

A \emph{\iindex{conjunctive query}} (\iindexsee{\cq}{conjunctive query}) over a database $\db$ with schema $\relset$ 
is a formula of the form
\begin{equation*}
	Q(\vec x)=\exists\vec y\,(R_1\land\dotso\land R_n),
\end{equation*} 
with $R_i\in\rel$ and each $R_i$ only uses arguments in $\vec x$, $\vec y$, or constants from $\consts(\db)$.
If $Q$ contains no free variables, 
we call $Q$ a \emph{\iindex{boolean conjunctive query}} (\iindexsee{\bcq}{boolean conjunctive query}),
and if $Q$ is a single fact, we call $Q$ an \emph{\iindex{atomic query}}.

A \emph{\iindex{union of conjunctive queries}} (\iindexsee{\ucq}{union of conjunctive queries}) 
is a disjunction of \cqs. 
Similarly, a \emph{\iindex{union of boolean conjunctive queries}} (\iindexsee{\ubcq}{union of boolean conjunctive queries})
is a disjunction of \bcqs. 
By abuse of notation, we sometimes consider a \ucq (\ubcq) $Q$ as a set of \cqs (\bcqs).
That is, we write $Q_i\in Q$ to mean ``$Q_i$ is a disjunct occurring in $Q$''.

Finally, we often want to enforce certain constraints on databases.
Practically, such constraints can come in a variety of forms, such as functional dependencies, referential integrity constraints, inclusion dependencies, join dependencies, etc.
Intuitively, many constraints are based on the idea that whenever
some tuples satisfy some condition, then some other tuples must also exist or 
some values must be equal, as described by \cite{beeri1984proof}.
They took this thought and gave a unified first-order logic framework for database constraints.
One prominent part of this framework is constituted by the class of tuple-generating dependencies.

A \emph{\iindex{tuple-generating dependency}} (\iindexsee{\tgd}{tuple-generating dependency}) is a function-free first-order sentence of the form
\begin{equation*}
	\tau=\forall \vec x \bigl[\body(\vec x) \rightarrow \exists \vec y\,\head(\vec x, \vec y)\bigr],
\end{equation*} 
where $\body$ and $\head$ are conjunctions of constant-free atoms, and
$\vars(\body)=\vec x$.
The conjunction~$\body$ is referred to as the \emph{\iindex{body}} and the conjunction $\head$ as the \emph{\iindex{head}} of $\tau$.

A \emph{\iindex{disjunctive \tgd}} (\iindexsee{\dtgd}{disjunctive TGD}) is a function-free first-order sentence of the form 
\begin{equation*}
	\forall \vec x [\body(\vec x) \rightarrow \bigvee_i^n\exists \vec y_i ~ \head_i(\vec x, \vec y_i)],
\end{equation*} 
where $\body$ and each $\head_i$ are conjunctions of constant-free atoms, and $\vars(\body)=\vec x$. 
We refer to $\head_i$ as the \emph{$\nth{i}$ \iindex{head conjunction}} of $\tau$,
and call $\head_i$ a \emph{\iindex{non-full head conjunction}} if it contains some existential variable. 
The variables $x_i\in\vars(\body)$ that also occur in some head conjunction of the 
\dtgd are called the \emph{exported variables}.
We will use \iindex[not]{$\vec x_{i\restriction}$} to refer to the set of exported variables of the $\nth{i}$ head conjunction of a \dtgd.

A \emph{\iindex{full disjunctive \tgd}} is one that has no existentials in the head. 
If it is also a \tgd, we call it a \emph{\iindex{full \tgd}}. 

A \emph{\iindex{disjunctive guarded \tgd}} (\iindexsee{\dgtgd}{disjunctive guarded \tgd}) is one where there is a single atom
in $\body$ that contains all of the variables of $\body$. 
Such an atom is called the \emph{\iindex{guard} of $\tau$}.
If $\tau$ is also a \tgd, we call it a \emph{\iindex{guarded \tgd}} (\iindexsee{\gtgd}{guarded \tgd}).

To avoid redundancy, we treat the body, each head conjunction, and the head of a (disjunctive) \tgd as sets.
For brevity, we will omit the universal quantifiers in front of (disjunctive) \tgds and 
assume a minimal quantification for all occurring variables.
For example, the \dgtgd
\begin{equation*}
	\forall x_1\bigl[B(x_1)\land B(x_1)\rightarrow \exists y_1,y_2\,H(y_2,x_1)\lor\exists y_1,y_2\bigl(H(y_2,x_1)\land H(y_2,x_1)\bigr)\bigr]
\end{equation*} 
will be treated as the pair
\begin{equation*}
	\Bigl(\{B(x_1),B(x_1)\},\bigl\{\{H(y_2,x_1)\},\{H(y_2,x_1),H(y_2,x_1)\}\bigr\}\Bigr)=\Bigl(\{B(x_1)\},\bigl\{\{H(y_2,x_1)\}\bigr\}\Bigr)
\end{equation*} 
and hence be simplified to 
\begin{equation*}
	\forall x_1\bigl[B(x_1)\rightarrow \exists y_2\,H(y_2,x_1)\bigr]
\end{equation*} 
and then be written as
\begin{equation*}
	B(x_1)\rightarrow \exists y_2\,H(y_2,x_1).
\end{equation*}

%% file: pages/prelims/query_answering.tex
\section{Problem Statement: Query Answering under \dtgds}\label{sec:problem_statement}
Given some database $\db$ and a set of \dtgds $\Sigma$, 
we can consider possible extensions of~$\db$ that satisfy the dependencies $\Sigma$;
that is, databases $\db'\supseteq\db$ satisfying $\db'\models\Sigma$.
Moreover, given a query $Q$, we then might not only wonder whether some extension of $\db$
satisfies~$Q$, but whether all extensions will satisfy $Q$.
This thought captures the general idea of our problem statement,
of which we now give a formal definition:

\begin{defn}[Query Answering Problem]
Given a database $\db$, a set of (disjunctive) \tgds $\Sigma$,
a \ucq $Q(\vec x)$, and some $\vec c\subseteq\consts(\db)$, 
the \emph{Query Answering Problem} is to decide whether
the \ubcq $Q(\vec c)$ holds in every possible world for $\db,\Sigma$.
That is, we want to decide whether $\db,\Sigma\models Q(\vec c)$.
If $Q$ is a quantifier-free \ucq,
we will call it the \emph{\iindex{\qfqaprob}}.
\end{defn} 
In this work, we are interested in theories that consist of (disjunctive) \gtgds $\Sigma$ and a quantifier-free \ucq $Q$.
It is noteworthy that this problem is easily expressible as a 
satisfiability problem in \gf:
\begin{lem}\label{lem:gf_reduce}
For any database $\db$, set of \dgtgds $\Sigma$, and quantifier-free \ubcq $Q$,
we can reduce the \qfqaprob to the satisfiability problem of \gf in polynomial time.
\end{lem} 
\begin{proof}
It is well know that $\db,\Sigma\models Q$ \ifftext $\db,\Sigma,\lnot Q$ is unsatisfiable.
Note that $\db$ and $\lnot Q$ are in \gf as they only consist of boolean combinations of ground atoms.
Now any \dgtgd can be transformed into its single-head normal form, as will be described in \cref{defn:dtgd_shnf} (a \dgtgd is single-headed if each head conjunction contains exactly one atom).
Any single-headed \dgtgd of the form
\begin{equation*}
\forall\vec x \left[G(\vec x)\land\bigwedge_{i=1}^nB_i(\vec x) \rightarrow \bigvee\limits_{i=1}^m \exists \vec y_i\,H_i(\vec x,\vec y_i)\right],
\end{equation*} 
can be rewritten into \gf by rewriting implications and using De Morgan's law for quantifiers 
in the usual way to obtain a \gf sentence of the form
\begin{equation*}
\lnot\exists\vec x \left[G(\vec x)\land\bigwedge_{i=1}^nB_i(\vec x) \land \bigwedge\limits_{i=1}^m \lnot\exists \vec y_i\,H_i(\vec x,\vec y_i)\right],
\end{equation*} 
noting that $\exists\vec y_i\,H_i(\vec x,\vec y)\equiv\exists\vec y_i\,[H_i(\vec x,\vec y)\land\top]$ 
and the guard atom $G(\vec x)$ contains each~$x_i$.
\end{proof} 

While the satisfiability problem of \gf is usually considered for an arbitrary sentence $\phi$,
our problem setting is strongly motivated by its connection to database querying.
Our input consists of three independent parts -- database, dependencies, and query -- 
and it is common in practice to ask many different queries while fixing the database and constraints.
This motivates a solution approach that also separates the concerns between dependencies and queries.
We hence follow a two step approach:

The first step performs an \emph{atomic rewriting} of given dependencies,
independent of any query and even database.
\begin{defn}[Atomic rewriting]\label{defn:atomic_rewriting}
Given a set of (disjunctive) \tgds $\Sigma$,
an \emph{\iindex{atomic rewriting} of $\Sigma$} is a finite
set of full (disjunctive) \tgds $\Sigma'$ that have the same answer for any
quantifier-free \ubcq $Q$ and database $\db$ as~$\Sigma$. 
\end{defn} 
\begin{exmpl}\label{exmpl:prelims_atomic_query_answering}
Given the database $\db=\{R(c,d),P(d)\}$, a set of \dtgds $\Sigma$
\begin{align*}
R(x_1,x_2)&\rightarrow\exists y_1\,S(x_1,y_1),\\
S(x_1,x_2)\land P(x_3)&\rightarrow T(x_1)\lor\bigl(T(x_1)\land U(x_1)\bigr),\\
R(x_1,x_2)\land T(x_1)&\rightarrow U(x_2),
\end{align*} 
and query $Q=U(x)$, we have
$\db,\Sigma\models Q(d)$ but $\db,\Sigma\not\models Q(c)$.
The rules
\begin{align*}
R(x_1,x_2)\land P(x_3)&\rightarrow T(x_1)\lor T(x_1)\land U(x_1)\\
R(x_1,x_2)\land T(x_1)&\rightarrow U(x_2)
\end{align*} 
are an atomic rewriting of $\Sigma$.
\end{exmpl} 
The rewriting step reduces the problem to the case of full (disjunctive) \tgds.
Using these full rules, the problem can then easily be solved by simple, well-known saturation algorithms
for any database $\db$ and query $Q$, 
as will be described in \cref{sec:atomic_rewriting_dec_proc_tgd,sec:atomic_rewriting_dec_proc_dtgd}.
This two-step approach does not only separate concerns, 
but is also believed to be one of the most promising approaches to 
achieve scalability of expressive query languages by reusing existing 
optimised database technologies \cite{ahmetaj2018rewriting}.

Our main goal, thus, is to derive a procedure that converts a set of (disjunctive) \gtgds into an atomic rewriting.
In order to verify the correctness of such a process, we will have to show that the returned answer $\Sigma'$
satisfies two key properties of an atomic rewriting: soundness and completeness.

\begin{defn}[Completeness]\label{defn:completeness}
Let $\Sigma,\Sigma'$ be two sets of (disjunctive) \tgds.
We say that $\Sigma'$ is \index{completeness}\emph{complete (\wrt $\Sigma$)} if $\db,\Sigma\models Q$ implies $\db,\Sigma'\models Q$.
for any database $\db$ and quantifier-free \ubcq $Q$.
\end{defn} 

\begin{defn}[Soundness]\label{defn:soundness}
Let $\Sigma,\Sigma'$ be two sets of (disjunctive) \tgds.
We say that~$\Sigma'$ is \index{soundness}\emph{sound (\wrt $\Sigma$)} if $\db,\Sigma'\models Q$ implies $\db,\Sigma\models Q$ for any database $\db$ and quantifier-free \ubcq~$Q$.
\end{defn} 

By plugging in the given definitions, we obtain the following lemma:
\begin{lem}\label{cor:atomic_rewriting_def}
Let $\Sigma,\Sigma'$ be two sets of (disjunctive) \tgds.
Then $\Sigma'$ is an atomic rewriting of~$\Sigma$
\ifftext $\Sigma'$ is finite, full, sound, and complete.
\end{lem} 
\begin{proof}
Follows straight from \cref{defn:atomic_rewriting}.
\end{proof} 

We hence want to look for ways that, given some \dgtgds $\Sigma$, return
a finite, complete, and sound set of full \dtgds $\Sigma'$.
We will do this first for \gtgds in \cref{chap:gtgd}, to give the idea, and then extend to \dgtgds in \cref{chap:qfqadgtgds}.

%% file: pages/prelims/atomic_rewriting_dec_proc/atomic_rewriting_dec_proc.tex
\section{Decision Procedures for Atomic Rewritings}\label{sec:atomic_rewriting_dec_proc}
To motivate our solution approach, the next two sections explain how any atomic rewriting 
can be used as a decision procedure for the {\qfqaprob}.
Even though we will only be interested in theories that consist of (disjunctive) \gtgds,
the procedures described in \cref{sec:atomic_rewriting_dec_proc_tgd,sec:atomic_rewriting_dec_proc_dtgd} work more generally for atomic rewritings of
(disjunctive) \tgds.

\input{pages/prelims/atomic_rewriting_dec_proc/atomic_rewriting_dec_proc_tgd}
\input{pages/prelims/atomic_rewriting_dec_proc/atomic_rewriting_dec_proc_dtgd}

%% file: pages/prelims/atomic_rewriting_dec_proc/atomic_rewriting_dec_proc_tgd.tex
\subsection{A Decision Procedure for Atomic Rewritings of \tgds}\label{sec:atomic_rewriting_dec_proc_tgd}
As has been mentioned, in the course of our work,
we will first derive an atomic rewriting procedure for \gtgds.
Once this will be done, we can easily obtain a decision procedure as follows:

\begin{prop}\label{prop:gtgd_atomic_rewriting_dec_proc}
Let $\db$ be a database,
$\Sigma$ be a set of \tgds,
$\Sigma'$ be an atomic rewriting of~$\Sigma$,
let $Q$ be a quantifier-free \ubcq,
and $w$ be the maximum number of variables in any rule in $\Sigma'$.
Then $\Sigma'$ provides a decision procedure for the \qfqaprob under~$\Sigma$. 
The complexity of the procedure is polynomial in $c^w$ and $\size(\db,\Sigma',Q)$.
\end{prop} 
\begin{proof}
Since $\Sigma'$ is an atomic rewriting,
we know that $Q$ follows from $\db$ and $\Sigma$ \ifftext~$Q$ follows from $\db$ and~$\Sigma'$.
Note that $\Sigma'$ is a set of full \tgds.
Such rules correspond to a language prominently known as \emph{Datalog},
see for example \citep[Chapter~12]{abiteboul1995foundations}.
We have thus reduced to the Datalog Query Answering Problem, which
is $\exptime$-complete in $w$ \citep{dantsin2001complexity}.
For the sake of completeness, we give an easy, non-optimal bottom-up approach.\footnote{The approach constructs the full chase of $\db$ under $\Sigma'$, as will be described in \cref{sec:chase}.}
A more efficient approach can be found, for example, in \citep[Chapter~12--13]{abiteboul1995foundations}.

Let $c$ be the number of constants in $\db$.
For any $\tau\in\Sigma'$, we have at most $c^w$ different ground instantiations.
For any instantiation, we can check whether all body facts are contained in $\db$.
If this is the case, we add all head facts to $\db$.
We only have to fire every rule at most once;
hence, this process takes at most $|\Sigma'|c^w$ iterations.
Once we reach the fixpoint of this saturation,
we only need to check if all facts of some $Q_i\in Q$ are contained in the final set of facts.
In total, the complexity is polynomial in $c^w$ and $\size(\db,\Sigma',Q)$.
\end{proof} 

%% file: pages/prelims/atomic_rewriting_dec_proc/atomic_rewriting_dec_proc_dtgd.tex
\subsection{A Decision Procedure for Atomic Rewritings of \dtgds}\label{sec:atomic_rewriting_dec_proc_dtgd}
Once we generalise our results to obtain an atomic rewriting procedure for \dgtgds,
we can again obtain a decision procedure in a simple way:

\begin{prop}\label{prop:dgtgd_atomic_rewriting_dec_proc}
Let $\db$ be a database,
$\Sigma$ be a set of \dtgds, 
$\Sigma'$ be an atomic rewriting of~$\Sigma$,
let $Q$ be a quantifier-free \ubcq,
$n$ be the number of relation symbols in $\db,\Sigma',Q$, 
let~$c$ be the number of constants in $\db$,
and $a$ be the maximum arity of any predicate in $\Sigma'$.
Then~$\Sigma'$ provides a decision procedure for the \qfqaprob under~$\Sigma$. 
The complexity of the procedure is polynomial in $2^{nc^a}$ and $\size(\db,\Sigma',Q)$.
\end{prop} 
\begin{proof}
Since $\Sigma'$ is an atomic rewriting, we know that $Q$ follows from $\db$ and $\Sigma$
\ifftext~$Q$ follows from $\db$ and $\Sigma'$.
Note that $\Sigma'$ is a set of full \dtgds.
Such rules can be transformed into their single-head normal form (as will be described in \cref{defn:dtgd_shnf}) so that each head conjunction of each rule contains exactly one atom.
The resulting rules correspond to a language prominently known as \emph{disjunctive Datalog}, see for example \citep{eiter1997disjunctive}.
We have thus reduced to the Disjunctive Datalog Query Answering Problem, which
is $\conexptime$-complete \citep{hustadt2007reasoning}.
Again, for the sake of completeness,
we give an easy, non-optimal approach by reducing the problem to the propositional unsatisfiability problem.
A more efficient approach can be found, for example, in \citep{cumbo2004enhancing}.
 
We know that $\db,\Sigma'\models Q$ \ifftext $\db,\Sigma',\lnot Q$ is unsatisfiable.
Now any $F\in\db$ corresponds to the ground clause $\{F\}$
and the negation of any ${(\bigwedge_{i=1}^n R_i)=Q_j\in Q}$ to the ground clause $\{\lnot R_1,\dotsc,\lnot R_n\}$.
Lastly, any single-headed rule $\bigwedge_{i=1}^n B_i\rightarrow \bigvee_{i=1}^m H_i$ corresponds
to a clause $\{\lnot B_1,\dotsc,\lnot B_n,H_1,\dotsc, H_m\}$.
By Herbrand's Theorem (\cref{thm:herbrand}), 
it suffices to consider all ground instantiations of these clauses with constants from $\db$. 
We then simply use the propositional resolution calculus (see for example \citep{fitting2012first}) on given clauses to obtain our answer.
There are at most~$2nc^a$ ground atoms, and hence at most~$2^{2nc^a}$ ground clauses.
Thus, the resolution process stops after at most $2^{2nc^a}$ iterations.
It is easy to check that every step in each iteration can be done in time polynomial in $2^{nc^a}$ and $\size(\db,\Sigma',Q)$.
Hence, we obtain the claimed complexity bound.
\end{proof}

%% file: pages/prelims/proof_machinery/proof_machinery.tex
\section{Proof Machinery -- Chasing the Truth}\label{sec:proof_machinery}

In order to show that some set $\Sigma'$ is indeed
an atomic rewriting of some (disjunctive) \gtgds $\Sigma$,
we need some method that allows us to compare 
the answers under $\Sigma$ and~$\Sigma'$ for any database $\db$ and query $Q$.
One such method is given by a procedure called \emph{the chase}.
In its simple form described in \cref{sec:chase}, the chase operates on an instance~$I$ and a set of \tgds $\Sigma$;
however, extensions to \dtgds exist and will be covered in \cref{sec:disjunctive_chase}.

\input{pages/prelims/proof_machinery/chase_sequence}
\input{pages/prelims/proof_machinery/disjunctive_chase}

%% file: pages/prelims/proof_machinery/chase_sequence.tex
\subsection{The Chase}\label{sec:chase}
The \iindex{chase}, first introduced by \cite{maier1979testing}, 
is a longstanding technique capable of checking implications of 
data dependencies in database systems.
Simply speaking, the chase repairs a database \wrt some
set of \tgds so that the resulting database satisfies all \tgds.
It takes some original instance $I_0$ 
(i.e.\ some database $\db$) and modifies the instance by a sequence of chase steps.

Let $I$ be an instance and $\tau=\body(\vec x)\rightarrow\exists\vec y\,\head(\vec x,\vec y)$ be a \tgd
with $\vec y$ possibly empty.
A \emph{\iindex{trigger}} for $\tau$ in $I$ is a mapping
${h:\vec x\rightarrow \consts(I)}$ such that $h(\body(\vec x))\coloneqq\body(h(\vec x))\subseteq I$.
Moreover, an \emph{active trigger} for
$\tau$ in $I$ is a trigger $h$ for $\tau$ in $I$ such that 
no extension $h':\vec x\cup\vec y\rightarrow \consts(I)$ of $h$ with $h'(\head(\vec x,\vec y))\subseteq I$ exists.

\begin{rmk}
Note that a trigger $h$ is just a unifier of all atoms 
in $\body$ with some facts in $I$.
\end{rmk} 

Let $h$ be an active trigger for $\tau$ in $I$. 
Applying a \emph{\iindex{chase step}} for $\tau$ and $h$ to~$I$ 
extends~$I$ with facts of the conjunction $h'\bigl(\head(\vec x,\vec y)\bigr)$, 
where $h'$ is an extension of $h$ 
and $h'(y_i)$ is a fresh constant that does not occur in $\consts(I)$ for each $y_i \in \vec y$. 
We also say that \emph{$\tau$ based on $h$ fires in $I$},
or shorter, \emph{$\tau$ fires in $I$} if there is some trigger
for which $\tau$ fires in $I$.

For $I_0$ an instance and $\Sigma$ a set of \tgds,
a \emph{\iindex{chase sequence} for $I_0$ and $\Sigma$} is a possibly infinite sequence
${I_0, \dotso}$ such that for each ${i>0}$, instance
$I_i$ (if it exists) is obtained from $I_{i-1}$ by applying a successful chase
step with some ${\tau \in \Sigma}$ and an active trigger $h$ for $\tau$ in~$I_{i-1}$. 

For a \ubcq $Q$ over $\db$, 
a finite chase sequence $I_0({=}\db),I_1,\dotsc,I_n$ for $\db$ and $\Sigma$ is 
a \index{chase proof!TGD}\emph{chase proof from $\db$ and $\Sigma$ of $Q$} if $I_n\models Q$.
The following result shows that the chase is a complete proof system:
whenever $Q$ logically follows from $\db$ and $\Sigma$, there is a chase proof.

\begin{thm}[\cite{fagin2005data}]\label{thm:chase_univ} 
A \ubcq $Q$ follows from a database $\db$ 
and a set of \tgds $\Sigma$  
\ifftext there is a chase proof of $Q$ from $\db$ and $\Sigma$.
\end{thm}

The proof involves showing that an infinite chase sequence has a homomorphism to every possible world for $\db,\Sigma$ -- one says that the chase is a \emph{universal model}.
It is a simple model-theoretic result that all existential positive first-order formulas are preserved under homomorphisms (by a basic induction on the structure of existential formulas).\footnote{In fact, the famous \emph{Homomorphism Preservation Theorem} strengthens this result: a first-order formula is preserved under homomorphisms on all structures \ifftext it is equivalent to an existential positive first-order formula.}
Putting both results together, it can be derived that a \ubcq is satisfied in all models \ifftext it is satisfied in its universal model, the chase.

In fact, there exist two different versions of the chase in the literature, called
the \emph{restricted} and the \emph{oblivious} chase.
The above described procedure corresponds to the former version.
The oblivious chase differs in that it relaxes the conditions of a chase step:
it removes the requirement that the trigger of a chase step must be active.
In other words, the oblivious chase ``forgets'' to check whether a \tgd
is already satisfied by the given instance. 
This might seem like a downside first;
however, using the oblivious chase in place of the restricted chase
sometimes facilitates proofs.
Fortunately, both versions can be used equivalently,
as shown by the following theorem.

\begin{thm}[\cite{cali2013taming}]
For any database $\db$, set of \tgds $\Sigma$, and \ubcq $Q$,
there is a restricted chase proof of $Q$ from $\db$ and $\Sigma$
\ifftext there is an oblivious chase proof of~$Q$ from $\db$ and $\Sigma$.
\end{thm} 
\begin{proof}[Proof sketch]
It is easy to show that there is a homomorphism from the oblivious chase into the restricted one.
Consequently, the oblivious chase is a universal model.
\end{proof} 

Henceforth, we collectively refer to both versions as the chase unless otherwise noted, 
and we will make use of the oblivious chase whenever needed.

\begin{exmpl}
Given the database $\db=\{R(c,d)\}$, a set of \tgds $\Sigma$
\begin{align*}
R(x_1,x_2)&\rightarrow\exists y\,S(x_1,y),\\
S(x_1,x_2)&\rightarrow T(x_1)\land U(x_2),\\
R(x_1,x_2)\land U(x_3)&\rightarrow P(x_2,x_3),
\end{align*} 
and the \cq $Q(x)=\exists y\,P(x,y)$, the sequence 
\begin{equation*}
I_0=\{R(c,d)\},I_1=I_0\cup\{S(c,d_1)\},I_2=I_1\cup\{T(c),U(d_1)\},I_3=I_2\cup\{P(d,d_1)\}
\end{equation*} 
is a chase proof from $\db$ and $\Sigma$ of $Q(d)$.
Note that $I_3$ is also the fixpoint of the restricted chase.
Since, for example, $U(d)\notin I_3$, we can also conclude that $\db,\Sigma\not\models U(d)$.
\end{exmpl} 

Even though the chase offers a complete proof system,
in its general form, it is merely a sequence of collected facts 
derived from some database $\db$ and \tgds $\Sigma$
and does not reveal any special underlying structure or interesting properties so far.
In the case of \gtgds, however, we can build the chase as a tree-like 
structure -- an important observation that we will exploit in our algorithms.
Before we show this tree-like structure property, we introduce two definitions:

\begin{defn}[Guarded subsets]
Given a set of facts $S$, a \emph{\iindex{guarded subset}} $U$ of $S$
is a subset of $S$ containing a fact $G$ that contains every constant of a fact of $U$. 
We call $G$ the \emph{\iindex{guard}} of $U$.
\end{defn} 
\begin{defn}[Guardedly-completeness]
Given two set of facts $S$ and $U$, we say that 
$U$ is \emph{\iindex{guardedly-complete} \wrt $S$} if for any fact $F$ of $S$, 
whenever $F$ is guarded by a fact in $U$, then $F$ is also in $U$.
If $U\subseteq S$, we also say that $U$ is \emph{guardedly-complete in $S$}.
\end{defn} 
\begin{exmpl}
Let $F_1=\{R(c,d),S(e)\}$ and $F_2=\{T(c,d,f),U(f)\}$.
Then $F_1$ is not guarded, $F_2$ is guarded with guard $T(c,d,f)$,
$F_1$ is guardedly-complete \wrt $F_2$, and $F_2$ is not guardedly-complete \wrt $F_1$
since $R(c,d)\notin F_2$.
\end{exmpl} 

We are now ready to show the tree-like structure property.
\begin{defn}[Guarded tree decomposition]\label{defn:tree_decomposition}
A \emph{\iindex{guarded tree decomposition}} of an instance $I$ is a directed tree $T=(V,E)$
and a function $f:V\rightarrow \calP(I)$, where $\calP(I)$ is the power set of $I$, 
such that the following hold:
\begin{propylist}
\item For any $v\in V$, $f(v)$ is guardedly-complete in $I$.
\item For any $F\in I$, there is $v\in V$ such that $F\in f(v)$.
\item \label{defn:tree_decomposition_3}
For any $c\in \consts(I)$, the set of nodes $S_c\coloneqq\{v\in V\mid c\in \consts(f(v))\}$ induces a subtree of $T$.
\end{propylist} 
We also write $F_v$ for the set of facts associated to $v$; that is, we set $F_v\coloneqq f(v)$.
The sets $F_v$ are called the \index{bag}\emph{bags} of the decomposition.
\end{defn} 
The bottom line about guarded tree decompositions 
is that whenever a bag $F_v$ guards a fact $F\in I$, then $F$ is also in the bag.

\begin{prop}\label{prop:tree_chase}
Given a chase sequence $I_0,\dotsc,I_n$ for a set of \gtgds $\Sigma$,
each instance $I_i$ can be decomposed into a tree $T_i=(V_i,E_i)$ with root $r$ such that $\consts(F_r)=\consts(I_0)$
and for any $v\neq r$, $\consts(F_v)\subseteq h(\vec x,\vec y)$ for a trigger $h$ of
some rule $\body(\vec x)\rightarrow\vec y\,\head(\vec x,\vec y)$ fired in $I_j$ for some $j<i$.
\end{prop}
In fact, we will construct the chase in such a way that once a node $v$ appears in some~$I_i$, 
it will never disappear, but its bag may grow at further steps.
We refer to the initial contents of $v$ as the \emph{\iindex{birth facts}} of $v$, 
denoted by \iindex[not]{$F_v^-$}, 
and write \iindex[not]{$F_v^i$} for the bag of $v$ at stage $i$ of the chase.
\begin{proof}[Proof of \cref{prop:tree_chase}]
Our construction will be by induction on the length of the sequence. 
The following inductive invariant will follow from the guardedly-completeness of any~$F_v^{i-1}$:

\medskip
For any trigger $h$ that fires for some rule $\tau$ with body $\body$, the image $h(\body)$ lies in some $F_v^{i-1}$ 
(we say that the chase step is \emph{triggered in $v$}).
\medskip

Our construction proceeds in the following way:
\begin{itemize}
\item The trivial chase sequence has a single root node, associated with the initial instance~$I_0$.
\item Suppose $I_i$ is formed by firing a chase step with a full \gtgd $\tau$ in $I_{i-1}$ based on a trigger $h$. 
By the inductive invariant, the image of the trigger lies in $F_v^{i-1}$ for some node $v$. 
We add the resulting facts to $v$.
\item Suppose $I_i$ is formed by firing a chase step with a non-full \gtgd $\tau$ in $I_{i-1}$ based on a trigger $h$. 
Again, by the inductive invariant, the image of the trigger lies in $F_v^{i-1}$ for some node $v$. 
We create a new node $v'$ as a child of $v$, and let $F_{v'}^i$ initially contain the newly generated facts.
As we have a guarded tree decomposition for step $i-1$ and $v'$ is a child of $v$,
the new construction still satisfies \cref{defn:tree_decomposition_3}.
\item In either of the above two cases, after firing the chase step,
we propagate every fact~$F_1$ that is guarded by another fact $F_2$ in $I_i$ to every node $v$ containing $F_2$. 
This propagation allows us to maintain the guardedly-completeness of any $F_v^{i}$.
\end{itemize}
We have to show that the invariant indeed holds at any step $i$: 
Let $h$ be the trigger for the \gtgd $\tau$ that fires at step $i$. 
Since $\tau$ is guarded, there is an atom $G(\vec x)$ in the body $\body(\vec x)$ of~$\tau$ containing all body variables. 
Let $v$ be a node at step $i-1$ containing the fact $h\bigl(G(\vec x)\bigr)$. 
By construction, $F_v^{i-1}$ is guardedly-complete in $I_{i-1}$, and 
hence all facts $h\bigl(\body(\vec x)\bigr)$ must be contained in $F_v^{i-1}$, 
i.e.\ the image of $h$ lies in $F_v^{i-1}$. 
\end{proof}

\begin{exmpl}\label{exmpl:tree_like_chase_sequence}
Given the database $\db=\{R(c,d)\}$ and a set of \gtgds $\Sigma$
\begin{align*}
R(x_1,x_2)&\rightarrow\exists y\,S(x_1,y),&
R(x_1,x_2)&\rightarrow\exists y\,T(x_1,x_2,y),\\
T(x_1,x_2,x_3)&\rightarrow\exists y\,U(x_1,x_2,y),&
U(x_1,x_2,x_3)&\rightarrow P(x_2),\\
T(x_1,x_2,x_3)\land P(x_2)&\rightarrow M(x_1),&
S(x_1,x_2)\land M(x_1)&\rightarrow\exists y\, N(x_1,y),
\end{align*} 
the sequence 
\begin{align*}
I_0&=\{R(c,d)\},I_1=I_0\cup\{S(c,d_1)\},I_2=I_1\cup\{T(c,d,d_2)\},I_3=I_2\cup\{U(c,d,d_3)\},\\
I_4&=I_3\cup\{P(d)\},I_5=I_4\cup\{M(c)\},I_6=I_5\cup\{N(c,d_4)\}
\end{align*} 
is a chase sequence for $\db$ and $\Sigma$.
The corresponding tree-like chase is depicted in \cref{fig:exmpl_tree_chase_seq}.
\begin{figure}[ht]
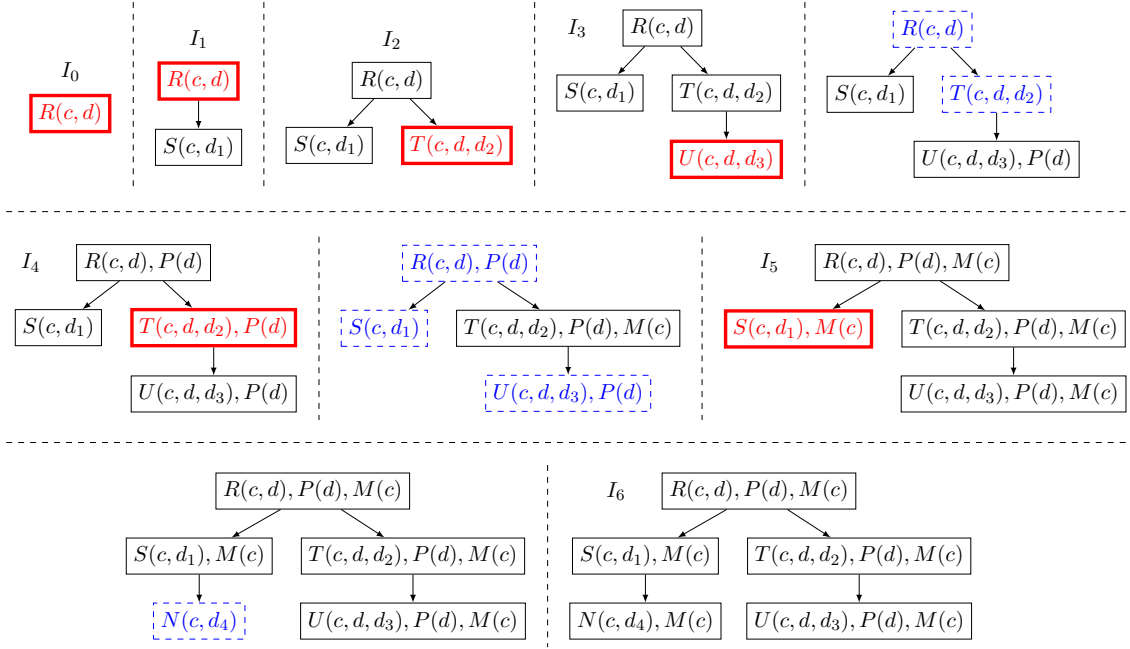

	\centering
	\includestandalone[width=\textwidth]{./figures/exmpl_tree_chase_seq}
	\caption{Tree-like chase for \cref{exmpl:tree_like_chase_sequence};
	triggered nodes are marked bold and red, and nodes with missing, non-propagated facts
	are marked dashed and blue.}
	\label{fig:exmpl_tree_chase_seq}
\end{figure}
\end{exmpl} 

%% file: figures/exmpl_tree_chase_seq.tex
\begin{tikzpicture}
\tikzset{node distance=0.5cm,
	r-0/.pic={\node (r) {$R(c,d)$};},
	I0/.pic={\pic[trigger]{r-0};},
	cl-1/.pic={\node (cl) {$S(c,d_1)$};},
	I1/.pic={%
		\pic[trigger]{r-0};
		\pic[below=of r]{cl-1};
		\path (r) edge (cl);
	},
	cr-2/.pic={\node (cr) {$T(c,d,d_2)$};},
	I2/.pic={%
		\pic{r-0};
		\pic[below=of r,xshift=-1.1cm]{cl-1};
		\pic[trigger][right=of cl]{cr-2};
		\path (r) edge (cl)
		(r) edge (cr);
	},
	crl-3/.pic={\node (crl) {$U(c,d,d_3)$};},
	I3/.pic={%
		\pic{r-0};
		\pic[below=of r,xshift=-1.1cm]{cl-1};
		\pic[right=of cl]{cr-2};
		\pic[trigger][below=of cr]{crl-3};
		\path (r) edge (cl)
		(r) edge (cr) 
		(cr) edge (crl);
	},
	crl-4/.pic={\node (crl) {$U(c,d,d_3),P(d)$};},
	I4-/.pic={%
		\pic[guardok]{r-0};
		\pic[below=of r,xshift=-1.1cm]{cl-1};
		\pic[guardok][right=of cl]{cr-2};
		\pic[below=of cr]{crl-4};
		\path (r) edge (cl)
		(r) edge (cr) 
		(cr) edge (crl);
	},
	r-4/.pic={\node (r) {$R(c,d),P(d)$};},
	cr-4/.pic={\node (cr) {$T(c,d,d_2),P(d)$};},
	I4/.pic={%
		\pic{r-4};
		\pic[below=of r,xshift=-1.5cm]{cl-1};
		\pic[trigger][right=of cl]{cr-4};
		\pic[below=of cr]{crl-4};
		\path (r) edge (cl)
		(r) edge (cr) 
		(cr) edge (crl);
	},
	cr-5/.pic={\node (cr) {$T(c,d,d_2),P(d),M(c)$};},
	I5-/.pic={%
		\pic[guardok]{r-4};
		\pic[guardok][below=of r,xshift=-1.5cm]{cl-1};
		\pic[right=of cl]{cr-5};
		\pic[guardok][below=of cr]{crl-4};
		\path (r) edge (cl)
		(r) edge (cr) 
		(cr) edge (crl);
	},
	r-5/.pic={\node (r) {$R(c,d),P(d),M(c)$};},
	crl-5/.pic={\node (crl) {$U(c,d,d_3),P(d),M(c)$};},
	cl-5/.pic={\node (cl) {$S(c,d_1),M(c)$};},
	I5/.pic={%
		\pic{r-5};
		\pic[trigger,below=of r,xshift=-2cm]{cl-5};
		\pic[right=of cl]{cr-5};
		\pic[below=of cr]{crl-5};
		\path (r) edge (cl)
		(r) edge (cr) 
		(cr) edge (crl);
	},
	cll-6-/.pic={\node (cll) {$N(c,d_4)$};},
	I6-/.pic={%
		\pic{r-5};
		\pic[below=of r,xshift=-2cm]{cl-5};
		\pic[guardok,below=of cl]{cll-6-};
		\pic[right=of cl]{cr-5};
		\pic[below=of cr]{crl-5};
		\path (r) edge (cl)
		(r) edge (cr) 
		(cl) edge (cll)
		(cr) edge (crl);
	},
	cll-6/.pic={\node (cll) {$N(c,d_4),M(c)$};},
	I6/.pic={%
		\pic{r-5};
		\pic[below=of r,xshift=-2cm]{cl-5};
		\pic[below=of cl]{cll-6};
		\pic[right=of cl]{cr-5};
		\pic[below=of cr]{crl-5};
		\path (r) edge (cl)
		(r) edge (cr) 
		(cl) edge (cll)
		(cr) edge (crl);
	},
}
\node[scope] (I0) {
	\begin{tikzpicture}
		\pic{I0};
		\node[lnode,above=0.1cm of r]{$I_0$};
	\end{tikzpicture}
};
\node[scope,right=of I0] (I1) {
	\begin{tikzpicture}
		\pic{I1};
		\node[lnode,above=0.1cm of r]{$I_1$};
	\end{tikzpicture}
};
\node[scope,right=of I1] (I2) {
	\begin{tikzpicture}
		\pic{I2};
		\node[lnode,above=0.1cm of r]{$I_2$};
	\end{tikzpicture}
};
\node[scope,right=of I2] (I3) {
	\begin{tikzpicture}
		\pic{I3};
		\node[lnode,left=of r]{$I_3$};
	\end{tikzpicture}
};
\pic{drawv=I0/I1/I3/I3};
\pic{drawv=I1/I2/I3/I3};
\pic{drawvr=I2/I3};
\node[scope,right=of I3] (I4-) {
	\begin{tikzpicture}
		\pic{I4-};
	\end{tikzpicture}
};
\pic{drawv=I3/I4-/I3/I3};
\node[scope,below=1.7cm of I0,xshift=1.5cm] (I4) {
	\begin{tikzpicture}
		\pic{I4};
		\node[lnode,left=of r]{$I_4$};
	\end{tikzpicture}
};
\node[scope,right=of I4] (I5-) {
	\begin{tikzpicture}
		\pic{I5-};
	\end{tikzpicture}
};
\pic{drawv=I4/I5-/I4/I4};
\node[scope,right=of I5-] (I5) {
	\begin{tikzpicture}
		\pic{I5};
		\node[lnode,left=of r] {$I_5$};
	\end{tikzpicture}
};
\pic{drawv=I5-/I5/I4/I4};
\node[scope,below=0.8cm of I4,xshift=3cm] (I6-) {
	\begin{tikzpicture}
		\pic{I6-};
	\end{tikzpicture}
};
\node[scope,right=of I6-] (I6) {
	\begin{tikzpicture}
		\pic{I6};
		\node[lnode,left=of r] {$I_6$};
	\end{tikzpicture}
};
\pic{drawvr=I6-/I6};
\pic{drawh=I3/I4/I4/I5};
\pic{drawh=I4/I6-/I4/I5};
\end{tikzpicture}

%% file: pages/prelims/proof_machinery/disjunctive_chase.tex
\subsection{The Disjunctive Chase}\label{sec:disjunctive_chase}
The preceding discussion extends to disjunctive \tgds.
The analogue of the chase for \dtgds is the \emph{\iindex{disjunctive chase}} \citep{deutsch2008chase}.
As before, we have the possibility to consider a restricted and oblivious version,
but since both versions can again be used equivalently,
we only define the latter.

Let $I$ be a set of facts, $\tau=\body\rightarrow \bigvee_{i=1}^n\exists\vec y_i\,\head_i$ be a \dtgd,
and $h$ be a trigger for $\tau$ in~$I$.
A \emph{\iindex{disjunctive chase step}} for $\tau$ and $h$ to~$I$
creates instances $I_1 \dotso I_n$, one for each head conjunction of $\tau$,
such that instance $I_i$ is formed by applying an ordinary chase step for $\body \rightarrow\exists\vec y_i\,\head_i$
and $h$ to~$I$.

For $I_0$ an instance and $\Sigma$ a set of \dtgds,
a \emph{\iindex{chase tree} for $I$ and $\Sigma$} 
is a possibly infinite tree $T$ with root $I_0$ such that for every 
node $I$, the children of $I$
are obtained by applying a disjunctive chase step for some $\tau\in\Sigma$ to $I$.
For any chase tree $T$, we let \iindex[not]{$\lvs(T)$} denote the set of leaves of $T$.
A path $I_0,\dotsc,L$ with $L\in\lvs(T)$ is called a \emph{\iindex{thread}} of $T$.
Note that every thread $I_0,\dotsc,I_n$ can be considered as a chase sequence.

For a \ubcq $Q$ over $\db$, a finite chase tree $T$ for $\db$ and $\Sigma$ 
is a \index{chase proof!DisTGD}\emph{chase proof from $\db$ and~$\Sigma$ of~$Q$} 
if every thread $I_0,\dotsc,I_n$ is a chase proof of $Q$.
We write $T\models Q$ if $T$ is a chase proof of $Q$.
The following result shows that the disjunctive chase is a complete proof system:

\begin{thm}[\cite{deutsch2008chase}]\label{thm:dis_chase_univ} 
A \ubcq $Q$ follows from a database $\db$ and a set of \dtgds $\Sigma$ 
\ifftext there is a chase proof of $Q$ from $\db$ and $\Sigma$.
\end{thm}
The idea of the proof is the same as for \cref{thm:chase_univ}:
one shows that an infinite chase tree has a homomorphism to every possible world for $\db,\Sigma$
and makes use of the fact that homomorphisms preserve \ubcqs.

\begin{exmpl}\label{exmpl:disj_chase}
Given the database $\db=\{R(c,d)\}$, a set of \dgtgds $\Sigma$
\begin{align*}
R(x_1,x_2)&\rightarrow\exists y\,S(x_1,y)\lor\exists y\,T(x_1,x_2,y),&
T(x_1,x_2,x_3)\rightarrow\exists y\,U(x_1,x_2,y),\\
U(x_1,x_2,x_3)&\rightarrow M(x_1,x_2)\lor P(x_2),&
T(x_1,x_2,x_3)\land P(x_2)\rightarrow M(x_1,x_1),\\
S(x_1,x_2)&\rightarrow M(x_1,x_1),
\end{align*} 
and the \bcq $Q=\exists y\,M(c,y)$,
the tree depicted in \cref{fig:exmpl_chase_tree} is a chase proof from~$\db$ and $\Sigma$ of~$Q$ as every leaf
models $Q$.
\begin{figure}[ht]
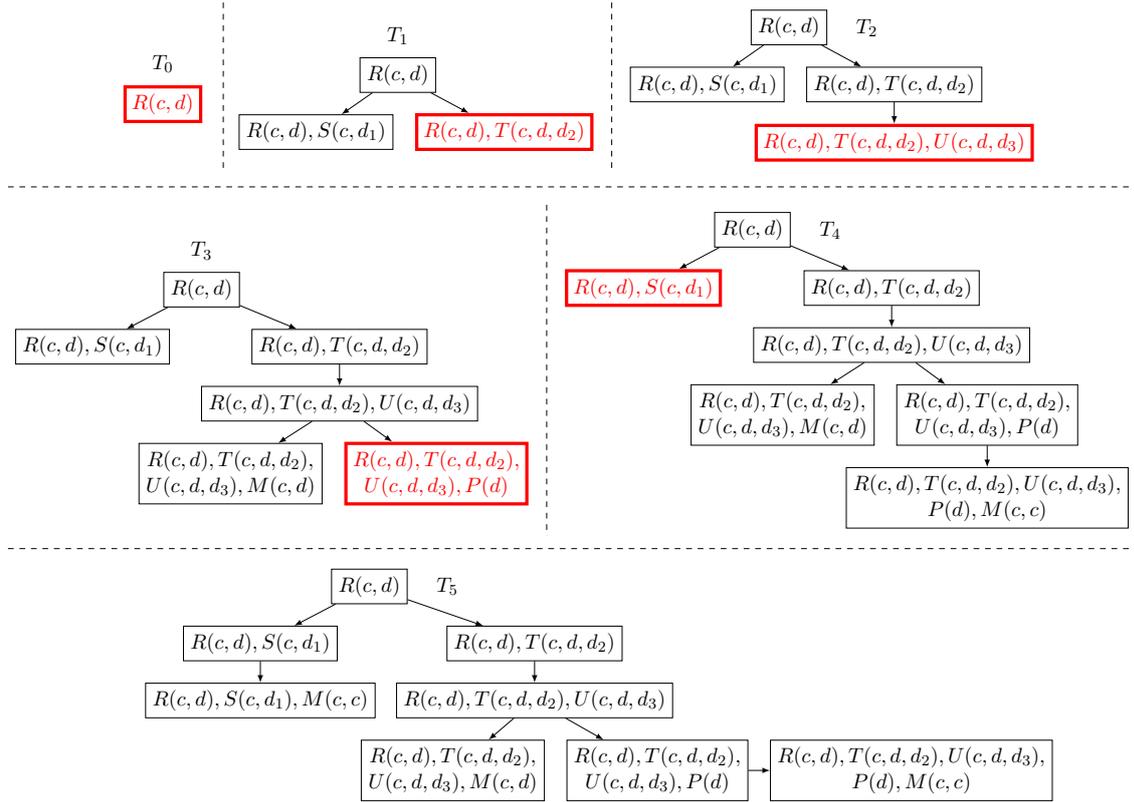

	\centering
	\includestandalone[width=\textwidth]{./figures/exmpl_chase_tree}
	\caption{Chase tree for \cref{exmpl:disj_chase};
	instances that will be extended are marked bold and red.}
	\label{fig:exmpl_chase_tree}
\end{figure}
\end{exmpl} 

As we did for \gtgds, we can again obtain a tree-like structure property for chase trees.
\begin{prop}\label{prop:dis_chase_tree}
Given a chase tree $T$ for a set of \dgtgds $\Sigma$,
for any thread $I_0,\dotsc,I_n$ of $T$, 
each instance $I_i$ can be decomposed into a tree $T_i=(V_i,E_i)$ with root $r$ such that $\consts(F_r)=\consts(I_0)$
and for any $v\neq r$, $\consts(F_v)\subseteq h(\vec x,\vec y_k)$ for a trigger $h$ of
some rule $\body(\vec x)\rightarrow\bigvee_{i=1}^m\vec y_i\,\head_i(\vec x,\vec y)$ fired in $I_j$ for some $j<i$ and $k\in\{1,\dotsc,m\}$.
\end{prop}
\begin{proof}
Apply \cref{prop:tree_chase} to every thread of $T$.
\end{proof} 
We will call such a chase tree \emph{tree-like}; that is,
a chase tree is tree-like if each thread in it is tree-like.
Given such a tree-like chase tree $T$,
we will make use of a few additional definitions
in order to conveniently modify and talk about the structure of $T$:
\begin{itemize}
\item As $T$ is created by a sequence of chase steps, we can 
refer to the chase tree created by only firing the first $i$ chase steps of $T$
by \iindex[not]{$T_i$}.

\item Given a node $N$ of $T$,
we write \iindex[not]{$T^N$} for the subtree of $T$ rooted at $N$.
Note that $T^N$ is a chase tree with root $N$.

\item Given a node $N$ of $T$ and a node $v$ in $N$,
we write \iindex[not]{$N_v$} for the bag of~$v$ in $N$. 
This definition is lifted to sets of nodes $S$ by setting ${S_v\coloneqq\{N_v\mid N\in S\}}$.

\item Given a set of nodes $S$ of $T$, 
we write \iindex[not]{$\cut(T,S)$} for the chase tree obtained from $T$
by cutting off the children of all nodes in~$S$.
If $S$ is a singleton $\{N\}$, we also write $\cut(T,N)\coloneqq\cut(T,\{N\})$.
\end{itemize}

%% file: figures/exmpl_chase_tree.tex
\begin{tikzpicture}
\tikzset{node distance=0.4cm,
	r-0/.pic={\node (r) {$R(c,d)$};},
	T0/.pic={\pic[trigger]{r-0};},
	Nl/.pic={\node (Nl) {$R(c,d),S(c,d_1)$};},
	Nr/.pic={\node (Nr) {$R(c,d),T(c,d,d_2)$};},
	T1/.pic={%
		\pic{r-0};
		\pic[below=of r,xshift=-1.5cm]{Nl};
		\pic[trigger,right=of Nl]{Nr};
		\path (r) edge (Nl)
		(r) edge (Nr);
	},
	Nrm/.pic={\node (Nrm) {$R(c,d),T(c,d,d_2),U(c,d,d_3)$};},
	T2/.pic={%
		\pic{r-0};
		\pic[below=of r,xshift=-1.5cm]{Nl};
		\pic[right=of Nl]{Nr};
		\pic[trigger,below=of Nr]{Nrm};
		\path (r) edge (Nl)
		(r) edge (Nr)
		(Nr) edge (Nrm);
	},
	Nrml/.pic={\node[align=center] (Nrml) {$R(c,d),T(c,d,d_2),$\\$U(c,d,d_3),M(c,d)$};},
	Nrmr/.pic={\node[align=center] (Nrmr) {$R(c,d),T(c,d,d_2),$\\$U(c,d,d_3),P(d)$};},
	T3/.pic={%
		\pic{r-0};
		\pic[below=of r,xshift=-2cm]{Nl};
		\pic[right=1.5cm of Nl]{Nr};
		\pic[below=of Nr]{Nrm};
		\pic[below=of Nrm,xshift=-2cm]{Nrml};
		\pic[trigger,right=of Nrml]{Nrmr};
		\path (r) edge (Nl)
		(r) edge (Nr)
		(Nr) edge (Nrm)
		(Nrm) edge (Nrml)
		(Nrm) edge (Nrmr);
	},
	Nrmrm/.pic={\node[align=center] (Nrmrm) {$R(c,d),T(c,d,d_2),U(c,d,d_3),$\\$P(d),M(c,c)$};},
	T4/.pic={%
		\pic{r-0};
		\pic[trigger,below=of r,xshift=-2cm]{Nl};
		\pic[right=1.5cm of Nl]{Nr};
		\pic[below=of Nr]{Nrm};
		\pic[below=of Nrm,xshift=-2cm]{Nrml};
		\pic[right=of Nrml]{Nrmr};
		\pic[below=of Nrmr]{Nrmrm};
		\path (r) edge (Nl)
		(r) edge (Nr)
		(Nr) edge (Nrm)
		(Nrm) edge (Nrml)
		(Nrm) edge (Nrmr)
		(Nrmr) edge (Nrmrm);
	},
	Nlm/.pic={\node (Nlm) {$R(c,d),S(c,d_1),M(c,c)$};},
	T5/.pic={%
		\pic{r-0};
		\pic[below=of r,xshift=-2cm]{Nl};
		\pic[below=of Nl]{Nlm};
		\pic[right=2cm of Nl]{Nr};
		\pic[below=of Nr]{Nrm};
		\pic[below=of Nrm,xshift=-1.5cm]{Nrml};
		\pic[right=of Nrml]{Nrmr};
		\pic[right=of Nrmr]{Nrmrm};
		\path (r) edge (Nl)
		(r) edge (Nr)
		(Nl) edge (Nlm)
		(Nr) edge (Nrm)
		(Nrm) edge (Nrml)
		(Nrm) edge (Nrmr)
		(Nrmr) edge (Nrmrm);
	},
}
\node[scope] (T0) {
	\begin{tikzpicture}
		\pic{T0};
		\node[lnode,above=0.1cm of r]{$T_0$};
	\end{tikzpicture}
};
\node[scope,right=of T0] (T1) {
	\begin{tikzpicture}
		\pic{T1};
		\node[lnode,above=0.1cm of r]{$T_1$};
	\end{tikzpicture}
};
\node[scope,right=of T1] (T2) {
	\begin{tikzpicture}
		\pic{T2};
		\node[lnode,right=of r]{$T_2$};
	\end{tikzpicture}
};
\node[scope,below=1.8cm of T0,xshift=2cm] (T3) {
	\begin{tikzpicture}
		\pic{T3};
		\node[lnode,above=0.1cm of r]{$T_3$};
	\end{tikzpicture}
};
\node[scope,right=of T3] (T4) {
	\begin{tikzpicture}
		\pic{T4};
		\node[lnode,right=of r] {$T_4$};
	\end{tikzpicture}
};
\node[scope,below=0.9cm of T3,xshift=6cm] (T5) {
	\begin{tikzpicture}
		\pic{T5};
		\node[lnode,right=of r] {$T_5$};
	\end{tikzpicture}
};
\pic{drawv=T0/I1/T2/T2};
\pic{drawvr=T1/T2};
\pic{drawvr=T3/T4};
\pic{drawh=T2/T4/T3/T4};
\pic{drawh=T4/T5/T3/T4};
\end{tikzpicture}

%% file: chapters/one_pass.tex
\chapter{The One-Pass Property}\label{chap:one_pass}

The tree-like chases presented so far are experiencing a chaotic
behaviour in the sense that triggers can potentially fire in any node of
the (internal) tree at any time.
It will be useful for our upcoming completeness proofs 
to tame this behaviour by further normalising the chase proofs.
This idea comes from \cite{amarilli2018can}, although the proofs that
are being normalised there are more specialised than general (disjunctive) \gtgds.
We will do this first for chase sequences and then extend to chase trees.

\input{pages/one_pass/one_pass_chases}
\input{pages/one_pass/one_pass_disjunctive_chases}

%% file: pages/one_pass/one_pass_chases.tex
\section{One-Pass Chases}

Let us start with the formal definition of our desired one-pass property.

\begin{defn}[One-pass chase sequence]\label{defn:one_pass_chase}
We say that a tree-like chase sequence $I_0,\dotsc,I_n$ is \index{one-pass chase sequence}\emph{one-pass}
if for every node~$v$ present at some stage
$I_i$, if at some stage $j>i$ we perform a chase step
triggered in a node outside of the subtree of $v$,
then for all~$k$ with $j\leq k<n$, the $\nth{k}$ step does not modify the subtree of $v$.
\end{defn} 
\begin{exmpl}\label{exmpl:one_pass_chase_sequence}
The chase depicted in \cref{fig:exmpl_tree_chase_seq} is \emph{not} one-pass.
From step four to five, we propagate a fact not only to ancestors that have a 
suitable guard, but also to a child and a sibling that have a guard for it. 
In the next step, we make crucial use of this propagated fact in the sibling.

In order to obtain a one-pass chase, instead of propagating facts to
all nodes in the tree, we can first only propagate facts to ancestors and 
then recursively re-create all subtrees, starting from
the node that triggered the chase step up to the root.
This idea is depicted in \cref{fig:exmpl_one_pass_chase_seq},
and it is the key idea of the proof of \cref{prop:gtgd_one_pass}.
\begin{figure}[ht]
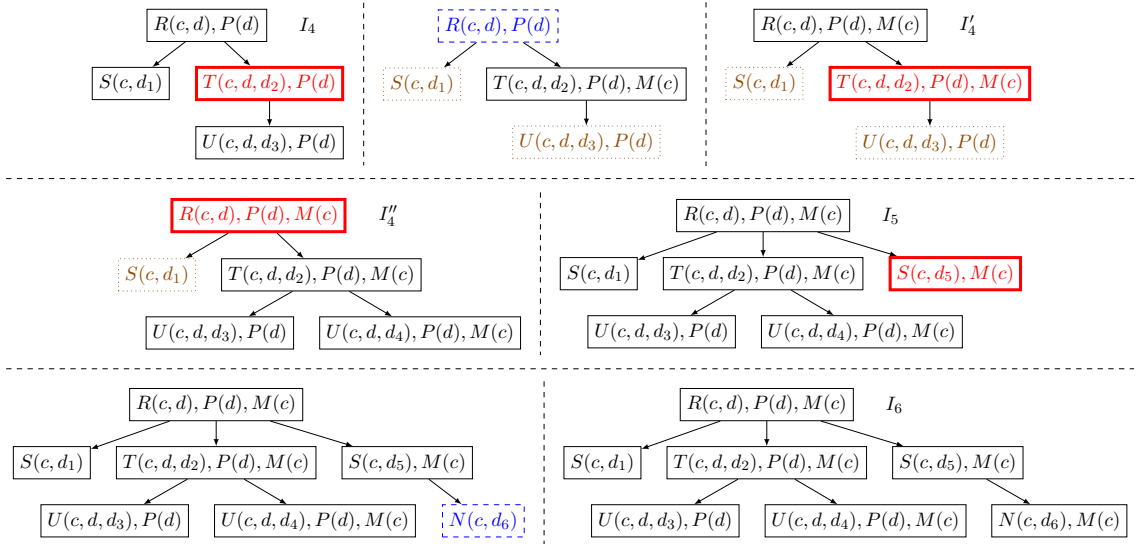

	\centering
	\includestandalone[width=\textwidth]{./figures/exmpl_one_pass_chase_seq}
	\caption{We can modify the chase from \cref{fig:exmpl_tree_chase_seq} to obtain a one-pass chase.
	Triggered nodes are marked bold and red, and nodes with missing, non-propagated facts are marked dashed and blue
	if one-pass propagation is allowed, and
	dotted and brown if one-pass propagation is prohibited and a new copy has not been created so far.}
	\label{fig:exmpl_one_pass_chase_seq}
\end{figure}
\end{exmpl}

Recall that in a tree-like chase sequence, we have a sequence
of trees, where each node is associated with a set of facts and
each fact is inferred from facts in prior trees.
The usefulness of a one-pass chase is that
the inferences we make are more closely related to the tree structure.
This thought is captured in \cref{lem:one_pass_proof}.
To prove the lemma, we first give the following definition.
\begin{defn}
Let $v$ be a node in a tree $T$. We say that a node $w$ is a
\index{non-strict ancestor (descendant)}\emph{non-strict ancestor of $v$} if $v=w$ or $w$
is an ancestor of $v$ in $T$.
Similarly, we say that $w$ is a \emph{non-strict descendant of $v$} if $v=w$ or $w$
is an descendant of $v$ in $T$.
\end{defn} 
\begin{lem}\label{lem:one_pass_proof}
Suppose $I_0,\dotsc,I_n$ is a one-pass chase sequence 
and~$v$ is a node created in $I_i$.
Suppose $d$ is a non-strict descendant of $v$ in $I_j$ with $j\geq i$,
and~$d$ adds on facts in~$I_j$.
Then we have a one-pass chase proof of $F_d^j$ from $F_v^i$ using steps from $i,\dotsc,{j-1}$.
\end{lem}
\begin{proof}
We prove the claim by induction on $j-i$. 
If $i=j$, we have $F_d^j=F_v^i$, and hence we do not need any steps to create $F_d^j$ from $F_v^i$.

Now assume $j>i$. 
Let $w$ be the node that was used to generate the new facts $F$.
Since the chase is one-pass, $w$ must be a non-strict descendant of $d$,
or $w$ is the parent that generated $d$ at step $j$;
in any case, $w$ is a non-strict descendant of $v$.
Let $k<j$ be the last time a fact was added to $w$.
By the inductive hypothesis, we have a one-pass chase proof 
of $F_w^k=F_w^{j-1}$ from $F_v^i$ using steps from $i,\dotsc,{k-1}$.
Further, we have a chase step from $F_w^{j-1}$ to $F$.
The claim now follows by transitivity.
\end{proof} 

\begin{cor}\label{cor:one_pass_proof}
Suppose $I_0,\dotsc, I_n$ is a one-pass chase sequence,
$v$ is a node in some~$I_i$ that was created in $I_k$, 
and $d$ is non-strict descendant of $v$ in $I_j$ with $j \geq i$.
Then we have a one-pass chase proof of $F_d^j$ from $F_v^i$ using steps from $k,\dotsc,j-1$.
\end{cor}
\begin{proof}
Let $l\leq j$ be the last time a fact was added to $d$.
By \cref{lem:one_pass_proof}, we have a one-pass chase proof of $F_d^l=F_d^j$ from $F_v^k$ 
using steps from $k,\dotsc,l-1$.
Since $F_v^k\subseteq F_v^i$ and $l\leq j$, we also have a one-pass chase proof of $F_d^j$ from $F_v^i$ 
using steps from $k,\dotsc,j-1$.
\end{proof} 

As a next step, we want to show that we can indeed
assume the existence of a one-pass chase proof whenever we have a chase
proof for some \ubcq $Q$.
We achieve this by constructing a suitable one-pass chase sequence
from a given arbitrary chase sequence as sketched in \cref{exmpl:one_pass_chase_sequence}.
\begin{prop}\label{prop:gtgd_one_pass} 
For every tree-like chase sequence $I_0,\dotsc,I_n$,
there is a one-pass chase sequence $\overline{I_0}({=}I_0),\overline{I_1},\dotsc,\overline{I_m}$ 
such that there is a homomorphism ${h:I_n\rightarrow \overline{I_m}}$ 
with $h(c)=c$ for any $c\in\consts(I_0)$.
\end{prop}
\begin{proof}
We write $F_v^i$ for the set of facts of some node $v$ in $I_i$ and
$\overline{F_{\overline{v}}^i}$ for the set of facts of some node 
$\overline{v}$ in $\overline{I_i}$.
We also show that $h$ preserves the tree structure of $I_n$;
that is, for any subtree $T$ in $I_n$, $h(T)$ is a subtree 
in $\overline{I_m}$ graph-isomorphic to $T$.
This allows us to write $\overline{v}\coloneqq h(v)$ to refer to the node 
$\overline{v}$ in $\overline{I_m}$ with $\overline{F_{\overline{v}}^m}=h(F_v^n)$
for any node $v$ in $I_n$.
We prove the claim by induction on $n$. 
For $n=0$, we simply set $\overline{I_0}\coloneqq I_0$ and let
$h$ be the identity.

Assume $n>0$. Then by the inductive hypothesis, we get a one-pass chase sequence $\overline{I_0}({=}I_0),\overline{I_1},\dotsc,\overline{I_m}$
such that $I_{n-1}$ is homomorphic to $\overline{I_m}$ via $h$.
Assume the chase step from $I_{n-1}$ to $I_n$ fires in $v$.
First, we use the inductive one-pass chase to generate a copy of the subtree of $v$ (\wrt $I_{n-1}$) 
starting from the root of $\overline{I_m}$;
more precisely, we re-fire all rules used to create the subtree of $\overline{v}$
but use fresh constants $c'$ for any~$c\in\consts(\overline{I_m})\setminus\consts(I_0)$ 
in this chase while any constant in $\consts(I_0)$ stays unchanged.
Call the copy of $v$ in this subtree~$\overline{v}'$.

We then fire the chase step in $\overline{v}'$ but only propagate new facts to ancestors.
Since some newly generated facts might have propagated to non-ancestor nodes $w$ of $v$,
we also have to create new copies of these non-ancestor nodes in our one-pass chase.
Further, we have to update $h$ to preserve the tree structure of $I_n$.
We achieve this by modifying the subtree of any non-strict 
ancestor of $\overline{v}'$ in the following way:

Beginning from $p_0\coloneqq v$, we will create copies
of all facts occurring in the subtree of $p_i$ in $I_n$ that are still missing in our one-pass chase.
For this, let $\overline{I_0},\dotsc,\overline{I_{m^*}}$ be the one-pass chase created so far,
$\overline{p_i}$ be the copy of $p_i$ in $\overline{I_m}$,
and $\overline{p_i}'$ be the non-strict ancestor of $\overline{v}'$ corresponding to $\overline{p_i}$.
By the inductive hypothesis and \cref{cor:one_pass_proof}, 
we have a one-pass chase proof of the subtree of $\overline{p_i}$ from 
$\overline{F_{\overline{p_i}}^-}$ (the initial set of facts of $\overline{p_i}$).
As $\overline{F_{\overline{p_i}}^-}$ is homomorphic to $\overline{F_{\overline{p_i}'}^{m^*}}$,
we can also fire the steps of this chase proof starting from $\overline{p_i}'$ in $\overline{I_{m^*}}$
to create all missing facts in the subtree of $\overline{p_i}'$
(while again using fresh constants $c'$ for any~$c\in\consts(\overline{I_{m^*}})\setminus\consts(I_0)$).
We then let $p_{i+1}$ be the parent of $p_i$ and proceed with $p_{i+1}$.

Finally, for any node $u$ in $I_n$, 
we let $h$ map to the newest copy of $u$ in $\overline{I_{m^*}}$.
As all newly generated facts were propagated to all ancestors of $\overline{v}'$,
$h$ preserved the tree structure of $I_{n-1}$, 
and we re-created the missing facts of all subtrees, 
$h(u)$ will be a copy of $u$ in $\overline{I_{m^*}}$.
Further, $h$ preserves the tree structure of $I_n$ as we re-created all subtrees.
\end{proof}

We have to show that our construction indeed preserves answers for \bcqs.
\begin{lem}\label{lem:hom_chase_seq_impl_bcq}
Let $I_0,\dotsc,I_n$ and $I_0'({=}I_0),\dotsc,I_m'$ be two chase sequences
such that there is a homomorphism ${h:I_n\rightarrow I_m'}$
with $h(c)=c$ for any $c\in\consts(I_0)$.
Then, for any \bcq~$Q$, if $I_0,\dotsc,I_n$ is a chase proof of $Q$, then
$I_0',\dotsc,I_m'$ is a chase proof of $Q$.
\end{lem} 
\begin{proof}
We show a stronger claim:
for any \cq $Q(\vec x)$ and $\vec c\subseteq\consts(I_n)$,
$I_n\models Q(\vec c)$ implies $I_m'\models h(Q(\vec c))$.
Since $h$ is the identity on $\consts(I_0)$, this implies our goal.
We show the claim by induction on the number of existential quantifiers of $Q$, denoted by $k$.

If $k=0$, using that $h:I_n\rightarrow I_m'$ is a homomorphism, we get
\begin{align*}
I_n\models Q(\vec c)\iff&I_n\models A(\vec c) \text{ for any atom $A$ in }Q\\
\implies &I_m'\models h(A(\vec c))\text{ for any atom $A$ in }Q\\
\iff&I_m'\models h(Q(\vec c)).
\end{align*} 

Assume $k>0$. That is $Q(\vec x)\equiv\exists y\,Q'(\vec x,y)$ with $Q'$ a \cq. 
If $I_n\models Q(\vec c)$, then there is $c'\in\consts(I_n)$ such that 
$I_n\models Q'(\vec c,c')$.
Thus, by the inductive hypothesis, $I_m'\models h(Q'(\vec c,c'))$,
and hence $I_m'\models h(Q(\vec c))$.
\end{proof} 

Finally, we derive our desired goal:
\begin{thm}\label{thm:gtgd_iff_one_pass}
A \ubcq $Q$ has a chase proof from $\db$ and \gtgds $\Sigma$
\ifftext it has a one-pass chase proof from $\db$ and $\Sigma$.
\end{thm} 
\begin{proof}
If we have one-pass chase proof of $Q$, we also have a chase proof of $Q$
since every one-pass chase proof is a chase proof.

Assume we have a chase proof $I_0,\dotsc,I_n$ of $Q$
and let $\overline{I_0},\dotsc,\overline{I_m}$ be the corresponding one-pass
chase sequence as created in \cref{prop:gtgd_one_pass}.
Then $I_n\models Q_i$ for some $Q_i\in Q$. 
Hence, by \cref{lem:hom_chase_seq_impl_bcq}, $\overline{I_m}\models Q_i$,
and thus $\overline{I_m}\models Q$.
\end{proof} 

%% file: figures/exmpl_one_pass_chase_seq.tex
\begin{tikzpicture}
\tikzset{node distance=0.5cm,
	cl-1/.pic={\node (cl) {$S(c,d_1)$};},
	crl-4/.pic={\node (crl) {$U(c,d,d_3),P(d)$};},
	r-4/.pic={\node (r) {$R(c,d),P(d)$};},
	cr-4/.pic={\node (cr) {$T(c,d,d_2),P(d)$};},
	I4/.pic={%
		\pic{r-4};
		\pic[below=of r,xshift=-1.5cm]{cl-1};
		\pic[trigger][right=of cl]{cr-4};
		\pic[below=of cr]{crl-4};
		\path (r) edge (cl)
		(r) edge (cr) 
		(cr) edge (crl);
	},
	cr-5/.pic={\node (cr) {$T(c,d,d_2),P(d),M(c)$};},
	I4'-1/.pic={%
		\pic[guardok]{r-4};
		\pic[guardnok][below=of r,xshift=-1.5cm]{cl-1};
		\pic[right=of cl]{cr-5};
		\pic[guardnok][below=of cr]{crl-4};
		\path (r) edge (cl)
		(r) edge (cr) 
		(cr) edge (crl);
	},
	r-5/.pic={\node (r) {$R(c,d),P(d),M(c)$};},
	I4'/.pic={%
		\pic{r-5};
		\pic[guardnok][below=of r,xshift=-1.5cm]{cl-1};
		\pic[trigger][right=of cl]{cr-5};
		\pic[guardnok][below=of cr]{crl-4};
		\path (r) edge (cl)
		(r) edge (cr) 
		(cr) edge (crl);
	},
	crr-5/.pic={\node (crr) {$U(c,d,d_4),P(d),M(c)$};},
	I4''/.pic={%
		\pic[trigger]{r-5};
		\pic[guardnok][below=of r,xshift=-2cm]{cl-1};
		\pic[right=of cl]{cr-5};
		\pic[below=of cr,xshift=-2cm]{crl-4};
		\pic[right=of crl]{crr-5};
		\path (r) edge (cl)
		(r) edge (cr) 
		(cr) edge (crl)
		(cr) edge (crr);
	},
	cnr-5/.pic={\node (cnr) {$S(c,d_5),M(c)$};},
	I5/.pic={%
		\pic{r-5};
		\pic[below=of r]{cr-5};
		\pic[left=of cr]{cl-1};
		\pic[below=of cr,xshift=-2cm]{crl-4};
		\pic[right=of crl]{crr-5};
		\pic[trigger][right=of cr]{cnr-5};
		\path (r) edge (cl)
		(r) edge (cr) 
		(cr) edge (crl)
		(cr) edge (crr)
		(r) edge (cnr);
	},
	cnrl-6-/.pic={\node (cnrl) {$N(c,d_6)$};},
	I6-/.pic={%
		\pic{r-5};
		\pic[below=of r,xshift=-0cm]{cr-5};
		\pic[left=of cr]{cl-1};
		\pic[below=of cr,xshift=-2cm]{crl-4};
		\pic[right=of crl]{crr-5};
		\pic[right=of cr]{cnr-5};
		\pic[guardok,below=of cnr,xshift=1.5cm]{cnrl-6-};
		\path (r) edge (cl)
		(r) edge (cr) 
		(cr) edge (crl)
		(cr) edge (crr)
		(r) edge (cnr)
		(cnr) edge (cnrl);
	},
	cnrl-6/.pic={\node (cnrl) {$N(c,d_6),M(c)$};},
	I6/.pic={%
		\pic{r-5};
		\pic[below=of r,xshift=-0cm]{cr-5};
		\pic[left=of cr]{cl-1};
		\pic[below=of cr,xshift=-2cm]{crl-4};
		\pic[right=of crl]{crr-5};
		\pic[right=of cr]{cnr-5};
		\pic[below=of cnr,xshift=2cm]{cnrl-6};
		\path (r) edge (cl)
		(r) edge (cr) 
		(cr) edge (crl)
		(cr) edge (crr)
		(r) edge (cnr)
		(cnr) edge (cnrl);
	},
}
\node[scope] (I4) {
	\begin{tikzpicture}
		\pic{I4};
		\node[lnode,right=of r] {$I_4$};
	\end{tikzpicture}
};
\node[scope,right=of I4] (I4'-1) {
	\begin{tikzpicture}
		\pic{I4'-1};
	\end{tikzpicture}
};
\pic{drawvr=I4/I4'-1};
\node[scope,right=of I4'-1] (I4') {
	\begin{tikzpicture}
		\pic{I4'};
		\node[lnode,right=of r] {$I_4'$};
	\end{tikzpicture}
};
\pic{drawvr=I4'-1/I4'};
\node[scope,below=of I4,xshift=2cm] (I4'') {
	\begin{tikzpicture}
		\pic{I4''};
		\node[lnode,right=of r] {$I_4''$};
	\end{tikzpicture}
};
\node[scope,right=of I4''] (I5) {
	\begin{tikzpicture}
		\pic{I5};
		\node[lnode,right=of r] {$I_5$};
	\end{tikzpicture}
};
\pic{drawvr=I4''/I5};
\node[scope,below=of I4'',xshift=-1cm] (I6-) {
	\begin{tikzpicture}
		\pic{I6-};
	\end{tikzpicture}
};
\node[scope,right=of I6-] (I6) {
	\begin{tikzpicture}
		\pic{I6};
		\node[lnode,right=of r] {$I_6$};
	\end{tikzpicture}
};
\pic{drawvr=I6-/I6};
\pic{drawh=I4''/I6-/I6-/I6};
\pic{drawh=I4/I4''/I6-/I6};
\end{tikzpicture}

%% file: pages/one_pass/one_pass_disjunctive_chases.tex
\section{One-Pass Disjunctive Chases}

Next, we generalise our idea to disjunctive chases. 
We proceed as in the previous section.
\begin{defn}[One-pass chase tree]\label{defn:one_pass_chase_tree}
We say that a chase tree is \index{one-pass chase tree}\emph{one-pass}
if each thread in it is a one-pass chase sequence.
\end{defn} 

\begin{exmpl}
The tree-like chase tree for \cref{exmpl:disj_chase}, depicted in \cref{fig:one_pass_chase_tree}, is \emph{not} one-pass.
From step three to four, we propagate the fact $M(c,c)$ to a child that has a suitable guard.

\begin{figure}[ht]
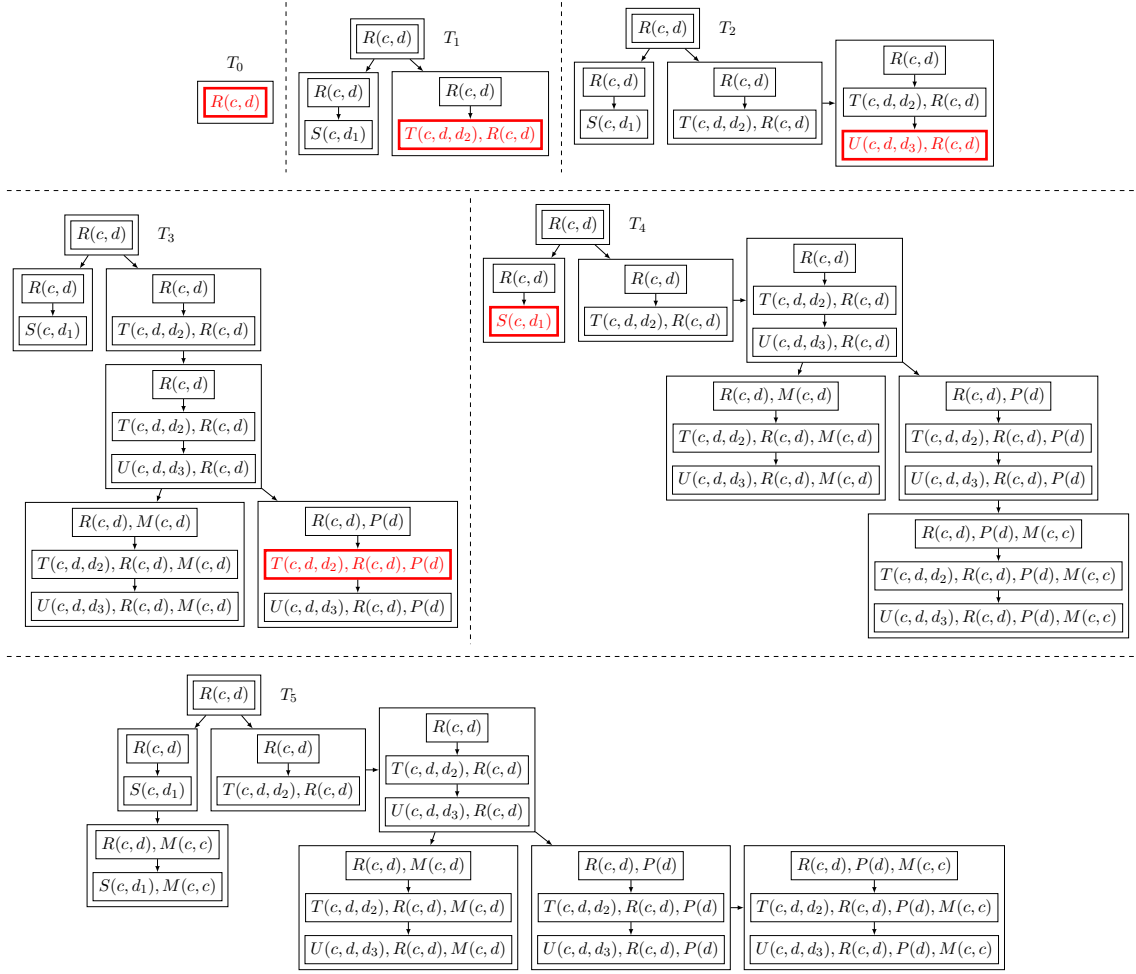

	\centering
	\includestandalone[width=\textwidth]{./figures/exmpl_tree_like_chase_tree}
	\caption{Tree-like chase tree corresponding to  \cref{fig:exmpl_chase_tree};
	triggered nodes are marked bold and red.}
	\label{fig:one_pass_chase_tree}
\end{figure}
\end{exmpl} 

Again, we want to construct a suitable one-pass chase tree $\overline{T}$
from a given arbitrary chase tree $T$.
Since a chase tree models some \ubcq $Q$ \ifftext every leaf
of the tree models~$Q$, we want to construct $\overline{T}$
in a way such that for any leaf $\overline{L}$ of $\overline{T}$,
there is at least one leaf $L$ of~$T$ that is homomorphic to $\overline{L}$.
The basic idea is to lift the proof of \cref{prop:gtgd_one_pass}. That is,
we first only propagate facts to ancestors and 
then recursively re-create all subtrees.
The introduction of disjunctions, however, causes some difficulties: 
whenever we want to create a copy of some subtree in a node of $\overline{T}$,
we might generate additional threads in $\overline{T}$ independent of the given subtree.
We thus have to consider these additional threads in our proof.

\begin{prop}\label{prop:dgtgd_one_pass}
For every tree-like chase tree $T$ with root $I_0$,
there is a one-pass chase tree $\overline{T}$ with root $I_0$ such that
for any leaf $\overline{L}$ in $\overline{T}$,
there is a leaf $L$ in $T$
and a homomorphism $h_{\overline{L}}:L\rightarrow \overline{L}$
with $h(c)=c$ for any $c\in\consts(I_0)$.
\end{prop}
\begin{proof}
As we did for \cref{prop:gtgd_one_pass},
we also show that each $h_{\overline{L}}$ preserves the tree structure of $L$.
$\overline{L}\in\lvs(\overline{T})$.
We prove the claim by induction on the number of steps used to create $T$, denoted by $n$.
For $n=0$, we set $\overline{I_0}\coloneqq I_0$ and let $h$ be the identity.

For $n>0$, by the inductive hypothesis, we get a one-pass chase tree $\overline{T}$
such that for any $\overline{L}\in\lvs(\overline{T})$,
there is some $L\in\lvs(T_{n-1})$ and 
constructed homomorphism $h_{\overline{L}}:L\rightarrow\overline{L}$.
Assume the chase step from $T_{n-1}$ to $T_n$ fires in node $v$ in $L$, 
creating leaves $L_1,\dotsc,L_m$,
and let $S$ be the set of all leaves $\overline{L}\in\lvs(\overline{T})$ 
such that $L$ is homomorphic to $\overline{L}$ via some constructed~$h_{\overline{L}}$.
If $S=\emptyset$, there is nothing to do as each leaf in $\overline{T}$
still has a homomorphism from some leaf in~$T_n$.
Hence, assume $S\neq\emptyset$.

First, we want to create a copy of the subtree of $v$ (\wrt $L$) using the inductive one-pass chase. 
For this, for any $\overline{L}\in S$, 
let $N_{\overline{L}}$ be the last node in the thread $\overline{I_0},\dotsc,\overline{L}$ in which 
the subtree of the node $\overline{v}$ in~$\overline{L}$ corresponding to~$v$ added on some fact.
Let $S_{\overline{v}}$ be the set consisting of all such nodes $N_{\overline{L}}$.
Note that $S_{\overline{v}}\neq\emptyset$ since $S\neq\emptyset$.
We now re-fire from any $\overline{L}\in S$ all rules that created the tree $\cut(\overline{T},S_{\overline{v}})$ 
but use fresh constants $c'$ for any constant~$c\in\consts(\overline{T})\setminus\consts(I_0)$ 
in this chase while any constant in $\consts(I_0)$ stays unchanged.
Let $\overline{T}'$ be the resulting chase tree.

Note that any descendant $\overline{D}$ of some $\overline{L}\in S$ in $\overline{T'}$
is just an extension of $\overline{L}$, 
and hence,~$L$ is still homomorphic to $\overline{D}$ via $h_{\overline{L}}$
(we will make use of this fact in a step below).
Even more, for any new leaf $\overline{L}'$ in $\overline{T}'$, there is either some
$L'\in\lvs(T_n)$ that is homomorphic to $\overline{L}'$ (if $\overline{L}'$ got created
by a branch that was not affected by $S_{\overline{v}}$),
or $\overline{L}'$ contains a fresh copy of the subtree of $v$, 
created by a one-pass proof and eligible for further one-pass chase steps.
Let $S'$ consist of all leaves in $\overline{T}'$ for which the latter is the case. 
Note that $S'\neq\emptyset$ as $S_{\overline{v}}\neq\emptyset$.

For any $\overline{L}'\in S'$, 
we then fire the chase step that extended $T_{n-1}$ to $T_n$ in the
fresh copy~$\overline{v}'$ of $v$,
creating new leaves $\overline{L_1},\dotsc,\overline{L_m}$ that we associate to $L_1,\dotsc,L_m$,
but only propagate new facts to ancestors in each~$\overline{L_i}$.
Let $\overline{T}^*$ be the chase tree created so far.
As we did for \cref{prop:gtgd_one_pass}, we now want to recursively 
re-create subtrees in each~$\overline{L_i}$ and create some $h_{\overline{L_i}}$ that preserves the tree structure of $L_i$.
We do this by iteratively updating $\overline{T}^*$ for $1\leq i\leq m$ in the following way:

Beginning from $p_0\coloneqq v$, 
for any leaf $\overline{L_i}$ in $\overline{T}^*$ associated to $L_i$,
let $\overline{p_j}$ be the copy of~$p_j$ in~$T$ as given by $h_{\overline{L}}$
($\overline{L_i}$ is a descendant of some $\overline{L}\in S$),
and $\overline{p_j}'$ be the non-strict ancestor of $\overline{v}'$ corresponding to $\overline{p_j}$.
Further, let $N_{\overline{L_i}}$ be the node in $\overline{T}$ in which $\overline{p_j}$ 
got created and $M_{\overline{L_i}}$ be the last node in 
$N_{\overline{L_i}},\dotsc,\overline{L}$ in which the subtree of $\overline{p_j}$ added 
on a fact.
Then let $S_N$ consist of all such nodes $N_{\overline{L_i}}$,
and $S_M$ consist of all such nodes $M_{\overline{L_i}}$.

For each~$\overline{L_i}$, we would now like to create all missing
facts in the subtree of $\overline{p_j}'$ by re-firing steps 
from $\cut\bigl(\overline{T}^{*N_{\overline{L_i}}},M_{\overline{L_i}}\bigr)$;
however, this might create new independent threads with leaves that
are again associated to $L_i$ but in which the subtree of $\overline{p_j}'$ is still missing facts.
To solve this issue, for $U\in\{S_N,S_M\}$, we first remove any node from $U$ that has an ancestor in~$U$.
Then, for each $N\in S_N$,
let $M_N$ consist of all $M\in S_M$
such that $N$ is an ancestor of~$M$.
Fix some $N\in S_N$.
\begin{claim}
$\cut\bigl(\overline{T}^{*N},M_N\bigr)$ contains no child of any $\overline{L}\in S$.
\end{claim} 
\begin{subproof}
Let $\overline{L}\in S$ and assume $\overline{L}$ is a non-strict descendant of $N$.
There is at least one descendant 
$\overline{L_i}$ associated to $L_i$ with corresponding $N_{\overline{L_i}},M_{\overline{L_i}}$.
As $N$ is in~$S_N$, $N_{\overline{L_i}}$ must be a non-strict descendant of $N$,
and hence, also $M_{\overline{L_i}}$ must be a non-strict descendant of $N$.
Thus, either $M_{\overline{L_i}}$ or one of its ancestors is in $S_M$.
All children of nodes in $S_M$ are removed, and consequently, 
no child of $\overline{L}$ occurs in $T_N$.
\end{subproof} 
Now, for any $\overline{L_i}$ associated to $L_i$ that is a descendant of $N$, 
we re-fire the steps that created $\cut\bigl(\overline{T}^{*N},M_N\bigr)$ starting from $\overline{p_j}'$
to create all missing facts in the subtree of $\overline{p_j}'$.
Since $\cut\bigl(\overline{T}^{*N},M_N\bigr)$ contains no children of any $\overline{L}\in\overline{S}$,
for any newly created leaf $\overline{L}'$ in this process, there is either 
some $L'\in\lvs(T_n)$ that is homomorphic to $\overline{L}'$,
or $\overline{L}'$ is associated to $L_i$ and the subtree of $\overline{p_j}'$
got repaired by the just fired chase steps. 
Once we processed any $N\in S_N$, we let $p_{j+1}$ be the parent of $p_j$ and proceed with $p_{j+1}$.

Finally, for any leaf $\overline{L_i}$ associated to some $L_i$
in the final $\overline{T}^*$, and any node $u$ in $L_i$,
we let $h_{\overline{L_i}}$ map to the newest copy of $u$ in $\overline{T}^*$.
Like in \cref{prop:gtgd_one_pass},
each constructed $h_{\overline{L_i}}$ will be a homomorphism that 
preserves the tree structure of $L_i$.
\end{proof}

\begin{thm}\label{thm:dgtgd_iff_one_pass}
A \ubcq $Q$ has a chase proof from $\db$ and \dgtgds $\Sigma$
\ifftext it has a one-pass chase proof from $\db$ and $\Sigma$.
\end{thm} 
\begin{proof}
If we have one-pass chase proof of $Q$, we also have a chase proof of $Q$
since every one-pass chase proof is a chase proof.

Assume we have a chase proof $T$ of $Q$, and
let $\overline{T}$ be the corresponding one-pass chase tree
as created in \cref{prop:dgtgd_one_pass}.
Take a leaf $\overline{L}$ in $\overline{T}$ and let $L$
be the leaf in $T$ that is homomorphic to $\overline{L}$.
Since $L\models Q$, we have $L\models Q_i$ for some $Q_i\in Q$.
Hence, by \cref{lem:hom_chase_seq_impl_bcq}, $\overline{L}\models Q_i$,
and thus $\overline{L}\models Q$.
\end{proof}

%% file: figures/exmpl_tree_like_chase_tree.tex
\begin{tikzpicture}
\tikzset{node distance=0.3cm,
	r-0/.pic={\node (r) {$R(c,d)$};},
	T0/.pic={%
		\node (N) {
			\begin{tikzpicture}
				\pic[trigger]{r-0};
			\end{tikzpicture}
		};
	},
	Nl-cm-1/.pic={\node (Nl-cm-1) {$S(c,d_1)$};},
	Nr-cm-1/.pic={\node (Nr-cm-1) {$T(c,d,d_2),R(c,d)$};},
	Nl-1/.pic={%
		\pic{r-0};
		\pic[below=of r]{Nl-cm-1};
		\path (r) edge (Nl-cm-1);
	},
	Nr-1/.pic={%
		\pic{r-0};
		\pic[trigger,below=of r]{Nr-cm-1};
		\path (r) edge (Nr-cm-1);
	},
	T1/.pic={%
		\node (N) {
			\begin{tikzpicture}
				\pic{r-0};
			\end{tikzpicture}
		};
		\node[below=of N,xshift=-1.1cm] (Nl) {
			\begin{tikzpicture}
				\pic{Nl-1};
			\end{tikzpicture}
		};
		\node[right=of Nl] (Nr) {
			\begin{tikzpicture}
				\pic{Nr-1};
			\end{tikzpicture}
		};
		\path (N) edge (Nl)
		(N) edge (Nr);
	},
	Nr-2/.pic={%
		\pic{r-0};
		\pic[below=of r]{Nr-cm-1};
		\path (r) edge (Nr-cm-1);
	},
	Nrm-cmm-2/.pic={\node (Nrm-cmm-2) {$U(c,d,d_3),R(c,d)$};},
	Nrm-2/.pic={%
		\pic{r-0};
		\pic[below=of r]{Nr-cm-1};
		\pic[trigger,below=of Nr-cm-1]{Nrm-cmm-2};
		\path (r) edge (Nr-cm-1)
		(Nr-cm-1) edge (Nrm-cmm-2);
	},
	T2/.pic={%
		\node (N) {
			\begin{tikzpicture}
				\pic{r-0};
			\end{tikzpicture}
		};
		\node[below=of N,xshift=-1.1cm] (Nl) {
			\begin{tikzpicture}
				\pic{Nl-1};
			\end{tikzpicture}
		};
		\node[right=of Nl] (Nr) {
			\begin{tikzpicture}
				\pic{Nr-2};
			\end{tikzpicture}
		};
		\node[right=of Nr] (Nrm) {
			\begin{tikzpicture}
				\pic{Nrm-2};
			\end{tikzpicture}
		};
		\path (N) edge (Nl)
		(N) edge (Nr)
		(Nr) edge (Nrm);
	},
	Nrm-3/.pic={%
		\pic{r-0};
		\pic[below=of r]{Nr-cm-1};
		\pic[below=of Nr-cm-1]{Nrm-cmm-2};
		\path (r) edge (Nr-cm-1)
		(Nr-cm-1) edge (Nrm-cmm-2);
	},
	Nrml-cmm-3/.pic={\node (Nrml-cmm-3) {$U(c,d,d_3),R(c,d),M(c,d)$};},
	Nrml-cm-3/.pic={\node (Nrml-cm-3) {$T(c,d,d_2),R(c,d),M(c,d)$};},
	Nrml-r-3/.pic={\node (Nrml-r-3) {$R(c,d),M(c,d)$};},
	Nrmr-cmm-3/.pic={\node (Nrmr-cmm-3) {$U(c,d,d_3),R(c,d),P(d)$};},
	Nrmr-cm-3/.pic={\node (Nrmr-cm-3) {$T(c,d,d_2),R(c,d),P(d)$};},
	Nrmr-r-3/.pic={\node (Nrmr-r-3) {$R(c,d),P(d)$};},
	Nrml-3/.pic={%
		\pic{Nrml-r-3};
		\pic[below=of Nrml-r-3]{Nrml-cm-3};
		\pic[below=of Nrml-cm-3]{Nrml-cmm-3};
		\path (Nrml-r-3) edge (Nrml-cm-3)
		(Nrml-cm-3) edge (Nrml-cmm-3);
	},
	Nrmr-3/.pic={%
		\pic{Nrmr-r-3};
		\pic[trigger,below=of Nrmr-r-3]{Nrmr-cm-3};
		\pic[below=of Nrmr-cm-3]{Nrmr-cmm-3};
		\path (Nrmr-r-3) edge (Nrmr-cm-3)
		(Nrmr-cm-3) edge (Nrmr-cmm-3);
	},
	T3/.pic={%
		\node (N) {
			\begin{tikzpicture}
				\pic{r-0};
			\end{tikzpicture}
		};
		\node[below=of N,xshift=-1.1cm] (Nl) {
			\begin{tikzpicture}
				\pic{Nl-1};
			\end{tikzpicture}
		};
		\node[right=of Nl] (Nr) {
			\begin{tikzpicture}
				\pic{Nr-2};
			\end{tikzpicture}
		};
		\node[below=of Nr] (Nrm) {
			\begin{tikzpicture}
				\pic{Nrm-3};
			\end{tikzpicture}
		};
		\node[below=of Nrm,xshift=-1.1cm] (Nrml) {
			\begin{tikzpicture}
				\pic{Nrml-3};
			\end{tikzpicture}
		};
		\node[right=of Nrml] (Nrmr) {
			\begin{tikzpicture}
				\pic{Nrmr-3};
			\end{tikzpicture}
		};
		\path (N) edge (Nl)
		(N) edge (Nr)
		(Nr) edge (Nrm)
		(Nrm) edge (Nrml)
		(Nrm) edge (Nrmr);
	},
	Nrmrm-cm-4/.pic={\node (Nrmrm-cm-4) {$T(c,d,d_2),R(c,d),P(d),M(c,c)$};},
	Nrmrm-cmm-4/.pic={\node (Nrmrm-cmm-4) {$U(c,d,d_3),R(c,d),P(d),M(c,c)$};},
	Nrmrm-r-4/.pic={\node (Nrmrm-r-4) {$R(c,d),P(d),M(c,c)$};},
	Nrmr-4/.pic={%
		\pic{Nrmr-r-3};
		\pic[below=of Nrmr-r-3]{Nrmr-cm-3};
		\pic[below=of Nrmr-cm-3]{Nrmr-cmm-3};
		\path (Nrmr-r-3) edge (Nrmr-cm-3)
		(Nrmr-cm-3) edge (Nrmr-cmm-3);
	},
	Nrmrm-4/.pic={%
		\pic{Nrmrm-r-4};
		\pic[below=of Nrmrm-r-4]{Nrmrm-cm-4};
		\pic[below=of Nrmrm-cm-4]{Nrmrm-cmm-4};
		\path (Nrmrm-r-4) edge (Nrmrm-cm-4)
		(Nrmrm-cm-4) edge (Nrmrm-cmm-4);
	},
	Nl-4/.pic={%
		\pic{r-0};
		\pic[trigger,below=of r]{Nl-cm-1};
		\path (r) edge (Nl-cm-1);
	},
	T4/.pic={%
		\node (N) {
			\begin{tikzpicture}
				\pic{r-0};
			\end{tikzpicture}
		};
		\node[below=of N,xshift=-1.1cm] (Nl) {
			\begin{tikzpicture}
				\pic{Nl-4};
			\end{tikzpicture}
		};
		\node[right=of Nl] (Nr) {
			\begin{tikzpicture}
				\pic{Nr-2};
			\end{tikzpicture}
		};
		\node[right=of Nr] (Nrm) {
			\begin{tikzpicture}
				\pic{Nrm-3};
			\end{tikzpicture}
		};
		\node[below=of Nrm,xshift=-1.1cm] (Nrml) {
			\begin{tikzpicture}
				\pic{Nrml-3};
			\end{tikzpicture}
		};
		\node[right=of Nrml] (Nrmr) {
			\begin{tikzpicture}
				\pic{Nrmr-4};
			\end{tikzpicture}
		};
		\node[below=of Nrmr] (Nrmrm) {
			\begin{tikzpicture}
				\pic{Nrmrm-4};
			\end{tikzpicture}
		};
		\path (N) edge (Nl)
		(N) edge (Nr)
		(Nr) edge (Nrm)
		(Nrm) edge (Nrml)
		(Nrm) edge (Nrmr)
		(Nrmr) edge (Nrmrm);
	},
	Nll-r-5/.pic={\node (Nll-r-5) {$R(c,d),M(c,c)$};},
	Nll-cm-5/.pic={\node (Nll-cm-5) {$S(c,d_1),M(c,c)$};},
	Nll-5/.pic={%
		\pic{Nll-r-5};
		\pic[below=of Nll-r-5]{Nll-cm-5};
		\path (Nll-r-5) edge (Nll-cm-5);
	},
	T5/.pic={%
		\node (N) {
			\begin{tikzpicture}
				\pic{r-0};
			\end{tikzpicture}
		};
		\node[below=of N,xshift=-1.5cm] (Nl) {
			\begin{tikzpicture}
				\pic{Nl-1};
			\end{tikzpicture}
		};
		\node[below=of Nl] (Nll) {
			\begin{tikzpicture}
				\pic{Nll-5};
			\end{tikzpicture}
		};
		\node[right=of Nl] (Nr) {
			\begin{tikzpicture}
				\pic{Nr-2};
			\end{tikzpicture}
		};
		\node[right=of Nr] (Nrm) {
			\begin{tikzpicture}
				\pic{Nrm-3};
			\end{tikzpicture}
		};
		\node[below=of Nrm,xshift=-1.1cm] (Nrml) {
			\begin{tikzpicture}
				\pic{Nrml-3};
			\end{tikzpicture}
		};
		\node[right=of Nrml] (Nrmr) {
			\begin{tikzpicture}
				\pic{Nrmr-4};
			\end{tikzpicture}
		};
		\node[right=of Nrmr] (Nrmrm) {
			\begin{tikzpicture}
				\pic{Nrmrm-4};
			\end{tikzpicture}
		};
		\path (N) edge (Nl)
		(Nl) edge (Nll)
		(N) edge (Nr)
		(Nr) edge (Nrm)
		(Nrm) edge (Nrml)
		(Nrm) edge (Nrmr)
		(Nrmr) edge (Nrmrm);
	},
}
\node[scope] (T0) {
	\begin{tikzpicture}
		\pic{T0};
		\node[lnode,above=0.1cm of N]{$T_0$};
	\end{tikzpicture}
};
\node[scope,right=of T0] (T1) {
	\begin{tikzpicture}
		\pic{T1};
		\node[lnode,right=of N]{$T_1$};
	\end{tikzpicture}
};
\node[scope,right=of T1] (T2) {
	\begin{tikzpicture}
		\pic{T2};
		\node[lnode,right=of N]{$T_2$};
	\end{tikzpicture}
};
\node[scope,below=1.8cm of T0,xshift=0cm] (T3) {
	\begin{tikzpicture}
		\pic{T3};
		\node[lnode,right=of N]{$T_3$};
	\end{tikzpicture}
};
\node[scope,right=of T3] (T4) {
	\begin{tikzpicture}
		\pic{T4};
		\node[lnode,right=of N]{$T_4$};
	\end{tikzpicture}
};
\node[scope,below=0.8cm of T3,xshift=7cm] (T5) {
	\begin{tikzpicture}
		\pic{T5};
		\node[lnode,right=of N]{$T_5$};
	\end{tikzpicture}
};
\pic{drawv=T0/T1/T2/T2};
\pic{drawvr=T1/T2};
\pic{drawvr=T3/T4};
\pic{drawh=T2/T3/T3/T4};
\pic{drawh=T4/T5/T3/T4};
\end{tikzpicture}

%% file: chapters/gtgd.tex
\chapter{\qfqagtgds}\label{chap:gtgd}
In this chapter, we present two procedures for \gtgds that terminate and output an atomic rewriting.
To accomplish this, we first introduce some normal forms.
Then, we describe an abstract property 
that implies completeness of a rewriting algorithm,
which we call \emph{evolve-completeness}.
We then present two sound saturation processes that are evolve-complete and hence complete.
For ease of exposition, we first introduce a non-optimal, 
but arguably more intuitive algorithm, which we call the \emph{simple saturation}.
We then discuss problems arising with this approach and introduce an improved version,
which we call the \emph{guarded saturation}.

\input{pages/gtgd/gtgd_normal_forms}
\input{pages/gtgd/gtgd_property_for_completeness}
\input{pages/gtgd/gtgd_simple_saturation}
\input{pages/gtgd/gtgd_guarded_saturation}

%% file: pages/gtgd/gtgd_normal_forms.tex
\section{Normal Forms}

We introduce two normal forms for \tgds
that facilitate subsequent proofs and ensure termination of our rewriting algorithms.
Firstly, we want to normalise the variables of our set of \tgds
to avoid duplicate isomorphic rules in our rewriting.

\begin{defn}[\tgd variable normal form]
Given a \tgd $\tau$ with body $\body$ and head $\head$, 
we say that $\tau$ is in \index{variable normal form!TGD}\emph{variable normal form}
if $\vars(\body)=\vec x$, $\vars(\head)\setminus\vec x=\vec y$,
$x_i$ is the $\nth{i}$ lexicographic variable occurring in $\body$, 
and $y_i$ the $\nth{i}$ lexicographic variable occurring in $\head$.

We write $\vnf(\tau)$ for the unique \tgd obtained by rewriting $\tau$ into variable normal form,
and given a set of \tgds $\Sigma$, we write \index[not]{$\vnf(\Sigma)$}$\vnf(\Sigma)\coloneqq\{\vnf(\tau)\mid\tau\in\Sigma\}$
for the variable normal form of $\Sigma$.
\end{defn} 

\begin{exmpl}
The \tgd 
\begin{equation*}
	\forall x_3,y_2,x_1\bigl[B(y_2,x_1,x_3)\rightarrow\exists y_1,z_1,y_2\bigl(H_1(x_1,z_1,y_1,y_2)\land H_2(y_1,y_2)\bigr)\bigr]
\end{equation*} 
is not in variable normal form. Its variable normal form is
\begin{equation*}
	\forall x_1,x_2,x_3\bigl[B(x_1,x_2,x_3)\rightarrow\exists y_1,y_2,y_3\bigl(H_1(x_2,y_1,y_2,y_3)\land H_2(y_2,y_3)\bigr)\bigr].
\end{equation*} 
\end{exmpl} 

\begin{lem}\label{lem:tgd_vnf_complexity}
Any \tgd $\tau$ can be rewritten into $\vnf(\tau)$ in linear time.
\end{lem} 
\begin{proof}
First rename the universal and existential variables of $\tau$ into disjoint sets $V_\forall,V_\exists$
and obtain the corresponding $\tau'$.
Then scan $\tau'$ left to right to get an order on $V_\forall$ and $V_\exists$, respectively,
and choose fresh variables ${v_1,\dotsc,v_{|V_\forall|},w_1,\dotsc,w_{|V_\exists|}}$.
Then rename $\tau'$ using~$\rho$ defined by
\begin{equation*}
	x\rho\coloneqq\begin{cases}
		v_i,&\text{if $x$ is the $\nth{i}$ variable in $V_\forall$}\\
		w_i,&\text{if $x$ is the $\nth{i}$ variable in $V_\exists$}\\
	\end{cases}.
\end{equation*} 
Finally, rename $\tau'\rho$ using $\rho'$ defined 
by $v_i\rho'\coloneqq x_i$ and $w_i\rho'\coloneqq y_i$.
All steps can clearly be done in linear time.
\end{proof} 

Secondly, we want to avoid ``mixed'' rules; that is,
non-full \tgds which become full rules when selecting a subset of head atoms.
\begin{defn}[\tgd head normal form]\label{def:tgd_hnf}
We say that a \tgd $\tau$ is in \index{head normal form}\emph{head normal form} 
if $\tau$ is full or every head atom contains at least one existential variable.

Given a set of \tgds $\Sigma$, the \emph{head normal form \iindex[not]{$\hnf(\Sigma)$} of $\Sigma$} 
is obtained by replacing every rule
\begin{equation*}
	\body(\vec x)\rightarrow\exists\vec y\,\bigwedge\limits_{i=1}^nH_i(\vec x,\vec y),
\end{equation*} 
that is not in head normal form by two rules
\begin{align*}
	\body(\vec x)&\rightarrow\exists\vec y\,\bigwedge\limits_{i\in I}H_i(\vec x,\vec y),\\
	\body(\vec x)&\rightarrow\bigwedge\limits_{i\in(\{1,\dotsc,n\}\setminus I)}H_i(\vec x),
\end{align*} 
where $I\coloneqq\{i\mid \vars(H_i)\cap\vec y\neq\emptyset\}$.
This replacement can be done in time linear in \index[not]{$\size(\Sigma)$}$\size(\Sigma)\coloneqq\sum_{\tau\in\Sigma}\size(\tau)$.
\end{defn}

\begin{lem}\label{lem:tgd_hnf_closure_coincide}
Given a database $\db$, a set of \tgds $\Sigma$, and a \ubcq $Q$,
we have $\db,\Sigma\models Q$ \ifftext $\db,\hnf(\Sigma)\models Q$.
\end{lem} 
\begin{proof}
By \cref{thm:chase_univ}, it suffices to show that $Q$ has a chase proof
from $\db$ using $\Sigma$ \ifftext it has a chase proof from $\db$ using $\hnf(\Sigma)$.
Any $\tau\in\Sigma\cap\hnf(\Sigma)$ can fire equivalently in both chases.
If we fire a $\tau\in\Sigma\setminus\hnf(\Sigma)$ in a chase proof using $\Sigma$, we can
simply fire both rules from $\hnf(\tau)$ in a chase proof using $\hnf(\Sigma)$.
Conversely, if we fire a $\tau\in\hnf(\Sigma)\setminus\Sigma$ in a chase proof using $\hnf(\Sigma)$, 
we can simply fire the rule from which $\tau$ was created to get a chase proof using $\Sigma$.
\end{proof} 

The reason for the head normal form transformation is to 
simplify our correctness proofs
and to minimise the number of variables occurring in each head of a rule.
This in return will facilitate required unification steps in our upcoming rewriting algorithms.
\begin{defn}[\tgd width]
Given a \tgd $\tau$ with body $\body$ and head $\head$, 
we define
\begin{thmlist}
\item the \index{body width!\tgd}\emph{body width of $\tau$} as $\bwidth(\tau)\coloneqq|\vars(\body)|$,
\item the \index{head width!\tgd}\emph{head width of $\tau$} as $\hwidth(\tau)\coloneqq|\vars(\head)|$, and
\item the \index{width!\tgd}\emph{width of $\tau$} as $\width(\tau)\coloneqq\max\{\bwidth(\tau),\hwidth(\tau)\}$.
\end{thmlist} 
The definitions are extended to set of \tgds $\Sigma$
by taking the maximum of any $\tau\in\Sigma$, e.g. 
\index[not]{$\bwidth(\Sigma)$}\index[not]{$\hwidth(\Sigma)$}\index[not]{$\width(\Sigma)$}
$\width(\Sigma)\coloneqq\max\{\width(\tau)\mid\tau\in\Sigma\}$.
\end{defn} 
\begin{exmpl}
The \tgd
\begin{equation*}
	\forall x_1,x_2\bigl[B(x_1,x_2)\rightarrow\exists y_1\, H_1(x_1,y_1)\land H_2(x_2)\bigr]
\end{equation*} 
has width 3, whereas its head normal form
\begin{align*}
	\forall x_1,x_2\bigl[B(x_1,x_2)&\rightarrow\exists y_1\, H_1(x_1,y_1)\bigr],\\
	\forall x_1,x_2\bigl[B(x_1,x_2)&\rightarrow H_2(x_2)\bigr]
\end{align*} 
has width 2.
\end{exmpl} 
\begin{lem}
For any set of \tgds $\Sigma$, we have $\hwidth(\hnf(\Sigma))\leq\hwidth(\Sigma)$.
\end{lem} 
\begin{proof}
Any head of a rule in $\hnf(\Sigma)$ is a subset of a head of a rule in $\Sigma$.
\end{proof} 

%% file: pages/gtgd/gtgd_property_for_completeness.tex
\section{An Abstract Property for Completeness}\label{sec:gtgd_property_for_completeness}
Recall the structure of a tree-like chase from some database $\db$ and set of \gtgds $\Sigma$.
A tree-like chase models some \ubcq $Q$ \ifftext all the facts of some $Q_i\in Q$ are contained in
the root node~$r$ of the final instance of the chase.
Also recall that an atomic rewriting $\Sigma'$ must be a set of full rules.
Consequently, in order to obtain completeness of $\Sigma'$, 
we must ensure that whenever a fact $F$ relevant to~$r$ is created
by a sequence of chase steps using $\Sigma$, $F$ can also be created by a sequence of full rules using $\Sigma'$.
It is evident that $F$ cannot only be created by firing a full rule in $r$, 
but also by firing a rule in a descendant of $r$, as can be seen in \cref{fig:exmpl_tree_chase_seq}.
To derive sufficient completeness properties for $\Sigma'$, we can consider these two cases separately.

If we fire a full rule $\tau\in\Sigma$ in $r$, we can simply add $\tau$ to $\Sigma'$ 
and then fire $\tau$ in the chase using $\Sigma'$.
Our first property will thus require that every full rule of $\Sigma$ will be contained in $\Sigma'$.
If, however, we fire a rule $\tau$ in a descendant $d$ of $r$,
we cannot create a copy of $d$ in a chase using $\Sigma'$ (as $\Sigma'$ is full).
We hence need a second property which ensures that whenever a fact relevant to~$r$ is created
by firing some $\tau\in\Sigma$ in a descendant of $r$, 
the fact can also be created by a sequence of full rules in $\Sigma'$.
Due to an inductive argument (\cref{lem:gtgd_abstr_completeness_lift}), 
we can even relax this property a bit further and focus on facts created in a 
child of $r$ instead of general descendants.
These thoughts capture the general idea of the following properties:

\begin{restatable}[Evolve-completeness]{defn}{gtgdabstrcompletenessprops}\label{defn:gtgd_abstr_completeness_props}
Let $\Sigma$ be a set of \gtgds and $\Sigma'$ be a set of full \tgds.
We say that $\Sigma'$ is \index{evolve-completeness!\gtgd}\emph{evolve-complete \wrt $\Sigma$} if it satisfies
the following properties:
\begin{propylist}
\item\label{defn:gtgd_abstr_completeness_props_1}
For every full $\tau\in\hnf(\Sigma)$, we have $\tau\in\Sigma'$ (modulo renaming).

\item\label{defn:gtgd_abstr_completeness_props_2}
Let $I_0,\dotsc,I_{n+1}$ be a chase,
let $r$ be the node inside $I_0$,
let $\tau_r\in\hnf(\Sigma)$ be non-full
and ${\tau_1,\dotsc,\tau_n\in\Sigma'}$ be full such that
\begin{assmlist}
\item firing $\tau_r$ in $I_0$ creates $I_1$,
\item firing $\tau_i$ in $v$ in $I_i$ creates $I_{i+1}$, where $v$ is the child of $r$.
\end{assmlist} 
Then there exists a chase $I_0',\dotsc,I_m'$ from $F_r^n$ and $\Sigma'$ of $F_r^{n+1}$.
\end{propylist}
An illustration of \cref{defn:gtgd_abstr_completeness_props_2} can be found in \cref{fig:gtgd_evolve_completeness}.
\end{restatable} 

\begin{figure}[ht]
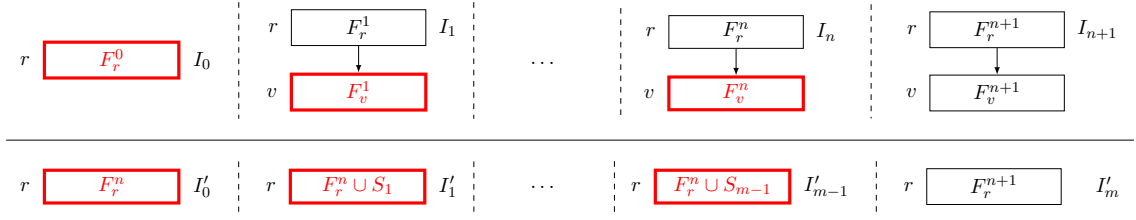

	\centering
	\includestandalone[width=\textwidth]{./figures/gtgd_evolve_completeness}
	\caption{Illustration of \cref{defn:gtgd_abstr_completeness_props_2};
	the upper chase is given by assumption and the bottom chase is obtained from \cref{defn:gtgd_abstr_completeness_props_2},
	where $S_1\subseteq\dotsb\subseteq S_{m-1}\subseteq(F_r^{n+1}\setminus F_r^n)$.
	Triggered nodes are marked bold and red, and the nodes' names are written on the left.}
	\label{fig:gtgd_evolve_completeness}
\end{figure}

Before we prove that our properties indeed imply completeness of $\Sigma'$,
we introduce an apparent strengthening of \cref{defn:gtgd_abstr_completeness_props_2}.
\begin{defn}[Strong evolve-completeness]
Let $\Sigma$ be a set of \gtgds and $\Sigma'$ be a set of full \tgds.
We say that $\Sigma'$ is \emph{strongly evolve-complete \wrt $\Sigma$} if it satisfies
\cref{defn:gtgd_abstr_completeness_props_1} and the following strengthening of
\cref{defn:gtgd_abstr_completeness_props_2}:
\begin{propylist}
\item\label{defn:gtgd_abstr_completeness_props_induct}
Let $I_0,\dotsc,I_{n+1}$ be a chase,
let $r$ be the node inside $I_0$,
let $\tau_r\in\hnf(\Sigma)$ be non-full
and ${\tau_1,\dotsc,\tau_n\in\Sigma'}$ be full such that
\begin{assmlist}
\item firing $\tau_r$ in $I_0$ creates $I_1$,
\item firing $\tau_i$ in $v$ in $I_i$ creates $I_{i+1}$, where $v$ is the child of $r$.
\end{assmlist} 
Then there exists a chase $I_0',\dotsc,I_m'$ from $I_0$ and $\Sigma'$ of $F_r^{n+1}$.
\end{propylist}
\end{defn} 

It turns out, however, that both properties are equivalent:
\begin{lem}\label{lem:gtgd_abstr_completeness_props_induct}
$\Sigma'$ is evolve-complete \wrt $\Sigma$ \ifftext $\Sigma'$ is strongly evolve-complete \wrt $\Sigma$.
\end{lem} 
\begin{proof}
Right to left is a strengthening: if there is a chase from $I_0$ of $F_r^{n+1}$, 
then there is also one from $F_r^{n}\supseteq I_0$.

We prove the other direction by induction on $n$. 
For $n=0$, there is nothing to do since~$\tau_r$ is in head normal form and non-full, and thus $I_0=F_r^1$.

Assume $n>0$.
By the inductive hypothesis, we get a chase $I_0',\dotsc,I_m'$ from $I_0$ and $\Sigma'$ of $F_r^n$.
Now we apply \cref{defn:gtgd_abstr_completeness_props_2} on the chase
$F_r^n,I_1,\dotsc,I_{n+1}$ to get a chase $I_{m+1}'\dotsc,I_{m+m'}'$ 
from $F_r^n$ and $\Sigma'$ of $F_r^{n+1}$.
Hence, by connecting the chases, we get a chase from $I_0$ and $\Sigma'$ of~$F_r^n$.
\end{proof} 

Due to the just shown equivalence, we subsequently use the notions of evolve-completeness and
strong evolve-completeness interchangeably.

Given a set of \gtgds $\Sigma$, we now show that every evolve-complete set $\Sigma'$ 
is in fact complete \wrt $\Sigma$; that is, any quantifier-free \ubcq following from $\Sigma$
will also follow from $\Sigma'$.
We first prove the following key lifting lemma that relies on the one-pass chase property.
\begin{lem}\label{lem:gtgd_abstr_completeness_lift}
Let $\Sigma$ be a set of \gtgds, 
let $\Sigma'$ be evolve-complete \wrt $\Sigma$,
and let $I_0,\dotsc,I_n$ be a one-pass chase using~$\hnf(\Sigma)$.
Assume a node $v$ is created at step $i>0$ by firing a rule $\tau_p$ in parent $p$
and let $j\geq i$ such that $v$ adds on facts.
Then there is a chase $I_0',\dotsc,I_m'$ from $F_p^{i-1}$ and~$\Sigma'$ of $F_p^j$.
\end{lem}
\begin{proof}
We additionally show that there is a chase $I_0^*,\dotsc,I_k^*$ from $F_v^i$
and $\Sigma'$ of $F_v^j$.
We prove the claim by induction on $j-i$. 
In case $i=j$, there is nothing to do since $\tau_p$ is in head normal form.

In case $i>j$,
we have a one-pass chase proof from $F_v^i$ to $F_v^j$ by \cref{lem:one_pass_proof}.
Hence, we can assume that the new facts relevant to $v$ are generated by some 
$\tau\in\hnf(\Sigma)$ acting on some node~$d$ that is a non-strict descendant of $v$. 

First consider the case $v=d$. 
By the inductive hypothesis, we get a chase $I_0^*,\dotsc,I_k^*$ from~$F_v^i$ and~$\Sigma'$ of $F_v^{j-1}$.
Since~$v$ adds on facts in step $j$, $\tau$ must be full, 
and hence $\tau\in\Sigma'$ (modulo renaming) by \cref{defn:gtgd_abstr_completeness_props_1}.
Thus, we can fire $\tau$ in $I_k^*$ to obtain $I_{k+1}^*$
with $I_{k+1}^*\models F_v^j$.
Now we can apply \cref{defn:gtgd_abstr_completeness_props_induct}
on $F_p^{i-1},I_0^*,\dotsc,I_{k+1}^*$ to obtain the desired $I_0',\dotsc,I_m'$.

Next consider the case where $d$ is a descendant of $v$. 
Let $c$ be the first child of $v$ on the path to $d$, created at time $k>i$. 
By the inductive hypothesis, we get a chase $I_0^*,\dotsc,I_k^*$ from~$F_v^i$ and~$\Sigma'$ of~$F_v^{k-1}$.
Also by the inductive hypothesis, we get a chase $I_{k+1}^*,\dotsc,I_{k'}^*$ from~$F_v^{k-1}$ and~$\Sigma'$ of~$F_v^j$.
Now we can connect both chases and then proceed as before. That is, we use \cref{defn:gtgd_abstr_completeness_props_induct} to obtain the desired $I_0',\dotsc,I_m'$.
\end{proof}

The completeness property now follows by a simple induction.
\begin{prop}[Completeness]\label{prop:gtgd_abstr_completeness}
Let $\Sigma$ be a set of \gtgds and
$\Sigma'$ be evolve-complete \wrt $\Sigma$.
Then $\Sigma'$ is complete \wrt $\Sigma$.
\end{prop} 
\begin{proof}
By \cref{thm:chase_univ,thm:gtgd_iff_one_pass,lem:tgd_hnf_closure_coincide}, 
it suffices to show that whenever a quantifier-free \ubcq $Q$ has a 
one-pass chase proof $I_0,\dotsc,I_n$ from $\db$ and $\hnf(\Sigma)$, 
then~$Q$ also follows from $\db$ and $\Sigma'$.
Let $r$ be the root inside the one-pass chase proof.
Since $I_n\models Q$ \ifftext $I_n\models Q_i$ for some $Q_i\in Q$ \ifftext $Q_i\subseteq F_r^n$,
it suffices to show that we can create a chase $I_0',\dotsc,I_m'$ from
from $\db$ and $\Sigma'$ with $I_m'\models F_r^n$.
We show the claim by induction on~$n$.
The case $n=0$ is vacuously true.

In case $n>0$,
let $\tau\in\Sigma$ be the rule firing at step $n$ in some node $d$,
and assume that~$r$ adds on a fact in $I_n$;
otherwise, we can just use the chase $I_0',\dotsc,I_m'$ obtained by the inductive hypothesis.

If $d=v$, then $\tau$ must be full, and thus, $\tau\in\Sigma'$ by \cref{defn:gtgd_abstr_completeness_props_1}.
Moreover, by the inductive hypothesis, we get $I_0',\dotsc,I_m'$ such that $I_m'\models F_r^{n-1}$.
Hence, we can fire $\tau$ in $I_m'$ to obtain $I_{m+1}'$ with $I_{m+1}'\models F_r^n$.

Otherwise, let $c$ be the first child of the root on the path to $d$, created at step ${k<n}$.
By the inductive hypothesis, we get $I_0',\dotsc,I_m'$ with $I_m'\models F_r^{k-1}$.
Further, by \cref{lem:gtgd_abstr_completeness_lift}, we get $I_{m+1}',\dotsc,I_{m+m'}'$
from $F_r^{k-1}$ and $\Sigma'$ with $I_{m+m'}'\models F_r^n$.
Hence, by connecting the chases, we get a chase from $\db$ and $\Sigma'$ of $F_r^n$.
\end{proof} 

\begin{cor}\label{cor:gtgd_abstr_atomic_rewriting}
Let $\Sigma$ be a set of \gtgds and
$\Sigma'$ be finite, sound, and evolve-complete \wrt~$\Sigma$.
Then $\Sigma'$ is an atomic rewriting of $\Sigma$.
\end{cor} 
\begin{proof}
$\Sigma'$ is full by \cref{defn:gtgd_abstr_completeness_props},
finite and sound by assumption,
complete by \cref{prop:gtgd_abstr_completeness},
and hence an atomic rewriting by \cref{cor:atomic_rewriting_def}.
\end{proof} 

Having a sufficient completeness property at hand, 
we next introduce two sound rewriting algorithms that fulfil this property.

%% file: figures/gtgd_evolve_completeness.tex
\begin{tikzpicture}
\tikzset{node distance=0.5cm,every node/.style={shape=rectangle,draw,minimum width=2.5cm},lnode/.style={draw=none,minimum width=0cm},
	r-0/.pic={\node (r) {$F_r^0$};},
	I0/.pic={\pic[trigger]{r-0};},
	r-1/.pic={\node (r) {$F_r^1$};},
	cl-1/.pic={\node (cl) {$F_v^1$};},
	I1/.pic={%
		\pic{r-1};
		\pic[trigger,below=of r]{cl-1};
		\path (r) edge (cl);
	},
	r-n/.pic={\node (r) {$F_r^{n}$};},
	cl-n/.pic={\node (cl) {$F_v^{n}$};},
	In/.pic={%
		\pic{r-n};
		\pic[trigger,below=of r]{cl-n};
		\path (r) edge (cl);
	},
	r-n1/.pic={\node (r) {$F_r^{n+1}$};},
	cl-n1/.pic={\node (cl) {$F_v^{n+1}$};},
	In1/.pic={%
		\pic{r-n1};
		\pic[below=of r]{cl-n1};
		\path (r) edge (cl);
	},
	r'-0/.pic={\node (r) {$F_r^n$};},
	I0'/.pic={%
		\pic[trigger]{r'-0};
	},
	r'-1/.pic={\node (r) {$F_r^n\cup S_1$};},
	I1'/.pic={%
		\pic[trigger]{r'-1};
	},
	r'-m-1/.pic={\node (r) {$F_r^n\cup S_{m-1}$};},
	Im-1'/.pic={%
		\pic[trigger]{r'-m-1};
	},
	Im'/.pic={%
		\pic{r-n1};
	},
}
\node[scope] (I0') {
	\begin{tikzpicture}
		\pic{I0'};
		\node[lnode,right=0.1cm of r]{$I_{0}'$};
		\node[lnode,left=0.1cm of r]{$r$};
	\end{tikzpicture}
};
\node[scope,right=of I0'] (I1') {
	\begin{tikzpicture}
		\pic{I1'};
		\node[lnode,right=0.1cm of r]{$I_1'$};
		\node[lnode,left=0.1cm of r]{$r$};
	\end{tikzpicture}
};
\pic{drawvr=I0'/I1'};
\node[scope,right=0.1cm of I1'] (dotsI') {$\dotso$};
\pic{drawv=I1'/dotsI'/I1'/I1'};
\node[scope,right=0.1cm of dotsI'] (Im-1') {
	\begin{tikzpicture}
		\pic{Im-1'};
		\node[lnode,right=0.1cm of r]{$I_{m-1}'$};
		\node[lnode,left=0.1cm of r]{$r$};
	\end{tikzpicture}
};
\pic{drawvr=dotsI'/Im-1'};
\node[scope,right=of Im-1'] (Im') {
	\begin{tikzpicture}
		\pic{Im'};
		\node[lnode,right=of r]{$I_m'$};
		\node[lnode,left=0.1cm of r]{$r$};
	\end{tikzpicture}
};
\pic{drawvr=Im-1'/Im'};

\node[scope,above=1.4cm of I0'] (I0) {
	\begin{tikzpicture}
		\pic{I0};
		\node[lnode,right=0.1cm of r]{$I_0$};
		\node[lnode,left=0.1cm of r]{$r$};
	\end{tikzpicture}
};
\node[scope,above=0.8cm of I1'] (I1) {
	\begin{tikzpicture}
		\pic{I1};
		\node[lnode,right=0.1cm of r]{$I_1$};
		\node[lnode,left=0.1cm of r]{$r$};
		\node[lnode,left=0.1cm of cl]{$v$};
	\end{tikzpicture}
};
\pic{drawvr=I0/I1};
\node[scope,above=2cm of dotsI'] (dotsI) {$\dotso$};
\pic{drawv=I1/dotsI/I1/I1};
\node[scope,above=0.8cm of Im-1'] (In) {
	\begin{tikzpicture}
		\pic{In};
		\node[lnode,right=0.1cm of r]{$I_{n}$};
		\node[lnode,left=0.1cm of r]{$r$};
		\node[lnode,left=0.1cm of cl]{$v$};
	\end{tikzpicture}
};
\pic{drawvr=dotsI/In};
\node[scope,above=0.8cm of Im'] (In1) {
	\begin{tikzpicture}
		\pic{In1};
		\node[lnode,right=0.1cm of r]{$I_{n+1}$};
		\node[lnode,left=0.1cm of r]{$r$};
		\node[lnode,left=0.1cm of cl]{$v$};
	\end{tikzpicture}
};
\pic{drawvr=In/In1};

\draw[-,solid] let \p1=(I1.south),\p2=(I0'.north),\p3=(I0'.west),\p4=(Im'.east) in
(\x3,{(\y1+\y2)/2}) -- (\x4,{(\y1+\y2)/2});

\end{tikzpicture}

%% file: pages/gtgd/gtgd_simple_saturation.tex
\section{The Simple Saturation}\label{sec:gtgd_simple_saturation}
There are different ways to achieve the properties defined in the previous section.
Our first algorithm will be based on the ideas of the following example.
\begin{exmpl}\label{exmpl:tree_chase_seq_sclo}
Given the database $\db=\{R(c),S(c)\}$ and a set of \gtgds $\Sigma$
\begin{align*}
\tau_1&=R(x_1)\rightarrow\exists y_1,y_2\,T(x_1,y_1,y_2),&
\tau_2&=T(x_1,x_2,x_3)\rightarrow\exists y\,U(x_1,x_2,y),\\
\tau_3&=U(x_1,x_2,x_3)\rightarrow P(x_1)\land V(x_1,x_2),&
\tau_4&=T(x_1,x_2,x_3)\land V(x_1,x_2)\land S(x_1)\rightarrow M(x_1),
\end{align*}
the tree-like chase depicted in \cref{fig:exmpl_tree_chase_seq_sclo} is a chase proof of $Q=P(c)\land M(c)$.
Let us call the root node of this chase $r$, its child $c$, and the leaf node $l$.
\begin{figure}[ht]
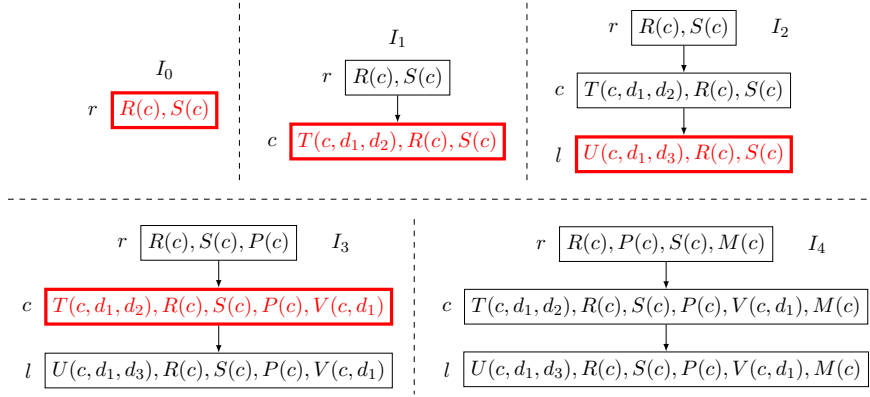

	\centering
	\includestandalone[width=0.77\textwidth]{./figures/exmpl_tree_chase_seq_sclo}
	\caption{Tree-like chase for \cref{exmpl:tree_chase_seq_sclo};
	triggered nodes are marked bold and red, node names are labelled at the left.}
	\label{fig:exmpl_tree_chase_seq_sclo}
\end{figure}

In the third step, $r$ adds on a fact $P(c)$ by firing the full $\tau_3$ in $l$.
In order to get a corresponding full rule that can be fired in~$r$,
we first lift $\tau_3$ so that it can be fired in the parent of~$l$
by composing it with the non-full $\tau_2$ to obtain ${\tau_1'\coloneqq T(x_1,x_2,x_3)\rightarrow P(x_1)\land V(x_1,x_2)}$.
Similarly, we can now compose the non-full $\tau_1$ and full $\tau_1'$ to obtain the full $\tau_2'\coloneqq R(x_1)\rightarrow P(x_1)$,
which can now be fired in $r$.
This idea of composing non-full and full rules to lift triggers of rules to ancestor nodes
will be captured by the \orig-rule in the following algorithm.

Continuing with our chase, $r$ adds on another fact $M(c)$ in step four by firing $\tau_4$ in~$c$.
Directly composing $\tau_1$ with $\tau_4$ as we did before, however, does not result in a full rule
since the variable $x_2$ of the premise $V(x_1,x_2)$ in $\tau_4$ will be mapped to the existential variable $y_1$ of~$\tau_1$.
Instead, we note that we already obtained the rule $\tau_1'$ that can be fired in $c$ and discharge this premise.
Hence, we first compose the full $\tau_1'$ and $\tau_4$ to obtain the full $\tau_3'\coloneqq T(x_1,x_2,x_3)\land S(x_1)\rightarrow M(x_1)$.
This idea of composing two full rules to discharge premises will be captured by the \comp-rule in the following algorithm.
Finally, composing the non-full $\tau_1$ and full $\tau_3'$ gives the full $\tau_4'\coloneqq R(x_1)\land S(x_1)\rightarrow M(x_1)$,
which can now be fired in $r$.

Lastly note that any child of $r$ uses at most $\hwidth(\Sigma)$ many constants,
which allows us to restrict our resolution process to only derive rules with up to $\hwidth(\Sigma)$ many variables.
This restriction will ensure termination.
\end{exmpl}

Using these observations, we now present a simple saturation process
for \gtgds that returns an atomic rewriting.
\begin{defn}[Simple saturation]\label{defn:gtgd_simple_closure}
Given a set of \gtgds~$\Sigma$,
the \emph{\iindex{simple saturation}} \iindex[not]{$\sclo(\Sigma)$} is
defined as the closure of all full rules in $\vnf(\hnf(\Sigma))$ under the
following inference rules:
\begin{itemize}
\item \emph{\comp}. We can resolve two full \tgds:

  Assume we have two full $\tau,\tau'\in\sclo(\Sigma)$ of the form\footnote{Update 25.03.2021: fixed typo: previous versions took $\tau,\tau'$ from $\Sigma$ rather than $\sclo(\Sigma)$; proofs were already done using the right definition.}
\begin{align*}
	\body(\vec x)&\rightarrow \head(\vec x),\\
        \body'(\vec z)&\rightarrow \head'(\vec z)
\end{align*}
that use distinct variables (otherwise, perform a renaming).
Suppose there is a unifier~$\theta$ of some atom in~$\head$ and an atom in~$\body'$
such that $\vec x\theta\cup\vec z\theta\subseteq\{x_1,\dotsc,x_{\hwidth(\hnf(\Sigma))}\}$.
Then we add the full $\vnf\bigl(\body\theta\cup(\body'\theta\setminus \head\theta) \rightarrow \head'\theta\bigr)$.

\item \emph{\orig}.
We can resolve a non-full with a full \tgd if this generates a full \tgd:

Assume we have a non-full $\tau\in\hnf(\Sigma)$ and a full $\tau'\in\sclo(\Sigma)$ of the form
\begin{align*}
        \body(\vec x)&\rightarrow \exists \vec y\,\head(\vec x,\vec y),\\
        \body'(\vec z)&\rightarrow \head'(\vec z)
\end{align*}
that use distinct variables (otherwise, perform a renaming).
Suppose there is a unifier~$\theta$ of some atoms in~$\head$ with some atoms in~$\body'$ such that
\begin{enumerate}
\item $\theta$ is the identity on $\vec y$,
\item $\theta$ only maps to variables,
\item $\vars\bigl(\body'\theta\setminus \head\theta\bigr)\subseteq\vec x\theta$, and
\item $\vec x\theta\cap\vec y=\emptyset$.\footnote{Update 25.03.2021: this fourth constraint was missing in previous versions of \orig; the proofs were not affected by this change.}
\end{enumerate}
Define
\begin{equation*}
	\head''\coloneqq\{H'\in\head'\mid \vars(H'\theta)\cap\vec y=\emptyset\}.
\end{equation*}
Then, if $\head''$ is non-empty, we add the full $\vnf\bigl(\body\theta\cup(\body'\theta\setminus \head\theta) \rightarrow \head''\theta\bigr)$.
\end{itemize}
\end{defn}

As we will derive another, more efficient rewriting for \gtgds,
we only state the key lemmas needed to prove the correctness of the simple saturation here.
The full proofs can be found in \cref{apx:simple_saturation}.
In the following, let $n$ be the number of predicate symbols in~$\Sigma$,~$a$ be the maximum arity of any predicate,
and let $w$ denote the width of $\hnf(\Sigma)$.

\begin{restatable}{lem}{simplesatresprops}\label{lem:gtgd_simple_closure_resolvent_properties}
Any $\sigma\in\sclo(\Sigma)$ is a full \tgd in variable normal form with $\width(\sigma)\leq w$.
\end{restatable}

The next properties will imply evolve-completeness of $\sclo(\Sigma)$.
\begin{restatable}{lem}{simplesatclosureprops}\label{lem:gtgd_simple_closure_props}
$\sclo(\Sigma)$ has the following properties:
\begin{thmlist}
\item\label{lem:gtgd_simple_closure_props_1}
For any $\tau\in\hnf(\Sigma)$, we have $\vnf(\tau)\in\sclo(\Sigma)$.

\item\label{lem:gtgd_simple_closure_props_2}
Let $I_0,I_1,I_2$ be a chase from $\sclo(\Sigma)$
and $\tau,\tau'\in\sclo(\Sigma)$ be full such that
\begin{assmlist}
\item $|\consts(I_0)|\leq\hwidth(\hnf(\Sigma))$,
\item firing $\tau$ in $I_0$ creates $I_1$,
\item firing $\tau'$ in $I_1$ creates $I_2$.
\end{assmlist}
Then there exists $\tau^*\in\sclo(\Sigma)$ such that firing $\tau^*$ in $I_0$ creates $I_0\cup(I_2\setminus I_1)$.

\item\label{lem:gtgd_simple_closure_props_3}
Let $I_0,I_1,I_2$ be a chase,
$\tau\in\hnf(\Sigma)$ be non-full,
and ${\tau'\in\sclo(\Sigma)}$ be full such that
\begin{assmlist}
\item firing $\tau$ in the root $r$ of $I_0$ creates $I_1$ and a new node $v$,
\item firing $\tau'$ in $F_v^1$ creates $I_2$,
\item\label{assm:gtgd_simple_closure_props_3} $F_r^2\setminus I_0\neq\emptyset$.
\end{assmlist}
Then there exists $\tau^*\in\sclo(\Sigma)$ such that firing $\tau^*$ in $I_0$ creates $F_r^2$.
\end{thmlist}
\end{restatable}
\restatablespacefix

\begin{restatable}[Completeness]{prop}{simplesatclosureevolvecomplete}\label{prop:gtgd_simple_closure_evolve_complete}
$\sclo(\Sigma)$ is evolve-complete \wrt $\Sigma$.
\end{restatable}
\restatablespacefix

\begin{restatable}{lem}{simplesatclosuresize}\label{cor:gtgd_simple_closure_size}
$|\sclo(\Sigma)|\leq 2^{2nw^a}$.
\end{restatable}
\restatablespacefix
\begin{restatable}{lem}{simplesatcomplexity}\label{lem:gtgd_simple_closure_complexity}
The simple saturation algorithm terminates in
$\ptime$ for bounded $w,n,a$,
$\exptime$ for bounded $a$, and $\texptime$ otherwise.
\end{restatable}
\restatablespacefix

\begin{restatable}{thm}{simplesatdecproc}
The simple saturation algorithm provides a decision procedure for the \qfqaprobgtgds.
The procedure terminates in
$\ptime$ for fixed $\Sigma$,
$\exptime$ for bounded $a$, and
$\texptime$ otherwise.
\end{restatable}

We now highlight a serious weak spot of the simple saturation,
which then motivates us to derive an improved rewriting algorithm for \gtgds in \cref{sec:gtgd_guarded_saturation}.
\begin{exmpl}\label{exmpl:unif_issue_sclo}
Given the database $\db=\{R(c,d),S(c)\}$ and a set of \gtgds $\Sigma$
\begin{align*}
\tau_1&=R(x_1,x_2)\rightarrow\exists y_1,y_2\bigl( S(x_1,x_2,y_1,y_2)\land T(x_1,x_2,y_2)\bigr),\\
\tau_2&=S(x_1,x_2,x_3,x_4)\rightarrow U(x_4),\\
\tau_3&=T(z_1,z_2,z_3)\land U(z_3)\rightarrow P(z_1),
\end{align*}
it is clear that $\db,\Sigma\models P(c)$.
Let us consider the steps performed by the simple saturation.

In the first iteration, \orig does not create any new rules since the composition of $\tau_1,\tau_2$
does not contain head atoms without existential variables and composing $\tau_1,\tau_3$ creates an existential variable in $U(z_3)$.
Taking the \comp-rule, we can clearly compose $\tau_2,\tau_3$;
however, it is not so clear which unifier $\theta$
one should consider for this.
One natural way is to restrict~$\theta$ to be an mgu of chosen atoms.
In our case, we obtain the mgu $\theta=[x_4/z_3]$, and
the resulting resolvent corresponds to $S(x_1,x_2,x_3,x_4)\land T(z_1,z_2,x_4)\rightarrow P(z_1)$.
But this rule contains more than $\hwidth(\Sigma)=4$ variables and hence will not be derived by \comp.
Eliminating the upper bound on the number of variables is no solution as this restriction
is needed for termination purposes.
Instead, every possible unifier that keeps this upper bound
must be considered, leading to a series of derived rules including
\begin{align*}
S(x_1,x_2,x_3,x_4)\land T(x_1,x_2,x_4)&\rightarrow P(x_1),&
S(x_1,x_2,x_3,x_4)\land T(x_2,x_1,x_4)&\rightarrow P(x_2),\\
S(x_1,x_2,x_3,x_4)\land T(x_1,x_3,x_4)&\rightarrow P(z_1),&
S(x_1,x_2,x_3,x_4)\land T(x_3,x_1,x_4)&\rightarrow P(x_3),
\end{align*}
and so forth, which is infeasible in practice.
\end{exmpl}
We hence continue our journey to find a more efficient rewriting algorithm for \gtgds.

%% file: figures/exmpl_tree_chase_seq_sclo.tex
\begin{tikzpicture}
\tikzset{node distance=0.5cm,
	r-0/.pic={\node (r) {$R(c),S(c)$};},
	I0/.pic={\pic[trigger]{r-0};},
	cl-1/.pic={\node (cl) {$T(c,d_1,d_2),R(c),S(c)$};},
	I1/.pic={%
		\pic{r-0};
		\pic[trigger,below=of r]{cl-1};
		\path (r) edge (cl);
	},
	cll-2/.pic={\node (cll) {$U(c,d_1,d_3),R(c),S(c)$};},
	I2/.pic={%
		\pic{r-0};
		\pic[below=of r]{cl-1};
		\pic[trigger,below=of cl]{cll-2};
		\path (r) edge (cl)
		(cl) edge (cll);
	},
	cll-3/.pic={\node[align=center] (cll) {$U(c,d_1,d_3),R(c),S(c),P(c),V(c,d_1)$};},
	cl-3/.pic={\node (cl) {$T(c,d_1,d_2),R(c),S(c),P(c),V(c,d_1)$};},
	r-3/.pic={\node (r) {$R(c),S(c),P(c)$};},
	I3/.pic={%
		\pic{r-3};
		\pic[trigger,below=of r]{cl-3};
		\pic[below=of cl]{cll-3};
		\path (r) edge (cl)
		(cl) edge (cll);
	},
	cll-4/.pic={\node[align=center] (cll) {$U(c,d_1,d_3),R(c),S(c),P(c),V(c,d_1),M(c)$};},
	cl-4/.pic={\node (cl) {$T(c,d_1,d_2),R(c),S(c),P(c),V(c,d_1),M(c)$};},
	r-4/.pic={\node (r) {$R(c),P(c),S(c),M(c)$};},
	I4/.pic={%
		\pic{r-4};
		\pic[below=of r]{cl-4};
		\pic[below=of cl]{cll-4};
		\path (r) edge (cl)
		(cl) edge (cll);
	},
}
\node[scope] (I0) {
	\begin{tikzpicture}
		\pic{I0};
		\node[lnode,above=0.1cm of r]{$I_0$};
		\node[lnode,left=0.1cm of r]{$r$};
	\end{tikzpicture}
};
\node[scope,right=of I0] (I1) {
	\begin{tikzpicture}
		\pic{I1};
		\node[lnode,above=0.1cm of r]{$I_1$};
		\node[lnode,left=0.1cm of r]{$r$};
		\node[lnode,left=0.1cm of cl]{$c$};
	\end{tikzpicture}
};
\pic{drawv=I0/I1/I3/I3};
\node[scope,right=of I1] (I2) {
	\begin{tikzpicture}
		\pic{I2};
		\node[lnode,right=of r]{$I_2$};
		\node[lnode,left=0.1cm of r]{$r$};
		\node[lnode,left=0.1cm of cl]{$c$};
		\node[lnode,left=0.1cm of cll]{$l$};
	\end{tikzpicture}
};
\pic{drawvr=I1/I2};
\node[scope,below=1.5cm of I0,xshift=1cm] (I3) {
	\begin{tikzpicture}
		\pic{I3};
		\node[lnode,right=of r]{$I_3$};
		\node[lnode,left=0.1cm of r]{$r$};
		\node[lnode,left=0.1cm of cl]{$c$};
		\node[lnode,left=0.1cm of cll]{$l$};
	\end{tikzpicture}
};
\node[scope,right=of I3] (I4) {
	\begin{tikzpicture}
		\pic{I4};
		\node[lnode,right=of r]{$I_4$};
		\node[lnode,left=0.1cm of r]{$r$};
		\node[lnode,left=0.1cm of cl]{$c$};
		\node[lnode,left=0.1cm of cll]{$l$};
	\end{tikzpicture}
};
\pic{drawvr=I3/I4};
\pic{drawh=I2/I3/I3/I4};
\end{tikzpicture}

%% file: pages/gtgd/gtgd_guarded_saturation.tex
\section{The Guarded Saturation}\label{sec:gtgd_guarded_saturation}

Again, we describe the basic ideas of the upcoming algorithm by means of an example.
\begin{exmpl}\label{exmpl:tree_chase_seq_gclo}
Let us revisit the tree-like chase from \cref{exmpl:tree_chase_seq_sclo}, depicted in \cref{fig:exmpl_tree_chase_seq_sclo}, that uses the \gtgds
\begin{align*}
\tau_1&=R(x_1)\rightarrow\exists y_1,y_2\,T(x_1,y_1,y_2),&
\tau_2&=T(x_1,x_2,x_3)\rightarrow\exists y\,U(x_1,x_2,y),\\
\tau_3&=U(x_1,x_2,x_3)\rightarrow P(x_1)\land V(x_1,x_2),&
\tau_4&=T(x_1,x_2,x_3)\land V(x_1,x_2)\land S(x_1)\rightarrow M(x_1),
\end{align*}
As before, we first lift the full $\tau_3$ by composing it with the non-full $\tau_2$
to obtain the full ${\tau_1'\coloneqq T(x_1,x_2,x_3)\rightarrow P(x_1)\land V(x_1,x_2)}$.
We then continue by composing the non-full $\tau_1$ with $\tau_1'$;
this time, however, we do not only derive the full $\tau_2'\coloneqq R(x_1)\rightarrow P(x_1)$ at
this step, but also add the non-full $\tau_1^*\coloneqq R(x_1)\rightarrow\exists y_1,y_2\bigl(T(x_1,y_1,y_2)\land V(x_1,y_1)\bigr)$.
The idea is to collect the derivable non-full head atoms so that we can
eliminate the \comp-step of our previous algorithm.
Indeed, instead of first discharging the premise $V(x_1,x_2)$ from $\tau_4$ by composing it with~$\tau_1'$,
we can now directly compose $\tau_1^*$ with $\tau_4$ to obtain the desired
full $\tau_4'\coloneqq R(x_1)\land S(x_1)\rightarrow M(x_1)$.
\end{exmpl}
\begin{rmk}
The unification issue of \cref{exmpl:unif_issue_sclo} now disappears: we first compose $\tau_1,\tau_2$ to obtain $\tau_1'\coloneqq R(x_1,x_2)\rightarrow\exists y_1,y_2\bigl( S(x_1,x_2,y_1,y_2)\land T(x_1,x_2,y_2)\land U(y_2)\bigr)$,
and then compose $\tau_1',\tau_3$ to obtain $\tau_2'\coloneqq R(x_1,x_2)\rightarrow P(x_1)$.
\end{rmk}

Using these observations, we now present our improved rewriting algorithm for \gtgds.
\begin{defn}[Guarded saturation] \label{defn:gtgd_guarded_closure}
We define the following inference rule:
\begin{itemize}
\item \emph{\evolve}. We can compose a non-full with a full \tgd:

Assume we have a non-full $\tau$ and a full $\tau'$ of the form
\begin{align*}
	\body(\vec x)&\rightarrow\exists\vec y\,\head(\vec x,\vec y),\\
	\body'(\vec z)&\rightarrow\head'(\vec z),
\end{align*}
and assume that $\tau,\tau'$ use distinct variables (otherwise, perform a renaming).
Suppose~$\theta$ is an mgu of atoms $S\coloneqq(H_1,\dotsc,H_n)$ and $S'\coloneqq(B_1',\dotsc,B_n')$ that is the identity on $\vec y$,\footnote{Such an mgu can be obtained by treating $\vec y$ as constants.} where $H_i\in\head,B_i'\in\body'$ for $1\leq i\leq n$.
Moreover, assume that the following \emph{\iindex{existential variable check} (\iindexsee{evc}{existential variable check})} is satisfied:

\medskip
Any atom $B'\in\beta'$ with $\vars(B'\theta)\cap\vec y\neq\emptyset$
occurs in $S'$
and $\vec x\theta\cap\vec y=\emptyset$.
\medskip

Now define
\begin{align*}
\beta''&\coloneqq\bigl(\beta\cup\bigl(\beta'\setminus\{B_1',\dotsc,B_n'\}\bigr)\bigr)\theta,\\
\head''&\coloneqq(\head\cup\head')\theta.
\end{align*}
We then add the rule(s)
$\vnf(\hnf(\body''\rightarrow\exists\vec y\,\head''))$.
\end{itemize}
Given a set of \gtgds~$\Sigma$, let \iindex[not]{$\gcloex(\Sigma)$} be the closure
of $\vnf(\hnf(\Sigma))$ under the given inference rule.
The \emph{\iindex{guarded saturation}} \iindex[not]{$\gclo(\Sigma)$} is obtained from
$\gcloex(\Sigma)$ by dropping all non-full rules.
\end{defn}

Before we verify the correctness of the saturation,
we derive two properties that follow from our restrictions imposed on \evolve.
\begin{lem}\label{lem:gtgd_guarded_guard_included}
Assume we perform an \evolve inference on $\tau$ and $\tau'$, $\tau$ is in head
normal form,
and $\tau'$ is guarded by $G'$.
Then $G'$ occurs in $S'$.
\end{lem}
\begin{proof}
Let $B'$ be an atom in $S'$ that unifies with some $H$ in $S$.
As $\tau$ is in head normal form and $\theta$ is the identity on $\vec y$,
we derive that $H\theta$ contains some $y_i$.
As $B'$ unifies with $H$,
we get $\vars(B'\theta)\cap\vec y\neq\emptyset$.
Finally, since $G'$ guards $\tau'$, we get
$\vars(B'\theta)\subseteq\vars(G'\theta)$ and hence $\vars(G'\theta)\cap\vec y\neq\emptyset$.
Thus, by the evc, $G'$ occurs in $S'$.
\end{proof}

\begin{prop}\label{prop:gtgd_guarded_closure_resolvent_properties}
Let $w_b\coloneqq\bwidth(\hnf(\Sigma))$, and $w_h\coloneqq\hwidth(\hnf(\Sigma))$.
Then any $\sigma\in\gcloex(\Sigma)$ is a \gtgd in variable normal form
and head normal form, $\bwidth(\sigma)\leq w_b$, and
$\hwidth(\sigma)\leq w_h$.
\end{prop}
\begin{proof}
We prove the claim inductively.
Clearly, all rules in $\vnf(\hnf(\Sigma))$ satisfy the conditions
and all derived rules will be in variable and head normal form due
to the final $\vnf(\hnf(\cdot))$ application.

Assume we create $\sigma\in\vnf(\hnf(\body''\rightarrow\exists\vec y\,\head''))$ by \evolve from
some non-full ${\tau=\body\rightarrow\exists\vec y\,\head}$ and
full $\tau'=\body'\rightarrow\head'$ in $\gcloex(\Sigma)$.
By the inductive hypothesis, $\tau$ is in head normal form and $\tau'$ is guarded by some $G'$.
Hence, by \cref{lem:gtgd_guarded_guard_included}, $G'$ occurs in $S'$.
Let $H$ be the atom in $S$ that is unified with $G'$ using $\theta$.
By \cref{lem:mgu_function_free}, $\theta$ does not introduce functions in $\sigma$.
As $G'$ guards $\tau'$, $H\theta=G'\theta$, and $\theta$ is the identity on $\vec y$,
we get
\begin{equation*}
\vars(\head'\theta)\subseteq\vars(\body'\theta)=\vars(G'\theta)=\vars(H\theta)
\subseteq\vars(\head\theta)\subseteq\vars(\body\theta)\cup\vec y.
\end{equation*}
From this, we can derive that $\vars(\body'')\subseteq\vars(\body\theta)\cup\vec y$
and $\vars(\head'')=\vars(\head\theta)$.
Since the evc is satisfied, the former refines to $\vars(\body'')=\vars(\body\theta)$.
By the inductive hypothesis, we have $\bwidth(\tau)\leq w_b$ and $\hwidth(\tau)\leq w_h$, and hence
$\bwidth(\sigma)\leq w_b$ and $\hwidth(\sigma)\leq w_h$.
Finally, again by the inductive hypothesis, $\tau$ is guarded by some $G$,
which will also be a guard for $\sigma$.
\end{proof}

\input{pages/gtgd/gtgd_guarded_correctness}
\input{pages/gtgd/gtgd_guarded_complexity}

%% file: pages/gtgd/gtgd_guarded_correctness.tex
\subsection{Correctness}
To prove the correctness of our algorithm, we have to verify that our rewriting is sound and complete.
We begin with the former.
\begin{lem}\label{lem:gtgd_guarded_closure_ex_soundness}
$\gcloex(\Sigma)$ is sound \wrt $\Sigma$.
\end{lem} 
\begin{proof}
By \cref{lem:tgd_hnf_closure_coincide}, it suffices to show that $\gcloex(\Sigma)$ is sound \wrt $\hnf(\Sigma)$.
For this, it suffices to show that $\hnf(\Sigma)\models \gcloex(\Sigma)$.
Assume we derive $\sigma\in\vnf(\hnf(\beta''\rightarrow\exists\vec y\,\head''))$ from some $\tau=\beta(\vec x)\rightarrow\exists \vec y\,\head$ and $\tau'=\beta'(\vec z)\rightarrow\head'$ using \evolve.
Let $H_1,\dotsc,H_n$ and $B_1',\dotsc,B_n'$ be the resolved atoms in~$\head$ and $\body'$, respectively,
let $\theta$ be the used mgu,
and let $\calA$ be a structure with $\calA\models\tau,\tau'$.
Again by \cref{lem:tgd_hnf_closure_coincide}, it suffices to show that $\calA\models\sigma'$,
where $\sigma'\coloneqq\beta''\rightarrow\exists\vec y\,\head''$.
Without loss of generality, assume $\sigma'$ is in variable normal form.
Then by \cref{prop:gtgd_guarded_closure_resolvent_properties} (and its proof), we know that
\begin{propylist}
\item $\sigma'$ is function-free (as it is a \gtgd),
\item $x_i\theta=x_{j_i}$ for any $x_i\in\vec x$, and
\item $z_i\theta\in\vec x$ or $z_i\theta\in\vec y$ for any $z_i\in\vec z$.
\end{propylist} 
Let $z_{l_1},\dotsc,z_{l_m}$ be all variables in $\vec z$ with $z_{l_i}\theta=x_{k_i}$,
and let $z_{l_{m+1}},\dotsc,z_{l_{m+m'}}$ be the remaining variables with $z_{l_{m+i}}\theta=y_{k_{m+i}}$.

Now assume $\calA\models\beta''[c_1/x_1,\dotsc,c_a/x_a]$ for some $c_1,\dotsc,c_a\in\dom(\calA)$.
Then by definition of $\beta''$, we get
\begin{align}
	\calA&\models\beta[c_{j_1}/x_1,\dotsc,c_{j_a}/x_a],\nonumber\\
	\calA&\models(\beta'\setminus\{B_1',\dotsc,B_n'\})[c_{k_1}/z_{l_1},\dotsc,c_{k_m}/z_{l_m}].\label{eq:gtgd_guarded_closure_ex_soundness_1}
\end{align} 
The former implies $\calA\models\exists\vec y\,\head[c_{j_1}/x_1,\dotsc,c_{j_a}/x_a]$, and thus
\begin{equation*}
\calA\models\head[c_{j_1}/x_1,\dotsc,c_{j_a}/x_a,d_1/y_1,\dotsc,d_{a'}/y_{a'}]
\end{equation*} 
for some $d_1,\dotsc,d_{a'}\in\dom(\calA)$.
This in return implies
\begin{equation}
\calA\models\head\theta[c_1/x_1,\dotsc,c_a/x_a,d_1/y_1,\dotsc,d_{a'}/y_{a'}].\label{eq:gtgd_guarded_closure_ex_soundness_2}
\end{equation} 
As $H_i\theta=B_i'\theta$, we get
\begin{equation*}
\calA\models B_i[c_{k_1}/z_{l_1},\dotsc,c_{k_m}/z_{l_m},d_{k_{m+1}}/z_{l_{m+1}},\dotsc,d_{k_{m+m'}}/z_{l_{m+m'}}]
\end{equation*} 
for $1\leq i\leq n$.
Together with \eqref{eq:gtgd_guarded_closure_ex_soundness_1}, we can thus derive that
\begin{equation*}
\calA\models\beta'[c_{k_1}/z_{l_1},\dotsc,c_{k_m}/z_{l_m},d_{k_{m+1}}/z_{l_{m+1}},\dotsc,d_{k_{m+m'}}/z_{l_{m+m'}}],
\end{equation*} 
and consequently
\begin{equation*}
\calA\models\head'[c_{k_1}/z_{l_1},\dotsc,c_{k_m}/z_{l_m},d_{k_{m+1}}/z_{l_{m+1}},\dotsc,d_{k_{m+m'}}/z_{l_{m+m'}}].
\end{equation*} 
This implies that
\begin{equation}\label{eq:gtgd_guarded_closure_ex_soundness_3}
\calA\models\head'\theta[c_1/x_1,\dotsc,c_a/x_a,d_1/y_1,\dotsc,d_{a'}/y_{a'}].
\end{equation} 
Now by \eqref{eq:gtgd_guarded_closure_ex_soundness_2} and \eqref{eq:gtgd_guarded_closure_ex_soundness_3}, we get
\begin{equation*}
\calA\models\head''[c_1/x_1,\dotsc,c_a/x_a,d_1/y_1,\dotsc,d_{a'}/y_{a'}],
\end{equation*} 
and thus $\calA\models\exists\vec y\,\head''[c_1/x_1,\dotsc,c_a/x_a]$.
\end{proof} 

\begin{cor}[Soundness]\label{cor:gtgd_guarded_closure_soundness}
$\gclo(\Sigma)$ is sound \wrt $\Sigma$.
\end{cor} 
\begin{proof}
By \cref{lem:gtgd_guarded_closure_ex_soundness}, $\gcloex(\Sigma)$ is sound \wrt $\Sigma$,
and hence, so is $\gclo(\Sigma)\subseteq\gcloex(\Sigma)$.
\end{proof}

The next lemma will be required multiple times during our completeness proof and complexity analysis.
\begin{lem}\label{lem:gtgd_delta_is_id_y}
Let $\theta,\theta'$ be substitutions that are the identity on $\vec y$,
and let $\delta$ be a substitution such that $\theta'=\theta\delta$.
Then $\delta$ is the identity on $\vec y$.
\end{lem} 
\begin{proof}
Assume otherwise.
Then there is $y_i\in\vec y$ with $y_i\delta\neq y_i$.
Now as $\theta,\theta'$ are the identity on $\vec y$, we get $y_i=y_i\theta'=y_i\theta\delta=y_i\delta\neq y_i$, a contradiction.
\end{proof} 

To show completeness, we will have to make use of the following properties.
\begin{lem}\label{lem:gtgd_guarded_closure_props}
$\gcloex(\Sigma)$ has the following properties:
\begin{thmlist}
\item\label{lem:gtgd_guarded_closure_props_1}
For every $\tau\in\hnf(\Sigma)$, we have $\vnf(\tau)\in\gcloex(\Sigma)$.

\item\label{lem:gtgd_guarded_closure_props_2} 
Let $I_0,I_1,I_2$ be a chase from $\gcloex(\Sigma)$,
$\tau$ be non-full, and ${\tau'}$ be full such that
\begin{assmlist}
\item firing $\tau$ in the root $r$ of $I_0$ creates $I_1$ and a new node $v$,
\item firing $\tau'$ in $F_v^1$ creates $I_2$,
\end{assmlist} 
Then there exists a non-full $\tau_\exists^*\in\gcloex(\Sigma)$ such that firing 
$\tau_\exists^*\in\gcloex(\Sigma)$ in $I_0$ creates $F_v^2\setminus F_r^2$.
Furthermore, if $F_r^2\neq I_0$, then there exists a full $\tau^*\in\gcloex(\Sigma)$ 
such that firing $\tau^*$ in $I_0$ creates $F_r^2\setminus I_0$.
\end{thmlist}
\end{lem}
\begin{proof}
The first property follows from the base case of the algorithm.
For the second property, without loss of generality, assume ${\consts(I_0)=\vec c}$
and $\consts(F_v^1)=\vec c\cup\vec d$,
where ${\vec c_\restriction\subseteq\vec c}$, and $\vec d$ are fresh constants.
We have 
\begin{align*}
\tau&=\beta(\vec x)\rightarrow\exists\vec y\,\head(\vec x,\vec y),\\
\tau'&=\beta'(\vec z)\rightarrow\head'(\vec z),
\end{align*} 
and $\tau$ is in head normal form by \cref{prop:gtgd_guarded_closure_resolvent_properties}.
Let $h,h'$ be the triggers used for $\tau,\tau'$.
If $h'(\body')\subseteq I_0$, we simply set $\tau_\exists^*\coloneqq\tau$ and $\tau^*\coloneqq\tau'$;
If, however, the chase step makes use of some facts in $I_1\setminus I_0$, 
let $B_1',\dotsc,B_n'$ be all atoms in $\body'$ such that
for each $1\leq i\leq n$, there is $H_i\in\head$ with $h(H_i)=h'(B_i')$.
Then the substitution $\theta$ defined by
\begin{equation*}
	v\theta\coloneqq\begin{cases}
		x_i,&\text{if }v\in\vec x \text{ and }h(v)=c_i\\
		x_i,&\text{if }v\in\vec z \text{ and }h'(v)=c_i\\
		y_i,&\text{if }v\in\vec z \text{ and }h'(v)=h(y_i)
	\end{cases}
\end{equation*} 
is a unifier of $H_1,\dotsc,H_n$ and $B_1',\dotsc,B_n'$ 
that is the identity on $\vec y$ and $\vec x\theta\cap\vec y=\emptyset$. 
Hence, there is an mgu $\theta^*$ of $H_1,\dotsc,H_n$ and $B_1',\dotsc,B_n'$ 
that is the identity on $\vec y$ (by treating $\vec y$ as constants).
Thus, there is a substitution $\delta$ such that $\theta^*=\theta\delta$.
By \cref{lem:gtgd_delta_is_id_y}, $\delta$ is the identity on $\vec y$.
\begin{claim}\label{claim:gtgd_guarded_closure_props_2_2}
$\vec x\theta^*\cap\vec y=\emptyset$.
\end{claim} 
\begin{subproof}
Assume otherwise.
Then for some $x_i\in\vec x$, we get $x_i\theta^*=y_j$,
and thus ${x_i\theta=x_i\theta^*\delta=y_j\delta\notin\vec y}$ since $\vec x\theta\cap\vec y=\emptyset$,
contradicting the fact that $\delta$ is the identity on $\vec y$.
\end{subproof} 
\begin{claim}\label{claim:gtgd_guarded_closure_props_2_3}
Any $B'\in\body'$ with $\vars(B'\theta^*)\cap\vec y\neq\emptyset$ occurs in $B_1',\dotsc,B_n'$.
\end{claim} 
\begin{subproof}
If $\vars(B'\theta^*)\cap\vec y\neq\emptyset$,
then also $\vars(B'\theta)=\vars(B'\theta^*\delta)\cap\vec y\neq\emptyset$ as $\delta$ is the identity on $\vec y$.
Thus, by construction of $\theta$, $\consts(h'(B'))\cap\vec d\neq\emptyset$, and hence,
$h'(B')\in I_1\setminus I_0\subseteq h(\head)$.
%
\end{subproof} 
Using \cref{claim:gtgd_guarded_closure_props_2_2,claim:gtgd_guarded_closure_props_2_3},
we can derive that $\theta^*$ satisfies the evc.
Now define
\begin{align*}
	\beta''&\coloneqq\bigl(\beta\cup\bigl(\beta'\setminus\{B_1',\dotsc,B_n'\}\bigr)\bigr)\theta^*,\\
	\head_\restriction'&\coloneqq\{H'\in\head'\mid \vars(H'\theta^*)\cap\vec y=\emptyset\}.
\end{align*} 
Then using \evolve, we get the non-full 
\begin{equation*}
\tau_\exists^*=\bigl(\body^*\rightarrow\exists\vec y\,\head_\exists^*\bigr)=\vnf\bigl(\body''\rightarrow\exists\vec y\,\bigl(\head\cup(\head'\setminus\head_\restriction')\bigr)\theta^*\bigr),
\end{equation*} 
and, if $\head_\restriction'\neq\emptyset$, the full
\begin{equation*}
\tau^*=(\body^*\rightarrow\head^*)=\vnf\bigl(\body''\rightarrow\head_\restriction'\theta^*\bigr).
\end{equation*} 
Let $\rho$ be the renaming used by $\vnf(\cdot)$ for $\tau_\exists^*$.
Note that $\rho$ is the identity on $\vec y$ as $\tau$ is in variable normal form and $\theta^*$ is the identity on $\vec y$.
Thus, $\rho$ is also the renaming used by $\vnf(\cdot)$ for $\tau^*$.
Let~$h''$ be the trigger defined by $h''(x_i)\coloneqq c_i$ and $h''(y_i)=h(y_i)$,
and set $h^*\coloneqq h''\circ\delta\circ\rho^{-1}$.
\begin{claim}\label{claim:gtgd_guarded_closure_props_2_4}
For any $v\in\vec x\cup\vec y$, we have $h''(v\theta)=h(v)$, and for any $z_i\in\vec z$, we have $h''(z_i\theta)=h'(z_i)$. 
\end{claim}
\begin{subproof}
Pick some $x_i\in\vec x$ and assume $h(x_i)=c_j$. 
Then $x_i\theta=x_j$ by definition of $\theta$.
Further, by definition of $h''$, $h''(x_j)=c_j$.
Hence, $h''(x_i\theta)=h''(x_j)=c_j=h(x_i)$.

Now pick some $y_i\in\vec y$. 
Then $y_i\theta=y_i$, and by definition of $h''$, $h''(y_i)=h(y_i)$.
Hence, $h''(y_i\theta)=h''(y_i)=h(y_i)$.

Lastly, pick some $z_i\in\vec z$.
The case $h'(z_i)=c_j$ is analogous to the first case.
In case $h'(z_i)=d_j$, there must be $y_k\in\vec y$ with $h(y_k)=d_j$.
Then $z_i\theta=y_k$ by definition of $\theta$.
Hence, $h''(z_i\theta)=h''(y_k)=h(y_k)=d_j=h'(z_i)$.
\end{subproof} 
\begin{claim}\label{claim:gtgd_guarded_closure_props_2_5}
$h^*(\body^*)\subseteq I_0$.
\end{claim} 
\begin{subproof}
We have 
\begin{align*}
h^*(\body^*)&=h''\circ\delta\circ\rho^{-1}(\body^*)=h''\circ\delta(\body'')
=h''\Bigl(\bigl(\beta\cup\bigl(\beta'\setminus\{B_1',\dotsc,B_n'\}\bigr)\bigr)\theta^*\delta\Bigr)\\
&=h''\Bigl(\bigl(\beta\cup\bigl(\beta'\setminus\{B_1',\dotsc,B_n'\}\bigr)\bigr)\theta\Bigr)
=h''\Bigl(\beta\theta\cup\bigl(\beta'\setminus\{B_1',\dotsc,B_n'\}\bigr)\theta\Bigr)\\
&=h''(\beta\theta)\cup h''\bigl(\bigl(\beta'\setminus\{B_1',\dotsc,B_n'\}\bigr)\theta\bigr)
\stackrel{\cref{claim:gtgd_guarded_closure_props_2_4}}{=}h(\beta)\cup h'\bigl(\beta'\setminus\{B_1',\dotsc,B_n'\}\bigr)\\
&\subseteq I_0\cup (I_1\setminus h(\head))=I_0,
\end{align*} 
where the last step follows from the choice of $B_1',\dotsc,B_n'$; that is,
if $h'(B')\in I_1\setminus I_0\subseteq h(\head)$ for some $B'\in\body'$, then $B'$ occurs in $B_1',\dotsc,B_n'$.
\end{subproof} 
\begin{claim}\label{claim:gtgd_guarded_closure_props_2_6}
$h^*(\head_\exists^*)=F_v^2\setminus F_r^2$.
\end{claim} 
\begin{subproof}
We have 
\begin{align*}
h^*(\head_\exists^*)&=h''\circ\delta\circ\rho^{-1}(\head_\exists^*)
=h''\Bigl(\bigl(\head\cup(\head'\setminus\head_\restriction')\bigr)\theta^*\delta\Bigr)
=h''\Bigl(\bigl(\head\cup(\head'\setminus\head_\restriction')\bigr)\theta\Bigr)\\
&=h''(\head\theta)\cup h''((\head'\setminus\head_\restriction')\theta)
\stackrel{\cref{claim:gtgd_guarded_closure_props_2_4}}{=}h(\head)\cup h'(\head'\setminus\head_\restriction')
=(I_1\setminus I_0)\cup h'(\head'\setminus\head_\restriction').
\end{align*} 
Now for any $H'\in\head'$, we have
$H'\in(\head'\setminus\head_\restriction')$ \ifftext $\vars(H'\theta^*)\cap\vec y\neq\emptyset$ 
\ifftext $\consts(h'(H'))\cap\vec d\neq\emptyset$ \ifftext $h'(H')\notin F_r^2$,
and hence
\begin{equation*}
(I_1\setminus I_0)\cup h'(\head'\setminus\head_\restriction')=(I_1\setminus I_0)\cup \bigl((I_2\setminus I_1)\setminus F_r^2\bigr)
=F_v^2\setminus F_r^2.\qedhere
\end{equation*} 
\end{subproof} 
\begin{claim}\label{claim:gtgd_guarded_closure_props_2_7}
$h^*(\head^*)=F_r^2\setminus I_0$.
\end{claim} 
\begin{subproof}
We have 
\begin{equation*}
h^*(\head^*)=h''\circ\delta\circ\rho^{-1}(\head^*)=h''(\head_\restriction'\theta^*\delta)=h''(\head_\restriction'\theta)
\stackrel{\cref{claim:gtgd_guarded_closure_props_2_4}}{=}h'(\head_\restriction').
\end{equation*} 
Now for any $H'\in\head'$, we have
$H'\in\head_\restriction'$ \ifftext $\vars(H'\theta^*)\cap\vec y=\emptyset$ 
\ifftext $\consts(h'(H'))\cap\vec d=\emptyset$ \ifftext $h'(H')\in F_r^2$,
and hence
\begin{equation*}
h'(\head_\restriction')=(I_2\setminus I_1)\cap F_r^2=F_r^2\setminus I_0.\qedhere
\end{equation*} 
\end{subproof} 
Using \cref{claim:gtgd_guarded_closure_props_2_5,claim:gtgd_guarded_closure_props_2_6,claim:gtgd_guarded_closure_props_2_7},
we derive that firing $\tau_\exists^*$ based on $h^*$ in $I_0$ creates $F_v^2\setminus F_r^2$ and
firing~$\tau^*$ based on $h^*$ in $I_0$ creates $F_r^2\setminus I_0$.
\end{proof}

\begin{prop}[Completeness]\label{prop:gtgd_guarded_closure_evolve_complete}
$\gclo(\Sigma)$ is evolve-complete \wrt $\Sigma$.
\end{prop} 
\begin{proof}
To show \cref{defn:gtgd_abstr_completeness_props_1}, we use \cref{lem:gtgd_guarded_closure_props_1}
to see that $\vnf(\tau)\in\gcloex(\Sigma)$ for every full $\tau\in\hnf(\Sigma)$, 
and hence, $\vnf(\tau)\in\gclo(\Sigma)$.

For \cref{defn:gtgd_abstr_completeness_props_induct}, 
we additionally prove that there is a non-full $\tau\in\gcloex(\Sigma)$
such that firing $\tau_\exists$ in $F_r^n$ 
creates $F_v^{n+1}\setminus F_r^{n+1}$.
We prove the claim by induction on $n$. 
For $n=0$, we can simply set $\tau_\exists\coloneqq\tau_r$ using that $\tau_r$ is in head normal form.

If $n>0$, by the inductive hypothesis, we get a chase $I_0',\dotsc,I_m'$ from $I_0$ and $\gclo(\Sigma)$ of $F_r^n$
and a non-full $\tau_\exists\in\gcloex(\Sigma)$ such that firing $\tau_\exists$ in $F_r^{n-1}$ creates $F_v^n\setminus F_r^n$.
Let $F_r^n,I_1'',I_2''$ be the chase obtained by firing $\tau_\exists$ in $F_r^n$ and $\tau_n$ in $I_1''$.
Then we can apply \cref{lem:gtgd_guarded_closure_props_2} on $F_r^n,I_1'',I_2''$
to obtain a non-full $\tau_\exists'\in\gcloex(\Sigma)$ and,
if $F_r^{n+1}\setminus F_r^n\neq\emptyset$,
a full $\tau'\in\gclo(\Sigma)$ such that firing $\tau_\exists'$ in $F_r^n$ creates
$F_v^{n+1}\setminus F_r^{n+1}$ and firing $\tau'$ in $F_r^n$ creates $F_r^{n+1}\setminus F_r^n$.
Hence, we can fire $\tau'$ in $I_m'$ to create $F_r^{n+1}$.
\end{proof} 

%% file: pages/gtgd/gtgd_guarded_complexity.tex
\subsection{Complexity Analysis}
In this section, we discuss the computational aspects and complexity of the algorithm.
We first derive some useful properties that we can exploit in our computations.

\begin{lem}\label{lem:gtgd_guarded_permutation_independent}
\evolve inferences are permutation independent of $S,S'$.
More precisely, for any $S=H_1,\dotsc,H_n$, $S'=B_1',\dotsc,B_n'$, and permutation
$\sigma\in S_n$, 
where $S_n$ is the symmetric group with n elements,
the rules derived using \evolve under $S,S'$ 
and $H_{\sigma(1)},\dotsc,H_{\sigma(n)}$, $B_{\sigma(1)}',\dotsc,B_{\sigma(n)}'$ coincide.
\end{lem} 
\begin{proof}
The evc does not consider the order of the sequences and the mgus of permuted sequences coincide.
\end{proof} 

\begin{lem}\label{lem:gtgd_guarded_evc_impl}
Assume we perform an \evolve inference on $\tau,\tau'\in\gcloex(\Sigma)$.
Then $\vars(B'\theta)\cap\vec y\neq \emptyset$ for any $B'$ in $S'$.
\end{lem} 
\begin{proof}
By \cref{prop:gtgd_guarded_closure_resolvent_properties}, $\tau$ is in $\hnf$. 
Using that $\theta$ is the identity on $\vec y$, 
we get that any atom in $\head\theta$ contains some $y_i$.
As any $B'$ in $S'$ unifies 
with some atom in $\head$ using $\theta$, 
this implies $\vars(B'\theta)\cap\vec y\neq\emptyset$.
\end{proof} 

\begin{cor}\label{cor:gtgd_guarded_evc_iff}
Assume we perform an \evolve inference on $\tau,\tau'\in\gcloex(\Sigma)$.
Then $B'$ occurs in $S'$ \ifftext $\vars(B'\theta)\cap\vec y\neq\emptyset$.
\end{cor} 
\begin{proof}
Left to right is \cref{lem:gtgd_guarded_evc_impl}.
Right to left follows from the evc.
\end{proof} 

Although we know that any guard $G'$ of $\tau'$ must occur in $S'$ by \cref{lem:gtgd_guarded_guard_included},
it is not yet clear how one determines the remaining atoms of $S,S'$ needed for an \evolve inference.
The next proposition (and its proof) will be a useful helper in this regard. 

\begin{prop}\label{prop:gtgd_guarded_unique_sequence}
Assume we want to use \evolve on $\tau=\beta(\vec x)\rightarrow\exists\vec y\,\head$ and
${\tau'=\beta'\rightarrow\head'}$ in $\gcloex(\Sigma)$ while unifying 
the guard $G'\in\body'$ with some $H\in\head$.
Then there exists at most one sequence $S'$ (up to permutation)
such that we can satisfy all restrictions imposed on \evolve.
\end{prop} 
\begin{proof}
Assume some $S,S',\theta$ satisfy all restrictions for an \evolve inference 
and $\theta$ unifies $G'$ with $H$.
Due to \cref{lem:gtgd_guarded_permutation_independent}, we can assume that
$S=H,H_2,\dotsc,H_n$ and $S'=G',B_2',\dotsc,B_n'$.
Since $H\theta=G'\theta$ and $\theta$ is the identity on $\vec y$, 
there must be an mgu of~$G'$ and~$H$ that is the identity on $\vec y$ (by treating $\vec y$ as constants).
Fix such an mgu, call it~$\theta^*$.
Then there is $\delta$ such that $\theta=\theta^*\delta$,
and by \cref{lem:gtgd_delta_is_id_y}, $\delta$ must be the identity on~$\vec y$.
\begin{claim}\label{claim:gtgd_guarded_unique_sequence_1}
For any $v\in\vars(\body')$, we have $v\theta\in\vec y$ \ifftext $v\theta^*\in\vec y$.
\end{claim} 
\begin{subproof}
Since $G'$ guards $\tau'$, we have $v\theta^*\in\vars(G'\theta^*)=\vars(H\theta^*)\subseteq\vec x\theta^*\cup \vec y$.
Now if $v\theta^*\in\vec y$, then also $v\theta^*\delta=v\theta\in\vec y$ as $\delta$ is the identity on $\vec y$.
If $v\theta^*\in\vec x\theta'$, then $v\theta^*\delta\in\vec x\theta^*\delta=\vec x\theta$,
and as $\vec x\theta\cap\vec y=\emptyset$, this implies $v\theta^*\delta=v\theta\notin\vec y$.
\end{subproof} 
By \cref{cor:gtgd_guarded_evc_iff}, $B'$ occurs in $S'$
\ifftext $\vars(B'\theta)\cap\vec y\neq\emptyset$.
Thus, by \cref{claim:gtgd_guarded_unique_sequence_1}, we can exactly obtain $S'$ (up to permutation) by selecting all atoms $B'\in\body'$ with
$\vars(B'\theta^*)\cap\vec y\neq\emptyset$.
Now $\theta^*$ was chosen arbitrarily, and hence, the choice of $S'$ is unique once we fix $G'$ and~$H$.
\end{proof} 

Let us continue with the complexity of our saturation algorithm.
In the following, let~$n$ be the number of predicate symbols in $\Sigma$, 
$a$ be the maximum arity of any predicate in $\Sigma$,
$w_b\coloneqq\bwidth(\hnf(\Sigma))$, and $w_h\coloneqq\hwidth(\hnf(\Sigma))$.

\begin{lem}\label{lem:tgd_number_full_vnf}
Any \tgd $\tau$ in variable normal form contains at most $n(\bwidth(\tau)^a+\hwidth(\tau)^a)$ atoms.
\end{lem} 
\begin{proof}
For any fixed predicate symbol, we have at most $\bwidth(\tau)^a$ 
and $\hwidth(\tau)^a$ different instantiations in the body and head of $\tau$, respectively.
\end{proof}

\begin{cor}\label{cor:gtgd_guarded_closure_size}
$|\gcloex(\Sigma)|\leq 2^{n(w_b^a+w_h^a)}$.
\end{cor} 
\begin{proof}
By \cref{prop:gtgd_guarded_closure_resolvent_properties}, $\gcloex(\Sigma)$ 
is a set of \gtgds in variable normal form with $\bwidth(\tau)\leq w_b$ and
$\hwidth(\tau)\leq w_h$.
A \tgd consists of a subset of all possible body and head atoms.
Now the claim follows by \cref{lem:tgd_number_full_vnf}.
\end{proof} 

\begin{cor}\label{cor:gtgd_guarded_closure_atomic_rewriting}
The guarded saturation $\gclo(\Sigma)$ is an atomic rewriting of~$\Sigma$.
\end{cor} 
\begin{proof}
By \cref{cor:gtgd_abstr_atomic_rewriting,cor:gtgd_guarded_closure_soundness,prop:gtgd_guarded_closure_evolve_complete}.
\end{proof}

\begin{lem}\label{lem:gtgd_guarded_closure_evolve_complexity}
Let $\tau,\tau'\in\gcloex(\Sigma)$.
Then all \evolve-resolvents of ${\tau=\body\rightarrow\head}$, ${\tau'=\body'\rightarrow\head'}$
can be obtained in $\ptime$ for bounded $w_b,n,a$,
$\exptime$ for bounded $a$, and $\texptime$ otherwise.
\end{lem} 
\begin{proof}
Let $G'\in\body'$ be a guard.
Using our observations in \cref{prop:gtgd_guarded_unique_sequence}, we first check for any $H\in\head$,
whether there is an mgu $\theta^*$ of $G'$ and $H$
that is the identity on~$\vec y$ (by treating $\vec y$ as constants).
We then select all atoms $B_2',\dotsc,B_m'\in\body'\setminus\{G'\}$ with ${\vars(B_i'\theta^*)\cap\vec y\neq\emptyset}$.
These first steps can be done in linear time using \cref{lem:mgu_complexity}.
Also note that the exact choice of $G'$ and the order of ${S'\coloneqq G',B_2',\dotsc,B_m'}$ are irrelevant by \cref{lem:gtgd_guarded_guard_included,lem:gtgd_guarded_permutation_independent}.

Now for any $B_i'\neq G'$, we have at most $|\head|$ possible counterparts in $\head$.
We check for any possible combination if we can extend $\theta^*$ to an mgu $\theta$ of the sequences that is 
the identity on $\vec y$ (again, by treating $\vec y$ as constants) and $\vec x\theta\cap\vec y=\emptyset$.
Finally, we resolve the rules and normalise the result to obtain a new rule.
Using \cref{lem:mgu_complexity,lem:tgd_vnf_complexity,def:tgd_hnf}, these steps can again be done in linear time.
Since $\tau$ contains at most $s\coloneqq nw_b^a$ body atoms and $\tau'$ at most $s'\coloneqq nw_h^a$ head atoms, we have to repeat above steps at most ${s'}^s$ times.
Hence, we obtain the claimed complexity bounds.
\end{proof} 

\begin{lem}\label{lem:gtgd_guarded_closure_complexity}
The guarded saturation algorithm terminates in 
$\ptime$ for bounded $w,n,a$,
$\exptime$ for bounded $a$, and $\texptime$ otherwise.
\end{lem} 
\begin{proof}
The initial variable and head normal form transformations can be done in linear time by \cref{lem:tgd_vnf_complexity,def:tgd_hnf}.
The algorithm then takes all pairs in $\gcloex(\Sigma)$
and performs all \evolve-inferences.
By \cref{cor:gtgd_guarded_closure_size}, this process takes at most~$c\coloneqq2^{n(w_b^a+w_h^a)}$ iterations,
and at any point, there are at most $\calO(c^2)$ pairs.
By \cref{lem:gtgd_guarded_closure_evolve_complexity}, 
the complexity for each $\tau,\tau'$ is at most as high as the stated bounds.
Thus, we obtain the claimed complexity bounds.
\end{proof} 

Finally, we can combine our results to obtain a decision procedure.
\begin{thm}\label{thm:gtgd_guarded_dec_proc}
The guarded saturation algorithm provides a decision procedure for the \qfqaprobgtgds. 
The procedure terminates in 
$\ptime$ for fixed $\Sigma$,
$\exptime$ for bounded~$a$, and
$\texptime$ otherwise.
\end{thm} 
\begin{proof}
By \cref{prop:gtgd_guarded_closure_resolvent_properties},
the number of variables for each $\tau\in\gclo(\Sigma)$ is bounded by~$w_b$.
Now the claim follows from \cref{prop:gtgd_atomic_rewriting_dec_proc,cor:gtgd_guarded_closure_atomic_rewriting,lem:gtgd_guarded_closure_complexity}.
\end{proof} 

Moreover, our results match the known lower bounds from the literature, and hence,
our algorithm is theoretically optimal.
In fact, the lower bounds already hold for a simpler problem setting 
in which only atomic queries are considered:
\begin{thm}[\citep{cali2013taming}]\label{thm:lower_bounds_gtgd}
Given a set of \gtgds $\Sigma$,
the Atomic Query Answering Problem under $\Sigma$ is 
$\ptime$-complete for fixed $\Sigma$,
$\exptime$-complete for bounded~$a$,
and $\texptime$-complete in general.
It remains $\texptime$-complete even if $n$ is bounded.
\end{thm} 
\begin{rmk}\label{rmk:compl_issue_gtgd}
\cite{cali2013taming} in fact state that the Atomic Query Answering Problem 
already lowers to $\ptime$ in case of bounded $a$ and $n$;
however, they only sketch this result for the case of single-headed \gtgds.
Later on, they extend their results to the case of multiple head atoms,
but neglect the case of bounded $a$ and $n$.
In particular, we believe that the number of clouds, as described in their paper,
is still exponential in case of bounded $a$ and $n$ due to the unrestricted
number of existential variables.
We hence do not include this result, but leave it as an open question
whether the problem indeed lowers to $\ptime$ in this case.
\end{rmk}

%% file: chapters/dgtgd.tex
\chapter{\qfqadgtgds}\label{chap:qfqadgtgds}
In this chapter, we extend the results that we derived for \gtgds to \dgtgds.
The structure of this chapter follows that of the previous one:
we first introduce a property that will imply completeness and then
present a sound resolution-based algorithm that satisfies this property.

\input{pages/dgtgd/dgtgd_property_for_completeness}
\input{pages/dgtgd/dgtgd_guarded_saturation}

\input{pages/dgtgd/dgtgd_guarded_complexity}

%% file: pages/dgtgd/dgtgd_property_for_completeness.tex
\section{An Abstract Property for Completeness}\label{sec:dgtgd_property_for_completeness}
Recall that a \dtgd is of the form
\begin{align*}
	\body(\vec x)\rightarrow \bigvee\limits_{i=1}^n \exists \vec y_i\,\head_i(\vec x,\vec y_i).
\end{align*} 
We might first want to generalise our head normal form definition to \dtgds
to minimise the number of body variables occurring in an existential head conjunction;
however, due to the introduction of disjunctions, 
we might lose completeness by naively splitting the head into a non-full and full part like we did for \tgds.
\begin{exmpl}
Consider the \dtgds $\Sigma$
\begin{align*}
	&B(x)\rightarrow\exists y\bigl(P(x,y)\land R(x)\bigr)\lor\exists y\bigl(S(x,y)\land T(x)\bigr),\\
	&S(x_1,x_2)\land T(x_1)\rightarrow R(x_1)
\end{align*} 
and the set of rules $\Sigma'$ obtained by naively generalising our head normal form 
\begin{align*}
	&B(x)\rightarrow\exists y\,P(x,y)\lor\exists y\, S(x,y),\\
	&B(x)\rightarrow R(x)\lor T(x),\\
	&S(x_1,x_2)\land T(x_1)\rightarrow R(x_1).
\end{align*} 
Let $\db\coloneqq\{B(c)\}$.
Then $\db,\Sigma\models R(c)$ but $\db,\Sigma'\not\models R(c)$ since, for example,
the structure $\calM\coloneqq\{B(c),P(c,d),T(c)\}$ models $\db,\Sigma'$ but not $R(c)$.
Note, however, that all rules in~$\Sigma'$ are consequences of rules in $\Sigma$.
That is, $\Sigma'$ is a sound rewriting.
\end{exmpl} 

We thus do not impose an analogous head normal form on \dtgds.
Instead, we introduce the notion of \emph{immediate full consequences}.
\begin{defn}[\dtgd immediate full consequences]
Given a \dtgd 
\begin{equation*}
	\tau=\body(\vec x)\rightarrow \bigvee_{i=1}^n \exists \vec y_i\head_i(\vec x,\vec y_i),
\end{equation*} 
let $\head_i'$ consist of all atoms $H\in\head_i$ that contain no $y_i$,
and suppose that each $\head_i'$ is non-empty.
Then the \emph{\iindex{immediate full consequence} of $\tau$} is defined by 
$\ifc(\tau)\coloneqq\body(\vec x)\rightarrow \bigvee_{i=1}^n \head_i'(\vec x)$.
Given a set of \dtgds $\Sigma$, we write \iindex[not]{$\ifc(\Sigma)$} for the set of
immediate full consequences following from all rules in~$\Sigma$.
\end{defn} 

\begin{exmpl}
The immediate full consequence of
\begin{equation*}
	B(x)\rightarrow\exists y\bigl(H_1(x,y)\land H_2(x)\land H_3(x)\bigr)\lor H_4(x)
\end{equation*} 
is
\begin{equation*}
	B(x)\rightarrow \bigl(H_2(x)\land H_3(x)\bigr)\lor H_4(x).
\end{equation*} 
\end{exmpl} 

\begin{lem}\label{lem:dtgd_ifc_sound}
Let $\Sigma$ be a set of \dtgds. Then $\ifc(\Sigma)$ is sound \wrt $\Sigma$.
\end{lem} 
\begin{proof}
Let $\tau=\bigl(\beta(\vec x)\rightarrow\bigvee_{i=1}^n \exists \vec y_i\,\head_i(\vec x,\vec y_i)\bigr)\in\Sigma$ such that $\ifc(\tau)$ exists,
and let $\calA$ be a structure with $\calA\models\tau$. 
We have to show that $\calA\models\ifc(\tau)$.
Assume $\calA\models\beta(\vec c)$ for some $\vec c\subseteq\dom(\calA)$.
Then $\calA\models\exists\vec y_i\,\head_i(\vec c,\vec y_i)$ for some $\head_i$.
Hence, there is $\vec d\subseteq\dom(\calA)$ such that $\calA\models H(\vec c,\vec d\,)$ for any $H\in\head_i$. 
Thus, $\calA\models\head_i'(\vec c)$, where $\head_i'$ consist of all atoms $H\in\head_i$ that contain no $y_i$,
and hence $\calA\models\head_i'(\vec c)$.
Consequently, $\calA\models\ifc(\tau)$.
\end{proof} 

As we did for \gtgds, we now want to derive a set of properties that
will imply completeness of a rewriting \wrt some set of \dgtgds.
To achieve this, we generalise our previous thoughts from tree-like chase sequences to tree-like chase trees.

Recall the structure of a tree-like chase tree from some database $\db$ and set of \mbox{\dgtgds}~$\Sigma$.
A tree-like chase tree $T$ models some \ubcq $Q$ \ifftext all the facts of some $Q_i\in Q$ are contained in
the root node~$r$ of \emph{every} leaf of $T$.
To obtain sufficient completeness properties for a full set of \dtgds $\Sigma'$, 
we can again separate our concerns into two cases:
either some rule fired in $r$ creates new facts relevant to $r$ in every new
node of $T$, or a non-full rule fired in $r$ creates new children $v_1,\dotsc,v_n$ of $r$
and there is a chase tree $T^i$ from any child $v_i$ such that $r$ adds on a fact in every leaf of $T^i$.

The first case is again rather simple: we just let $\Sigma'$ contain all immediate full consequences of $\Sigma$.
For the second case, it again suffices to consider a simpler situation
in which each $T^i$ is created by firing full rules only.
By an inductive argument (\cref{lem:dgtgd_abstr_completeness_lift}),
we can then again lift the property for this situation to general chase trees.

The following definition will be quite useful for the formal definition of our abstract
properties and our upcoming proofs.

\begin{defn}[Query of a chase tree]
Given a chase tree $T$ rooted at $N$ with node~$r$,
the \emph{query of $T$} is defined by \index[not]{$Q^T$}$Q^T\coloneqq\lvs(T)_r$.
\end{defn} 

\begin{restatable}[Evolve-completeness]{defn}{dgtgdabstrcompletenessprops}\label{defn:dgtgd_abstr_completeness_props}
Let $\Sigma$ be a set of \dgtgds,
and $\Sigma'$ be a set of full \dtgds.
We say that $\Sigma'$ is \index{evolve-completeness!\dgtgd}\emph{evolve-complete \wrt $\Sigma$} 
if it satisfies the following properties:
\begin{propylist}
\item\label{defn:dgtgd_abstr_completeness_props_1}
For every $\tau\in\ifc(\Sigma)$, we have $\tau\in\Sigma'$ (modulo renaming).

\item\label{defn:dgtgd_abstr_completeness_props_2}
Let $T$ be a chase tree with root $N$,
let $r$ be the node inside~$N$,
let $N^1,\dotsc,N^{n+m}$ be the children of $N$,
and $\tau_r\in\Sigma$ be non-full such that
\begin{assmlist}
\item $\tau_r$ is the first rule fired in $T$,
\item the first $n$ head conjunctions of $\tau_r$ are non-full,
\item for $1\leq i\leq n$, $T^{N^i}$ was created by firing full rules from~$\Sigma'$ in $v_i$ only, 
where $v_i$ is the new child of $r$ in $N^i$,
\item $N^{n+1},\dotsc,N^{n+m}$ have no subtrees.
\end{assmlist} 
Then there exists a chase tree $T'$ from $N$ and~$\Sigma'$ with $T'\models Q^T$.
\end{propylist}
\end{restatable}

In the following proofs, we will frequently assume that whenever a node $N^i$ in $T$
is created by a non-full head conjunction,
its subtree $T^{N^i}$ is fully expanded before expanding the subtree of any sibling $N^{n+1},\dotsc,N^{n+m}$
that got created by a full head conjunction.
This can be done without loss of generality since the final structure of $T$ 
does not depend on the order used to expand the subtrees.

We again show that any evolve-complete set is in fact complete
by first proving a key lifting lemma that relies on the one-pass chase property.
\begin{lem}\label{lem:dgtgd_abstr_completeness_lift}
Let $\Sigma$ be a set of \dgtgds,
let $\Sigma'$ be evolve-complete \wrt $\Sigma$,
let $T$ be a one-pass chase tree from $\Sigma$,
let $N$ be a node in $T$,
let $p$ be a node inside~$N$,
let $N^1,\dotsc,N^{n+m}$ be the children of $N$,
and $\tau_p\in\Sigma$ be non-full such that
\begin{assmlist}
\item $\tau_p$ was fired in node $p$ of $N$,
\item the first $n$ head conjunctions of $\tau_p$ are non-full,
\item for $1\leq i\leq n$, all leaves of $T^{N^i}$ added on facts in $v_i$, where $v_i$ is the new child of $p$ in~$N^i$,
\item $N^{n+1},\dotsc,N^{n+m}$ have no subtrees.
\end{assmlist} 
Then there exists a chase tree $T'$ from~$N_p$ and $\Sigma'$ such that $T'\models \lvs\bigl(T^N\bigr)_p$.
\end{lem}
\begin{proof}
We can assume that whenever a non-full rule is fired to create $T$,
every subtree of a non-full head conjunction is fully expanded first.
We additionally show for $1\leq i\leq n$, 
that there exists a chase tree $T^i$ from $N^i_{v_i}$ and $\Sigma'$
with $T^i\models \lvs\bigl(T^{N^i}\bigr)_{v_i}$.
We prove the claim by induction on the number of steps used to create $T^{N_p}$, denoted by $k$. 
In case $k=1$, we initialise $T'$ with the single node $N_p$,
and if $N^1,\dotsc,N^m$ added on facts in $p$, we additionally obtain
$\ifc(\tau_p)\in\Sigma'$ by \cref{defn:dgtgd_abstr_completeness_props_1}
that we then fire in $N_p$.

In case $k>1$, let $\tau\in\Sigma$ be the rule that got fired in step $k$ in a leaf $L$ of $T_{k-1}^{N^j}$ for some $j\in\{1,\dotsc,n\}$.
Due to one-pass, we can assume that the new facts relevant to~$v_j$ 
are generated by acting on some node~$d$ that is a non-strict descendant of $v_j$.

First consider the case $v_j=d$. 
By the inductive hypothesis,
we get a chase tree $T^i$ from $N^i_{v_i}$ and $\Sigma'$
with $T^i\models \lvs\bigl(T_{k-1}^{N^i}\bigr)_{v_i}$ for $1\leq i\leq n$.
Since all new leaves added on facts in~$v_j$, 
there is $\ifc(\tau)\in\Sigma$ creating the new facts relevant to $v_j$,
and by \cref{defn:dgtgd_abstr_completeness_props_1}, $\ifc(\tau)\in\Sigma'$.
Hence, we can fire $\ifc(\tau)$ in any leaf $L'$ of $T^j$ 
with $L'\models L_{v_j}$ to obtain a chase tree $(T^j)'$ from $N^j_{v_j}$ and $\Sigma'$
with $(T^j)'\models\lvs\bigl(T_k^{N^j}\bigr)_{v_j}$.
Now we obtain the desired $T'$ by using \cref{defn:dgtgd_abstr_completeness_props_2}.

Next consider the case where $d$ is a descendant of $v_j$. 
Let $c$ be the first child of $v_j$ on the path to $d$, created
by firing a non-full rule in node $N'$.
Set ${T_{N'}\coloneqq\cut\bigl(T_{k-1}^N,N'\bigr)}$.
By the inductive hypothesis,
we get a chase tree $T^i$ from $N^i_{v_i}$ and $\Sigma'$
with $T^i\models \lvs\bigl(T_{N'}^{N^i}\bigr)_{v_i}$ for $1\leq i\leq n$.
Let $(N')^1,\dotsc,(N')^{n'+m'}$ be the children of $N'$,
where the first $n'$ nodes are created by non-full head conjunctions.
Also by the inductive hypothesis,
we get a chase tree $T'$ from $N'_{v_j}$ and $\Sigma'$ with $T'\models\lvs\bigl(T_k^{N'}\bigr)_{v_j}$, 
where we used the assumption that subtrees rooted at nodes created by a non-full head conjunction are fully expanded first.
Hence, by attaching $T'$ to any leaf $L'$ of $T^j$ with $L'\models N'_{v_j}$,
we obtain $(T^j)'$ with $(T^j)'\models \lvs\bigl(T_k^{N^j}\bigr)_{v_j}$.
Finally, we again use \cref{defn:dgtgd_abstr_completeness_props_2}
to obtain the desired chase tree from $N_p$ and~$\Sigma'$.
\end{proof}

The completeness property again follows by a simple induction.
\begin{prop}[Completeness]\label{prop:dgtgd_abstr_completeness}
Let $\Sigma$ be a set of \dgtgds and
$\Sigma'$ be evolve-complete \wrt $\Sigma$.
Then $\Sigma'$ is complete \wrt $\Sigma$.
\end{prop} 
\begin{proof}
By \cref{thm:chase_univ,thm:dgtgd_iff_one_pass}, 
it suffices to show that whenever there is a one-pass chase tree 
$T$ from $\db$ and $\Sigma$ for a quantifier-free \ubcq $Q$, 
then $Q$ also follows from $\db$ and~$\Sigma'$.
Let $r$ be the single node of the root of $T$
and $L^1,\dotsc,L^m$ be the leaves of~$T$.
Since $T\models Q$ \ifftext $L_r^i\models Q_{j_i}$ for each $L^i$ and some $Q_{j_i}\in Q$,
it suffices to show that we can create a chase tree $T'$
from $\db$ and $\Sigma'$ with $T'\models Q^T$.
We can again assume that subtrees rooted at nodes created by a non-full head conjunction are fully expanded first.
We show the claim by induction on the number of steps fired to create~$T$,
denoted by $n$.
In case $n=0$, we just set $T'\coloneqq T$.

Assume $n>0$. 
By the inductive hypothesis, we get $T_{n-1}'$ such that $T_{n-1}'\models Q^{T_{n-1}}$.
We have to create $T_n'$ such that $T_n'\models Q^{T_n}$.
Let $\tau\in\Sigma$ be the rule firing at step $n$ 
in some node $d$ of some $L\in\lvs(T_{n-1})$,
and assume that $r$ adds on a fact in every new leaf;
otherwise, we can simply set $T_n'\coloneqq T_{n-1}'$.

If $d=r$, then the new facts relevant to $r$ can be created by $\ifc(\tau)$.
By \cref{defn:dgtgd_abstr_completeness_props_1}, $\ifc(\tau)\in\Sigma'$,
and thus, we can fire $\ifc(\tau)$ in any leaf $L'$ of $T_{n-1}'$ with $L'\models L_r$ to
obtain~$T_n'$ with $T_n'\models Q^{T_n}$.

Otherwise, let $c$ be the first child of $r$ on the path to $d$, 
created by firing a rule in node $N$ with children $N^1,\dotsc,N^{k+l}$,
where the first $k$ nodes are created by non-full head conjunctions.
By \cref{lem:dgtgd_abstr_completeness_lift}, there is a chase tree $T''$ from $N_r$ 
and $\Sigma'$ with $T''\models\lvs\bigl(T_n^N\bigr)_r$, 
where we used the assumption that subtrees rooted at nodes created by a non-full head conjunction are fully expanded first.
Further, by the inductive hypothesis, we get $T'$ with $T'\models Q^{\cut(T_{n-1},N)}$.
Thus, by attaching $T''$ to any leaf $L'$ of $T'$ with $L'\models N_r$,
we obtain~$T_n'$ with $T_n'\models Q^{T_n}$.
\end{proof} 

\begin{cor}\label{cor:dgtgd_abstr_atomic_rewriting}
Let $\Sigma$ be a set of \dgtgds and
$\Sigma'$ be finite, sound, and evolve-complete \wrt~$\Sigma$.
Then $\Sigma'$ is an atomic rewriting of $\Sigma$.
\end{cor} 
\begin{proof}
$\Sigma'$ is full by \cref{defn:dgtgd_abstr_completeness_props},
finite and sound by assumption,
complete by \cref{prop:dgtgd_abstr_completeness},
and hence an atomic rewriting by \cref{cor:atomic_rewriting_def}.
\end{proof} 

%% file: pages/dgtgd/dgtgd_guarded_saturation.tex
\section{The Disjunctive Guarded Saturation}\label{sec:dgtgd_guarded_saturation}
In this section, we present a procedure for \dgtgds that returns an atomic rewriting.
Of course, we want to reuse some ideas established for our previous rewriting algorithms.
Unfortunately, simply generalising our guarded saturation from \gtgds to \dgtgds leads to a combinatorial explosion:

\begin{exmpl}
Given the \dgtgds 
\begin{align*}
\tau_1&=R(x_1)\rightarrow\exists y\bigl(S(x_1,y)\land T(x_1)\bigr)\lor U(x_1),\\
\tau_2&=S(x_1,x_2)\rightarrow V(x_1)\lor W(x_2),
\end{align*} 
the composition of $\tau_1,\tau_2$ amounts to 
\begin{align*}
&R(x_1)\rightarrow\exists y\bigl(S(x_1,y)\land T(x_1)\land(V(x_1)\lor W(y))\bigr)\lor U(x_1)\\
\equiv\quad&R(x_1)\rightarrow\exists y\bigl(S(x_1,y)\land T(x_1)\land V(x_1)\bigr)
\lor \exists y\bigl(S(x_1,y)\land T(x_1)\land W(y)\bigr)\lor U(x_1).
\end{align*} 
It is evident that the head of a derived rule can grow explosively.
In general, the head of a rule can consist of any possible boolean combination of atoms (in disjunctive normal form).
Using \cref{lem:tgd_number_full_vnf}, 
the number of different atoms using $w$ variables is bounded by
$s\coloneqq 2nw^a$, where $n$ is the number of predicate symbols and $a$ the maximum arity of any predicate.
The head of a \dgtgd consists of a set of head conjunctions, which in turn consist of a set of atoms;
consequently, this results in a triple exponential $2^{2^s}$ bound on the number of different heads.
\end{exmpl} 

To avoid this explosion, we have to introduce two new concepts.
The first is given by another normal form transformation.

\begin{defn}[Single head normal form]\label{defn:dtgd_shnf}
We say that a \dtgd $\tau$ is in \emph{\iindex{single head normal form}} or \emph{single-headed} if every head conjunction of $\tau$
consists of exactly one atom.
Given a set of \dtgds $\Sigma$, the \emph{single head normal form \iindex[not]{$\shnf(\Sigma)$} of $\Sigma$} is obtained by replacing every rule
\begin{equation*}
	\body(\vec x)\rightarrow \bigvee\limits_{i=1}^n \exists \vec y_i\,\bigwedge\limits_{j=1}^{n_i}H_{i,j}(\vec x,\vec y_i)
\end{equation*} 
in $\Sigma$ by 
\begin{equation*}
	\body(\vec x)\rightarrow \bigvee\limits_{i=1}^n \exists \vec y_i\,R_i(\vec x_{i\restriction},\vec y_i),
\end{equation*} 
and, for any $1\leq i\leq n$,
\begin{align*}
	R_i(\vec x_{i\restriction},\vec y_i)&\rightarrow H_{i,1}(\vec x_{i\restriction},\vec y_i),\\
	&\vdotswithin{\rightarrow}\\
	R_i(\vec x_{i\restriction},\vec y_i)&\rightarrow H_{i,n_i}(\vec x_{i\restriction},\vec y_i),
\end{align*} 
where $R_i\notin\relset$ is a fresh relation symbol and $\vec y_i$ might be empty. 
This replacement can be done in time linear in $\size(\Sigma)$.
\end{defn}

\begin{exmpl}
The single head normal form of 
\begin{equation*}
	B(x_1,x_2)\rightarrow\exists y_1,y_2\bigl(H_1(y_1,x_1,y_2)\land H_2(y_2)\bigr)\lor\exists y_1\,H_3(y_1,x_2)
\end{equation*} 
is
\begin{align*}
	B(x_1,x_2)&\rightarrow\exists y_1,y_2\,R_1(x_1,y_1,y_2)\lor\exists y_1\,R_2(x_2,y_1),\\
	R_1(x_1,y_1,y_2)&\rightarrow H_1(y_1,x_1,y_2),\\
	R_1(x_1,y_1,y_2)&\rightarrow H_2(y_2),\\
	R_2(x_2,y_1)&\rightarrow H_3(y_1,x_2).
\end{align*} 
\end{exmpl} 

\begin{lem}\label{lem:dgtgd_shnf_equiv}
For any database $\db$, set of \dtgds $\Sigma$, and \ubcq $Q$,
we have $\db,\Sigma\models Q$ \ifftext $\db,\shnf(\Sigma)\models Q$.
\end{lem} 
\begin{proof}
If we fire a rule $\tau$ in a chase using $\Sigma$, 
we can simply fire all rules in $\shnf(\tau)$ to create the corresponding facts
in a chase using $\shnf(\Sigma)$.

Conversely, if we fire a rule $\tau\in\shnf(\Sigma)\setminus\Sigma$ in a chase using $\shnf(\Sigma)$, 
we can fire the rule $\tau'\in\Sigma$ that produced $\tau$
to create all corresponding facts that use relations in $\relset$
in a chase using $\Sigma$.
All other facts with relations not in $\relset$ produced by $\tau$ can only trigger
further rules in $\shnf(\tau')$ that create facts which have then already been
added in the chase using $\Sigma$.

Since, by definition, $Q$ only uses relations in $\relset$, this finishes the proof.
\end{proof} 

In the following, given a single-headed \dtgd $\tau=\body(\vec x)\rightarrow\head$
with $\head=\bigvee_{i=1}^n\vec y_i\,H_i(\vec x,\vec y_i)$,
instead of treating $\head$ as the set $\{\{H_i\}\mid 1\leq i\leq n\}$,
we treat $\head$ as the set $\{H_i\mid 1\leq i\leq n\}$.

The single head normal form transformation alone is not yet enough as our previous \evolve-rule
collects all head atoms of both resolved rules and hence does not preserve single-headedness. 
As a running example, take the single-headed \dgtgds
\begin{align*}
\tau_1&=R(x)\rightarrow \exists y\,P(x,y),\\
\tau_2&=P(x_1,x_2)\rightarrow S(x_1,x_2),\\
\tau_3&=P(x_1,x_2)\land S(x_1,x_2)\rightarrow T(x_1).
\end{align*} 
Then composing $\tau_1,\tau_2$ results in $\tau_1'=R(x)\rightarrow \exists y\bigl(P(x,y)\land S(x,y)\bigr)$, 
which then can be composed with $\tau_3$ to obtain $\tau_2'=R(x)\rightarrow T(x)$.
But unfortunately, $\tau_1'$ is not single-headed.
If we want to preserve single-headedness, we can no longer collect all derivable non-full head atoms in a single rule.
Instead, when composing $\tau_1,\tau_2$, let us try to just derive the single-headed $\tau_1^*=R(x)\rightarrow \exists y\,S(x,y)$.
It is then not possible, however, to derive $\tau_2'$ by composing $\tau_3$ with $\tau_1$ and $\tau_1^*$.
We lost the information that $\tau_1$ and $\tau_1^*$
can make use of the same witnesses for their existential variables.
To prevent this loss of information, we introduce the \emph{Skolem normal form} of \dgtgds (see also \cref{sec:prelims_first_order_logic}).
\begin{defn}[Skolemisation]\label{defn:dtgd_Skolemisation}
Given a \dtgd $\tau=\body(\vec x)\rightarrow \bigvee_{i=1}^n \exists \vec y_i\,\head_i(\vec x,\vec y_i)$,
the \index{Skolemisation}\emph{\iindex{Skolem normal form}} $\skol(\tau)$ of $\tau$ 
is obtained by replacing every existential $y_j\in\vec y_i$ by a term
$f_{i,j}(\vec x)$, where $f_{i,j}$ is a fresh function symbol.
An atom that contains a function is called \emph{functional},
and a rule is called \emph{functional} if it contains a functional atom.
The definition is lifted to set of \dtgds by setting \index[not]{$\skol(\Sigma)$}$\skol(\Sigma)\coloneqq\{\skol(\tau)\mid\tau\in\Sigma\}$.
\end{defn} 
\begin{exmpl}
The Skolemisation of 
\begin{equation*}
	B(x_1,x_2)\rightarrow\exists y_1,y_2\bigl(H_1(y_1,x_1,y_2)\land H_2(y_2)\bigr)\lor\exists y_1\,H_3(y_1,x_2)
\end{equation*} 
is
\begin{equation*}
B(x_1,x_2)\rightarrow\bigl(H_1(f_{1,1}(x_1,x_2),x_1,f_{1,2}(x_1,x_2))\land H_2(f_{1,2}(x_1,x_2))\bigr)\lor H_3(f_{2,1}(x_1,x_2),x_2).
\end{equation*} 
\end{exmpl} 
Note that the Skolemisation of \dtgds only creates atoms of a simple form.
\begin{defn}[Shallow term]
A term $t$ is called \index{shallow term}\emph{shallow} if $t$ contains no nested function.
\end{defn} 
\begin{defn}[Simple atom/fact/set]
An atom $R$ is called \index{simple!atom}\emph{simple} if $R$ only contains shallow terms.
If $R$ is also ground, we call it a \index{simple!fact}\emph{simple fact}.
A set $S$ is called \index{simple!set}\emph{simple} if every $R\in S$ is a simple fact.
Given a simple set $S$, we define \index[not]{$\func{S}$}$\func{S}\coloneqq\{R\in S\mid \text{R is functional}\}$ and
\index[not]{$\nfunc{S}$}$\nfunc{S}\coloneqq\{{R\in S}\mid \text{R is non-functional}\}$.
Given a family of simple sets~$U$, we lift this definition to
$\func{U}\coloneqq\{\func{S}\mid S\in U\}$ and $\nfunc{U}\coloneqq\{\nfunc{S}\mid S\in U\}$.
\end{defn} 
Coming back to our running example, we now start with a Skolemised set of rules
\begin{align*}
\tau_1&=R(x)\rightarrow P(x,f(x)),\\
\tau_2&=P(x_1,x_2)\rightarrow S(x_1,x_2),\\
\tau_3&=P(x_1,x_2)\land S(x_1,x_2)\rightarrow T(x_1).
\end{align*}  
Then composing $\tau_1$ and $\tau_2$ results in $\tau_1'=R(x)\rightarrow S(x,f(x))$,
and composing $\tau_1$ with $\tau_3$ results in $\tau_2'=R(x)\land S(x,f(x))\rightarrow T(x)$. 
The functional premise $S(x,f(x))$ in $\tau_2'$ can then be discharged
by composing $\tau_1'$ and $\tau_2'$ to obtain $\tau_3'=R(x)\rightarrow T(x)$.
This process gives the general idea of the upcoming algorithm, which will be described shortly.
But first, we examine which kind of rules we might obtain during this process.
Since \dtgds are function-free by definition, we are not dealing with pure \dtgds anymore;
instead, we will consider rules of a more general form, namely
\begin{equation*}
	\forall\vec x\left(\body(\vec x)\rightarrow\bigvee\limits_{i=1}^n\head_i(\vec x)\right),
\end{equation*} 
where $\body$ and each $\head_i$ are conjunctions of simple atoms in which
every functional term contains all variables of the rule, and further,
$\body$ contains a function-free guard.
We refer to such rules as \emph{\iindex{guarded simple rules}}.

Note that the Skolemisation of non-full \dgtgds only produces guarded simple rules 
with a non-functional body and functional head.
Conversely, a guarded simple rule with non-functional body and functional head 
corresponds to a non-full \dgtgd via ``de-Skolemisation''.
For this reason, we call such rules \emph{non-full}.
Similarly, we call a function-free guarded simple rule \emph{full}
as they correspond to full \dgtgds.

We now extend some of our definitions to guarded simple rules.
\begin{defn}[Variable normal form]
Given a guarded simple rule $\tau=\body\rightarrow \bigvee_{i=1}^n\head_i$,
we say that $\tau$ is in \index{variable normal form!guarded simple rule}\emph{variable normal form} 
if $x_i$ is the $i^{th}$ lexicographic variable occurring in~$\body$.
We write $\vnf(\tau)$ and \iindex[not]{$\vnf(\Sigma)$} for the variable normal form of some 
guarded simple rule $\tau$ and set of guarded simple rules $\Sigma$, respectively.
\end{defn} 

\begin{exmpl}
The rule
\begin{equation*}
	\forall x_3,y_2\bigl[B(y_2,x_3)\rightarrow\bigl(H_1(f(x_3,y_2),y_2,x_3)\land H_2(x_3)\bigr)\lor H_3(y_2,g(x_3,y_2))\bigr]
\end{equation*} 
is not in variable normal form. Its variable normal form is
\begin{equation*}
	\forall x_1,x_2\bigl[B(x_1,x_2)\rightarrow\bigl(H_1(f(x_2,x_1),x_1,x_2)\land H_2(x_2)\bigr)\lor H_3(x_1,g(x_2,x_1))\bigr].
\end{equation*} 
\end{exmpl} 

\begin{lem}\label{lem:guarded_simple_vnf_complexity}
Any guarded simple rule $\tau$ can be rewritten into $\vnf(\tau)$ in linear time.
\end{lem} 
\begin{proof}
Simply extend \cref{lem:tgd_vnf_complexity} to multiple head conjunctions.
\end{proof} 

\begin{defn}
\index[not]{$\bwidth(\Sigma)$}\index[not]{$\hwidth(\Sigma)$}\index[not]{$\width(\Sigma)$}
Given a \dtgd (or a guarded simple rule) $\tau$ with body $\body$ and head~$\head=\{\head_1,\dotsc,\head_n\}$, 
we define
\begin{thmlist}
\item the \index{body width!\dtgd}\index{body width!guarded simple rule}\emph{body width of $\tau$} as $\bwidth(\tau)\coloneqq|\vars(\body)|$,
\item the \emph{$\nth{i}$ head width of $\tau$} as $\hwidth_i(\tau)\coloneqq|\vars(\head_i)|$,
\item the \index{head width!\dtgd}\index{head width!guarded simple rule}\emph{head width of $\tau$} as $\hwidth(\tau)\coloneqq\max\{\hwidth_i(\tau)\}$, and
\item the \index{width!\dtgd}\index{width!guarded simple rule}\emph{width of $\tau$} as $\width(\tau)\coloneqq\max\{\bwidth(\tau),\hwidth(\tau)\}$.
\end{thmlist} 
These definitions are naturally extended to sets of \dtgds (or guarded simple rules).
\end{defn} 
\begin{lem}
For any guarded simple rule $\tau$, we have $\width(\tau)=\bwidth(\tau)$.
\end{lem} 
\begin{proof}
Any variable in the head of $\tau$ occurs in the guard of $\tau$.
\end{proof} 

\begin{defn}[Immediate full consequences]
Given a guarded simple rule $\tau$ with non-functional body,
the \emph{\iindex{immediate full consequence} of $\tau$} is obtained
by taking the immediate full consequence of the \dgtgd corresponding to $\tau$ via de-Skolemisation.
\end{defn} 
\begin{exmpl}
The immediate full consequence of
\begin{equation*}
	B(x)\rightarrow\bigl(H_1(x,f(x))\land H_2(x)\land H_3(x)\bigr)\lor H_4(x)
\end{equation*} 
is
\begin{equation*}
	B(x)\rightarrow \bigl(H_2(x)\land H_3(x)\bigr)\lor H_4(x).
\end{equation*} 
\end{exmpl} 
\begin{lem}\label{lem:guarded_simple_ifc_sound}
Let $\Sigma$ be a set of guarded simple rules. 
Then $\ifc(\Sigma)$ is sound \wrt $\Sigma$.
\end{lem} 
\begin{proof}
Analagous to \cref{lem:dtgd_ifc_sound}.
\end{proof} 

Finally, we are ready to define our rewriting algorithm.
\begin{defn}[Disjunctive guarded saturation] \label{defn:dgtgd_guarded_closure}
We define the following inference rule on single-headed guarded simple rules:
\begin{itemize}
\item \emph{\devolve}. Assume we have some non-full $\tau=\beta\rightarrow\head$,
a functional atom $H\in\head$,
a rule $\tau'=\beta'\rightarrow\head'$,
and some $B'\in\body'$ such that
\begin{propylist}
\item $\tau'$ is full and $B'$ is a guard, or
\item $B'$ is functional,
\end{propylist} 
and assume that $\tau,\tau'$ use distinct variables (otherwise, perform a renaming).
Suppose~$\theta$ is an mgu of $H$ and $B'$ and define 
\begin{align*}
\beta''&\coloneqq\bigl(\beta\cup\bigl(\beta'\setminus\{B'\}\bigr)\bigr)\theta,\\
\head''&\coloneqq\bigl(\bigl(\head\setminus\{H\}\bigr)\cup\head'\bigr)\theta.
\end{align*} 
We then add the rule $\vnf(\body''\rightarrow\head'')$.
\end{itemize}
Given a set of \dgtgds~$\Sigma$, let \iindex[not]{$\dgcloex(\Sigma)$} be the closure
of $\vnf(\skol(\shnf(\Sigma)))$ under the given inference rule.
The \emph{\iindex{disjunctive guarded saturation}} is then defined as
\index[not]{$\dgclo(\Sigma)$}$\dgclo(\Sigma)\coloneqq\ifc(\dgcloex(\Sigma))$.
\end{defn}

Before we verify the correctness of the saturation, we show that all derived rules are of a common form.
\begin{prop}\label{prop:dgtgd_guarded_closure_resolvent_properties}
Any $\sigma\in\dgcloex(\Sigma)$ is a single-headed guarded simple rule in variable normal form
and $\width(\sigma)\leq \bwidth(\shnf(\Sigma))$.
\end{prop} 
\begin{proof}
We prove the claim inductively.
Clearly, all rules in $\vnf(\skol(\shnf(\Sigma)))$ satisfy the conditions,
and all derived rules will be in~variable normal form due to the final $\vnf(\cdot)$ application.

Assume we create $\sigma=\vnf(\body''\rightarrow\head'')$ by \devolve from
some non-full ${\tau=\body\rightarrow\head}$ and some guarded simple
$\tau'=\body'\rightarrow\head'$ in $\dgcloex(\Sigma)$.
Let $H,B'$ be the resolved atoms in~$\head,\body'$, respectively,
let $\theta$ be the mgu of $H,B'$,
and without loss of generality, assume $\vars(\body)=\vec x$ and $\vars(\body')=\vec y$.
\begin{claim}\label{claim:dgtgd_guarded_closure_resolvent_properties_shnf}
$\sigma$ is single-headed
\end{claim} 
\begin{subproof}
By assumption, $\tau$ and $\tau'$ are single-headed, and as $\head''\subseteq(\head\cup\head')\theta$, also $\sigma$ will be single-headed.
\end{subproof} 
\begin{claim}\label{claim:dgtgd_guarded_closure_resolvent_properties_res_vars}
$H$ and $B'$ contain all variables of $\tau$ and $\tau'$, respectively.
\end{claim} 
\begin{subproof}
$H$ is functional and hence contains all variables of $\tau$.
$B'$ is either a guard of $\tau'$ or functional, and hence $B'$ contains all variables of $\tau'$.
\end{subproof} 
Since $H\theta=B'\theta$, it follows that $\vars(\vec x\theta)=\vars(\vec y\theta)$.
\begin{claim}\label{claim:dgtgd_guarded_closure_resolvent_properties_res_simple}
$H\theta$ and $B'\theta$ are simple.
\end{claim} 
\begin{subproof}
Let $H=R(t_1,\dotsc,t_a)$ and $B'=R(t_1',\dotsc,t_a')$ for some relation symbol $R$,
and pick a fresh variable $z$.
We consider two cases: either $B'$ is a function-free guard, or $B'$ is functional.
Assume the former.
We define a substitution $\theta'$ by
\begin{equation*}
v\theta'\coloneqq
\begin{cases}
	z,&\text{if }v\in\vec x\\
	t_i\theta',&\text{if }v=t_i'\in\vec y
\end{cases}.
\end{equation*} 
We show that $\theta'$ is well-defined; that is, if $y_k=t_i'=t_j'$, then $t_i\theta'=t_j\theta'$.
If $t_i,t_j\in\vec x$, clearly $t_i\theta'=t_j\theta'$.
Assume $t_i=f_1(\vec x)$ and $t_j=f_2(\vec x)$.
As $\theta$ is a unifier, we have $t_i\theta=t_i'\theta=t_j'\theta=t_j\theta$, and hence $f_1=f_2$. 
Consequently, $t_i\theta'=t_j\theta'$.
Lastly, we check the cases $t_i\in\vec x$ and $t_j=f(\vec x)$, or $t_j\in\vec x$ and $t_i=f(\vec x)$.
Without loss of generality, assume the former (the other case is symmetric).
Then $t_i\theta=t_i'\theta=t_j'\theta=t_j\theta=f(\vec x)\theta$.
Moreover, $\vars(t_j)=\vec x$ and $t_i\in\vec x$.
But then $t_i\theta\neq f(\vec x)\theta$, a contradiction.
Consequently, the case $t_i\in\vec x$ and $t_j=f(\vec x)$ cannot emerge.
Hence, $\theta'$ is well-defined.
Moreover, $t_i'\theta'=t_i\theta'$ by definition. 
That is, $\theta'$ unifies $H$ and $B'$, 
and further, $H\theta'$ and $B'\theta'$ are simple.
Now $\theta'=\theta\delta$ for some substitution $\delta$, and hence, $H\theta$ and $B'\theta$ must be simple.

Now assume $B'$ is functional.
We proceed similarly as in the previous case.
We define a substitution $\theta'$ defined by $v\theta'\coloneqq z$ for $v\in\vec x\cup\vec y$
and claim that $H\theta'=B'\theta$; that is, $t_i\theta'=t_i'\theta'$ for $1\leq i\leq a$.
Clearly, the claim holds for $t_i\in\vec x$ and $t_i'\in\vec y$.
If $t_i=f_1(\vec x)$ and $t_i'=f_2(\vec y)$, then again $f_1=f_2$ using that $\theta$ is a unifier,
and consequently, $t_i\theta'=t_i'\theta'$.
Lastly, we check the cases $t_i\in\vec x$ and $t_i'=f(\vec y)$, or $t_i=f(\vec x)$ and $t_i'\in\vec y$.
Without loss of generality, assume the former.
Then $t_i\theta=t_i'\theta=f(\vec y)\theta$. 
As $H$ is functional, there is $j\in\{1,\dotsc,a\}$ with $t_j=f'(\vec x)$,
and thus $t_j\theta=t_j'\theta=f'(\vec x)\theta$.
We have two cases:

Assume $t_j'=f'(\vec y)$.
Since $\vars(t_j)=\vec x$, this implies $t_i\theta=y_k\theta$ for some $y_k\in\vec y$.
Hence, $y_k\theta=t_i\theta=t_i'\theta=f(\vec y)\theta$.
But also $y_k\theta\neq f(\vec y)\theta$, a contradiction.
Now assume $t_j'=y_k$.
As $t_i\in\vec x$ and $y_k\theta=t_j'\theta=f'(\vec x\theta)$, 
we get that $t_i\theta$ is a subterm in $y_k\theta$.
But we also have $t_i\theta=t_i'\theta=f(\vec y\theta)$, leading to a contradiction.

Consequently, the case $t_i\in\vec x$ and $t_i'=f(\vec y)$ cannot emerge.
Hence, $\theta'$ unifies $H$ and $B'$, and further, $H\theta'$ and $B'\theta'$ are simple.
Thus, as before, $H\theta$ and $B'\theta$ must also be simple.
\end{subproof} 
\begin{claim}\label{claim:dgtgd_guarded_closure_resolvent_properties_res_vars_non_func}
If $A\in\{H,B'\}$ is functional, then $\vars(A)\theta$ is non-functional.
\end{claim} 
\begin{subproof}
Let $t$ be a functional term in $A$.
Then $t$ contains all variables of the corresponding rule.
Since $A\theta$ is simple, $t\theta$ is shallow, and thus $\vars(A)\theta$ is non-functional.
\end{subproof} 
As $H$ is functional, \cref{claim:dgtgd_guarded_closure_resolvent_properties_res_vars_non_func} implies that $\vars(\vec x\theta)=\vec x\theta$, and hence 
\begin{equation*}
	\width(\sigma)=|\vars(\vec x\theta)|=|\vec x\theta|\leq|\vec x|=\width(\tau)\leq\bwidth(\shnf(\Sigma)).
\end{equation*} 
\vspace{-1.5\baselineskip}
\begin{claim}\label{claim:dgtgd_guarded_closure_resolvent_properties_simple}
$\sigma$ is simple.
\end{claim} 
\begin{subproof}
By \cref{claim:dgtgd_guarded_closure_resolvent_properties_res_vars_non_func},
$\vec x\theta$ is non-functional, and thus, $\tau\theta$ is simple.
Similarly, if $B'$ is functional, $\vec y\theta$ is non-functional, and thus $\tau'\theta$ is simple.
If $B'$ is non-functional, then~$\tau'$ is full and $B'$ is a guard containing all variables of $\tau'$.
Now $B'\theta$ is simple by \cref{claim:dgtgd_guarded_closure_resolvent_properties_res_simple},
and hence also~$\tau'\theta$ must be simple.
\end{subproof} 
\begin{claim}\label{claim:dgtgd_guarded_closure_resolvent_properties_funcs}
Every functional term in $\sigma$ contains all variables.
\end{claim} 
\begin{subproof}
As $\vec x\theta$ is non-functional, all functional terms $t$ in $\tau\theta$ 
still contain all variables of $\vars(\vec x\theta)=\vars(\vec y\theta)$.
If $B'$ is functional, then $\vec y\theta$ is non-functional and the same reasoning applies.
If $\tau$ is full and $B'$ is a guard, let $t$ be a functional term in $\tau'\theta$.
Then there must be $y_i\in\vec y=\vars(B')$ such that $y_i\theta=t$.
As $B'\theta=H\theta$, $t$ occurs in $\tau\theta$, and hence, by our first case, contains all variables.
\end{subproof} 
\begin{claim}\label{claim:dgtgd_guarded_closure_resolvent_properties_guard}
$\sigma$ contains a function-free guard.
\end{claim} 
\begin{subproof}
Let $G$ be the non-functional guard of $\tau$.
Then $G\theta$ is not removed in $\sigma$.
As~$\vec x\theta$ is non-functional, $G\theta$ is non-functional,
and further, $G\theta$ contains all the variables of $\vec x\theta=\vars(\vec y\theta)$.
\end{subproof} 
Now by \cref{claim:dgtgd_guarded_closure_resolvent_properties_shnf,claim:dgtgd_guarded_closure_resolvent_properties_simple,claim:dgtgd_guarded_closure_resolvent_properties_funcs,claim:dgtgd_guarded_closure_resolvent_properties_guard}, $\sigma$ is a single-headed guarded simple rule.
\end{proof} 
\input{pages/dgtgd/dgtgd_guarded_correctness}

%% file: pages/dgtgd/dgtgd_guarded_correctness.tex
\subsection{Correctness}
To prove the correctness of the disjunctive guarded saturation,
we again first verify that the algorithm is sound.
\begin{lem}\label{lem:dgtgd_guarded_closure_ex_soundness}
$\dgcloex(\Sigma)$ is sound \wrt $\shnf(\Sigma)$.
\end{lem} 
\begin{proof}
By \cref{thm:Skolem_equisat}, it suffices to show that $\dgcloex(\Sigma)$ is sound \wrt $\skol(\shnf(\Sigma))$.
For this, it suffices to show that $\skol(\shnf(\Sigma))\models \dgcloex(\Sigma)$,
which is just a special case of the soundness proof for standard first-order logic resolution \citep{robinson1965machine}.
A detailed, formalised proof can be found in the Masters thesis of \cite{schlichtkrull2015formalization}.
\end{proof} 

\begin{cor}[Soundness]\label{cor:dgtgd_guarded_closure_soundness}
$\dgclo(\Sigma)$ is sound \wrt $\shnf(\Sigma)$.
\end{cor} 
\begin{proof}
By \cref{lem:dgtgd_guarded_closure_ex_soundness}, $\dgcloex$ is sound \wrt $\shnf(\Sigma)$,
and by \cref{lem:guarded_simple_ifc_sound}, $\dgclo(\Sigma)=\ifc(\dgcloex(\Sigma))$ is sound \wrt $\dgcloex(\Sigma)$.
\end{proof} 

To prove (evolve)-completeness of the algorithm, we want to talk about
chase steps using full as well as non-full rules from $\dgcloex(\Sigma)$.
The non-full rules of $\dgcloex(\Sigma)$, however,
do not contain existential variables for which a fresh constant
can be introduced in a chase step;
instead, each existential variable got replaced by a fresh Skolem function.
We thus generalise the notion of chase steps in the natural way. That is, we allow a trigger $h$ to map to ground terms with constants from $\consts(I)$.
In fact, during our completeness proof, we will only have to consider chase trees that introduce simple facts, like the one in the following example.

\begin{exmpl}\label{exmpl:chase_tree_guarded}
Given the database $\db=\{R(c,d)\}$ and a set of guarded simple rules $\Sigma$
\begin{align*}
R(x_1,x_2)&\rightarrow S(x_1,f_{1,1}(x_1,x_2))\lor T(x_1,x_2),\\
T(x_1,x_2)&\rightarrow U(x_1,x_2,f_{2,1}(x_1,x_2)),\\
U(x_1,x_2,x_3)&\rightarrow M(x_1)\lor P(x_2),\\
S(x_1,x_2)&\rightarrow M(x_2),
\end{align*} 
the tree depicted in \cref{fig:exmpl_chase_tree_guarded} is a chase tree from~$\db$ and $\Sigma$.
\begin{figure}[ht]
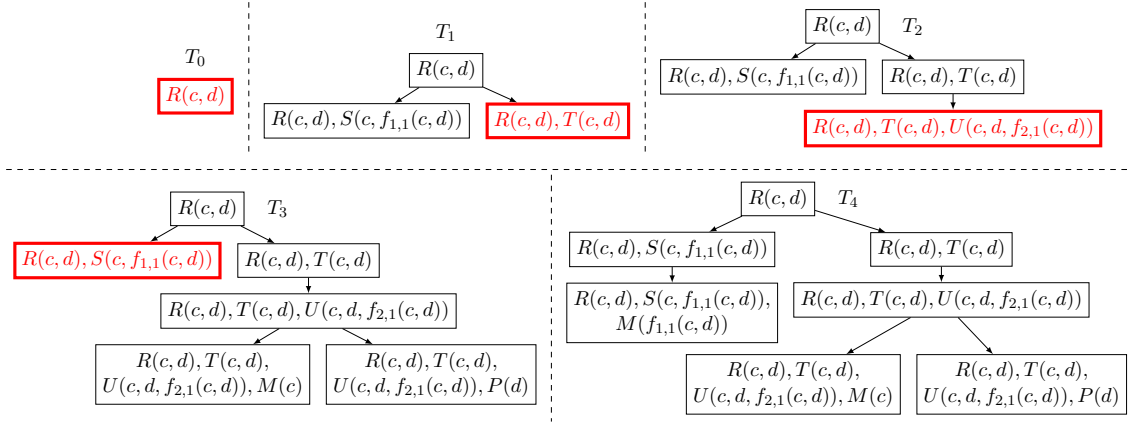

	\centering
	\includestandalone[width=\textwidth]{./figures/exmpl_chase_tree_guarded}
	\caption{Chase tree for \cref{exmpl:chase_tree_guarded};
	nodes that will be extended are marked bold and red.}
	\label{fig:exmpl_chase_tree_guarded}
\end{figure}
\end{exmpl} 

Now recall the definition of evolve-completeness:
\dgtgdabstrcompletenessprops*

The first property again follows immediately:
\begin{lem}\label{lem:dgtgd_guarded_closure_prop_ifc}
For every $\tau\in\ifc(\shnf(\Sigma))$, we have $\vnf(\tau)\in\dgcloex(\Sigma)$.
\end{lem} 
\begin{proof}
Follows from the base case and final step of the algorithm.
\end{proof} 

The proof of the second property requires two lemmas.
Broadly speaking, the first lemma shows that whenever a rule $\tau$ triggers in a set of simple facts $N$,
we can create a corresponding rule $\tau^*$ that triggers in a set of function-free facts and whose consequences contain those of $\tau$, provided that we have a rule $\tau_F$ that creates $F$ and triggers in a function-free set of facts for each $F\in \func{N}$:

\begin{lem}\label{lem:dgtgd_guarded_closure_prop_devolve}
Let $M,N,S$ be sets, $F_1,\dotsc,F_n$ be simple facts,
$\tau\in\dgcloex(\Sigma)$,
and for any $F\in \func{N}$, let $\tau_F\in\dgcloex(\Sigma)$ be non-full such that
\begin{assmlist}
\item \label{assm:dgtgd_guarded_closure_prop_devolve_1}
$M\supseteq \nfunc{N}$ and $M$ is function-free,
\item \label{assm:dgtgd_guarded_closure_prop_devolve_2}
$N,S$ are simple and $\consts(N),\consts(S)\subseteq\consts(M)$,
\item \label{assm:dgtgd_guarded_closure_prop_devolve_3}
$\tau=\body\rightarrow\head$ is full or has a functional body atom,
\item \label{assm:dgtgd_guarded_closure_prop_devolve_4}
firing~$\tau$ in~$N$ produces $\bigvee_{i=1}^n F_i$,
\item \label{assm:dgtgd_guarded_closure_prop_devolve_5}
firing $\tau_F$ in $M$ produces $F\lor \bigvee_{F'\in S_F} F'$ for some $S_F\subseteq S$
\end{assmlist} 
Then there exists $\tau^*\in\dgcloex(\Sigma)$
such that firing $\tau^*$ in $M$ produces $\bigvee_{i=1}^{n} F_i\lor \bigvee_{F'\in S'}F'$
for some $S'\subseteq S$.
\end{lem}
\begin{proof}
We show the claim by induction on $k\coloneqq|U|$, where $U\coloneqq\{B\in\body\mid h(B)\in \func{N}\}$.
In case $k=0$, we just set $\tau^*\coloneqq\tau$.

In case $k>0$, let $h$ be the trigger obtained from \cref{assm:dgtgd_guarded_closure_prop_devolve_4} for $\tau$.
\begin{claim}\label{claim:dgtgd_guarded_closure_prop_devolve_1}
There is $B\in U$ such that $\tau$ is full and $B$ guards $\tau$, or $B$ is functional.
\end{claim} 
\begin{subproof}
If there is a functional $B\in U$, we are done.
Otherwise, $\tau$ must be full. Let $G$ be the guard of $\tau$.
As $k>0$, $h$ unifies some variable in $\tau$ with a functional term.
Hence, $h(G)$ is functional, and thus $G\in U$.
\end{subproof} 
Fix some $B\in U$ obtained by the just proven claim and set $F\coloneqq h(B)$. 
Pick the corresponding $\tau_F=\body_F\rightarrow\head_F$ and the trigger $h_F$ 
obtained from \cref{assm:dgtgd_guarded_closure_prop_devolve_5}.
Then there is $H_F\in\head_F$ such that $h_F(H_F)=h(B)$.
Without loss of generality, assume $\vars(\body_F)=\vec x$, $\vars(\body)=\vec y$, and $h_F(\vec x)=\vec c$.
Then the substitution $\theta$ defined by
\begin{equation*}
	v\theta\coloneqq\begin{cases}
		x_i,&\text{if }v\in\vec x \text{ and }h_F(v)=c_i\\
		x_i,&\text{if }v\in\vec y \text{ and }h(v)=c_i\\
		f(x_{i_1},\dotsc,x_{i_a}),&\text{if }v\in\vec y \text{ and }h(v)=f(c_{i_1},\dotsc,c_{i_a})
	\end{cases}
\end{equation*} 
is a unifier of $H_F$ and $B$.
Hence, there is an mgu $\theta^*$ of $H$ and $B$, and thus, there is a substitution $\delta$ such that $\theta^*=\theta\delta$.
\begin{claim}\label{claim:dgtgd_guarded_closure_prop_devolve_2}
$H_F$ is a functional atom.
\end{claim} 
\begin{subproof}
By \cref{assm:dgtgd_guarded_closure_prop_devolve_1}, $M$ is function-free,
and hence, $h_F(\vec x)$ is function-free;
however,~$F=h_F(H_F)$ is functional, and thus $H_F$ must be functional.
\end{subproof} 
Using \cref{claim:dgtgd_guarded_closure_prop_devolve_1,claim:dgtgd_guarded_closure_prop_devolve_2},
we can see that all requirements for \devolve on $\tau_F,\tau$ are satisfied.
Thus, we obtain $\tau^*=(\body^*\rightarrow\head^*)=\vnf(\body'\rightarrow\head')\in\dgcloex(\Sigma)$, where
\begin{align*}
	\body'&=\bigl(\body_F\cup\bigl(\body\setminus\{B\}\bigr)\bigr)\theta^*,\\
	\head'&=\bigl(\bigl(\head_F\setminus\{H_F\}\bigr)\cup\head\bigr)\theta^*.
\end{align*} 
Let $\rho$ be the renaming used by $\vnf(\cdot)$ for $\body'\rightarrow\head'$,
let $h'$ be the trigger defined by $h'(x_i)\coloneqq c_i$,
and set $h^*\coloneqq h'\circ\delta\circ\rho^{-1}$.
\begin{claim}\label{claim:dgtgd_guarded_closure_prop_devolve_3}
For any $x_i\in\vec x$, we have $h'(x_i\theta)=h_F(x_i)$, 
and for any $y_i\in\vec y$, we have $h'(y_i\theta)=h(y_i)$. 
\end{claim}
\begin{subproof}
First note that $\dom(\theta)=\vec x\cup\vec y$ by \cref{assm:dgtgd_guarded_closure_prop_devolve_1,assm:dgtgd_guarded_closure_prop_devolve_2}.
Pick some $x_i\in\vec x$ and assume $h_F(x_i)=c_j$. 
Then $x_i\theta=x_j$ by definition of $\theta$.
Further, by definition of $h'$, $h'(x_j)=c_j$.
Hence, $h'(x_i\theta)=h'(x_j)=c_j=h_F(x_i)$.

Now pick some $y_i\in\vec y$.
The case $h(y_i)=c_j$ is analogous to the previous one.
In case $h(y_i)=f(c_{j_1},\dotsc,c_{j_a})$,
we get $y_i\theta=f(x_{j_1},\dotsc,x_{j_a})$ and $h'(f(x_{j_1},\dotsc,x_{j_a}))=f(c_{j_1},\dotsc,c_{j_a})$.
\end{subproof} 
\begin{claim}
Firing~$\tau^*$ based on $h^*$ in $M\cup N$ produces 
$\bigvee_{i=1}^{n} F_i\lor \bigvee_{F'\in S_F} F'$.
\end{claim} 
\begin{subproof}
We have 
\begin{align*}
h^*(\body^*)&=h'\circ\delta\circ\rho^{-1}(\body^*)=h'\circ\delta(\body')
=h'\Bigl(\bigl(\body_F\cup\bigl(\body\setminus\{B\}\bigr)\bigr)\theta^*\delta\Bigr)
\subseteq h'\bigl((\body_F\cup\body)\theta^*\delta\bigr)\\ 
&=h'\bigl((\body_F\cup\body)\theta\bigr)
=h'(\body_F\theta\cup\body\theta)
=h'(\body_F\theta)\cup h'(\body\theta)
\stackrel{\cref{claim:dgtgd_guarded_closure_prop_devolve_3}}{=}h_F(\body_F)\cup h(\body)\\
&\subseteq M\cup N.
\end{align*} 
and
\begin{align*}
h^*(\head^*)&\stackrel{\phantom{\cref{claim:dgtgd_guarded_closure_prop_devolve_3}}}{=}h'\circ\delta\circ\rho^{-1}(\head^*)=h'\circ\delta(\head')
=h'\Bigl(\bigl(\bigl(\head_F\setminus\{H_F\}\bigr)\cup\head\bigr)\theta\Bigr)\\
&\stackrel{\cref{claim:dgtgd_guarded_closure_prop_devolve_3}}{=}h'\bigl(\bigl(\head_F\setminus\{H_F\}\bigr)\theta\bigr)\cup h(\head)
=h'\bigl(\bigl(\head_F\setminus\{H_F\}\bigr)\theta\bigr)\cup\{F_1,\dotsc,F_n\}\\
&\stackrel{\cref{claim:dgtgd_guarded_closure_prop_devolve_3}}{=}h_F\bigl(\head_F\setminus\{H_F\}\bigr)\cup \{F_1,\dotsc,F_n\}
=S_F\cup\{F_1,\dotsc,F_n\}.\qedhere
\end{align*} 
\end{subproof} 
Now if $h^*(\tau^*)\subseteq M$, we are done.
Hence assume there is $B\in\body^*$ with $h^*(B)\in \func{N}$.
\begin{claim}
For any $B\in\body^*$ with $h^*(B)\in \func{N}$,
we have $B\in(\body\setminus\{B\})\theta^*$.
Moreover, $\tau^*$ is full or has a functional body atom.
\end{claim} 
\begin{subproof}
By \cref{assm:dgtgd_guarded_closure_prop_devolve_1}, $M$ is function-free.
Hence, no atom in $\body_F\theta^*$ unifies with an atom in $\func{N}$ via $h^*$, 
and thus $B\in(\body\setminus\{B\})\theta^*$.

Next, pick some $B\in\body^*$ with $h^*(B)\in \func{N}$.
Now either $B$ is functional, and we are done, 
or~$h^*$ maps some $x\in\vars(B)$ to a functional term.
For the sake of contradiction, assume $\tau^*$ is non-full.
Then there is a functional head atom $H$ in $\tau^*$ that contains $x$.
Hence, $h^*(H)$ is non-simple, contradicting the fact that $h^*(\head)$ is simple.
\end{subproof} 
We have $|(\body\setminus\{B\})\theta^*|<k$. 
Thus, we can apply the inductive hypothesis on $\tau^*$, which finishes the proof.
\end{proof}

The second lemma basically provides an inductive proof of \cref{defn:dgtgd_abstr_completeness_props_2};
however, in order to make the induction work, we have to show a more general claim.
\begin{lem}\label{lem:dgtgd_guarded_pull_induct}
Let $T$ be a chase tree from $\dgcloex(\Sigma)$ with root~$N$,
let $M,S$ be sets, 
for any $F\in \func{N}$, let $\tau_F\in\dgcloex(\Sigma)$ be non-full,
and if the first fired rule is non-full,
let $\tau\in\dgcloex(\Sigma)$ such that
\begin{assmlist}
\item the first rule fired in $T$ is full or non-full, and all subsequent rules are full,
\item $M\supseteq \nfunc{N}$ and $M$ is function-free,
\item $N,S$ are simple and $\consts(N),\consts(S)\subseteq\consts(M)$,
\item\label{assm:dgtgd_guarded_pull_induct_4} 
firing $\tau_F$ in $M$ produces $F\lor \bigvee_{F'\in S_F} F'$ for some $S_F\subseteq S$,
\item if the first fired rule is non-full, it produces $\bigvee_{i=1}^{n}F_i$
for some simple $F_1,\dotsc,F_n$,
\item $\tau$ is non-full,
\item\label{assm:dgtgd_guarded_pull_induct_7}
firing~$\tau$ in~$\nfunc{N}$ produces $\bigvee_{i=1}^n F_i\lor\bigvee_{F'\in S'}F$ for some $S'\subseteq S$.
\end{assmlist} 
Then there exists a chase tree $T^*$ from $M$ and $\dgcloex(\Sigma)$ 
such that for any $L^*\in\lvs(T^*)$,
\begin{thmlist}
\item $L^*\models\nfunc{\lvs(T)}$ and $L^*$ is function-free, or
\item $L^*=\{F\}\cup U_{L^*}$ for some $F\in S$ 
and $U_{L^*}\subseteq M\cup\bigcup_{L\in\lvs(T)}\nfunc{L}$.
\end{thmlist} 
\end{lem} 
\begin{proof}
We prove the claim by induction on the number of steps used to create $T$, 
denoted by $k$.
In case $k=0$, we let $T^*$ be the single node $M$.

In case $k>0$,
let $\sigma$ be the first rule fired in $T_k$, producing $\bigvee_{i=1}^nF_i$.
\begin{claim}
There is $\sigma'\in\dgcloex(\Sigma)$ such that firing $\sigma'$ in $M$ produces
$\bigvee_{i=1}^n F^i\lor\bigvee_{F'\in S'}F'$ for some $S'\subseteq S$.
\end{claim} 
\begin{subproof}
If $\sigma$ is non-full, we can pick $\tau$ by \cref{assm:dgtgd_guarded_pull_induct_7}.
If $\sigma$ is full, by \cref{lem:dgtgd_guarded_closure_prop_devolve}, 
there is $\sigma'\in\dgcloex(\Sigma)$ such that firing~$\sigma'$
in $M$ produces $\bigvee_{i=1}^n F^i\lor\bigvee_{F\in S'}F$ for some $S'\subseteq S$.
\end{subproof} 
Let $T^*$ be the chase tree created by firing $\sigma'$ in $M$.
We now iteratively modify $T^*$ for $0\leq i\leq n$ 
while maintaining the following invariant:

\medskip
After step $i$, for any $L^*\in\lvs(T^*)$, we have
\begin{invarlist}[nolistsep]
\item \label{invar:dgtgd_guarded_pull_induct_1}
$L^*\models U^i$ and $L^*$ is function-free, 
where $U^i\coloneqq \bigcup_{j=1}^i\nfunc{\lvs\bigl(T_k^{N^j}\bigr)}$, or
\item \label{invar:dgtgd_guarded_pull_induct_2}
$L^*=\{F\}\cup U_{L^*}^i$ for some $F\in (\{F_{i+1},\dotsc,F_n\}\cup S)$ 
and $U_{L^*}^i\subseteq M\cup\bigcup_{L\in U^i}L$.
\end{invarlist} 
\medskip

Note that our initial $T^*$ already satisfies the invariant for $i=0$.
In case $i>0$, for any leaf~$L^*$ in~$T^*$ with $L^*=\{F_i\}\cup U_{L^*}^{i-1}$, 
let $N_{L^*}$ be the ancestor of $L^*$ in which some rule~$\tau_{L^*}$
producing $F_i\lor\bigvee_{F\in V_{L^*}} F$ was fired for some set $V_{L^*}$.
We modify the subtree of any $N_{L^*}$ in the following way:

Pick some $N_{L^*}$ with maximum distance to the root.
Due to our invariant, $U_{L^*}^{i-1}$ is function-free, and hence, so is $N_{L^*}$.
As a consequence, $V_{L^*}$ is simple.
\begin{claim}
For any $F\in \func{N^i}$, there is a non-full~$\tau_F$  such that
firing $\tau_F$ in $\func{N^i}$ produces $F\lor \bigvee_{F'\in S_F} F'$ for some $S_F\subseteq S\cup V_{L^*}$.
\end{claim} 
\begin{subproof}
For any $F\in \func{N}\cap \func{N^i}$, there is a rule by \cref{assm:dgtgd_guarded_pull_induct_4},
and if $F_i$ is functional, we can additionally use the 
non-full $\tau_{L^*}$ as $\tau_{F_i}$.
\end{subproof} 
Thus, we can apply the inductive hypothesis on $T_k^{N_i}$ to obtain a chase tree $(T')^i$ from $N_{L^*}$
and $\dgcloex(\Sigma)$ such that for any $L'\in\lvs((T')^i)$,
either $L'\models\nfunc{\lvs\bigl(T_k^{N^i}\bigr)}$ and $L'$ is function-free,
or $L'=\{F\}\cup U_{L'}$ for some $F\in(V_{L^*}\cup S)$ and 
\begin{equation*}
U_{L'}\subseteq \Biggl(N_{L^*}\cup\bigcup_{L\in\lvs\bigl(T_k^{N^i}\bigr)}\nfunc{L}\Biggr)
\subseteq\Biggl(M\cup\bigcup_{L\in U^{i-1}}L\cup\bigcup_{L\in\lvs\bigl(T^{N^i}\bigr)}\nfunc{L}\Biggr)
\subseteq \Biggl(M\cup\bigcup_{L\in U^i}L\Biggr).
\end{equation*} 
For any $L'\in\lvs((T')^i)$ with $L'=\{F\}\cup U_{L'}$ and $F\in V_{L^*}$,
there is a child $N_F$ of $N_{L^*}$ in~$T^*$ containing $F$.
As $T^*$ satisfies the invariant for $i-1$ and due to maximality of $N_{L^*}$, 
the leaves of $(T^*)^{N_F}$ already satisfy the invariant for step $i$.
Hence, we can attach $(T^{*})^{N_F}$ to any such~$L'$
to obtain $(T^*)^i$ that satisfies the invariant for step $i$.
Finally, we attach $(T^*)^i$ to $N_{L^*}$
and proceed with the next choice for $N_{L^*}$.
Once all $N_{L^*}$ are processed, we continue with step $i+1$.

Once step $n$ is finished, we obtain the desired chase tree $T^*$, noting that $U^n=\nfunc{\lvs(T)}$.
\end{proof} 

\begin{prop}[Completeness]\label{prop:dgtgd_guarded_closure_evolve_complete}
$\dgclo(\Sigma)$ is evolve-complete \wrt $\shnf(\Sigma)$.
\end{prop} 
\begin{proof}
\cref{defn:dgtgd_abstr_completeness_props_1} is \cref{lem:dgtgd_guarded_closure_prop_ifc}.
For \cref{defn:dgtgd_abstr_completeness_props_2}, 
note that~$N$ is a set of facts, which means $\func{N}=\emptyset$.
Hence, we can instantiate \cref{lem:dgtgd_guarded_pull_induct} (with $M=N,S=\emptyset$) to obtain 
a chase tree~$T^*$ from~$N$ and $\dgcloex(\Sigma)$ 
such that $L^*\models\nfunc{\lvs(T)}$ and $L^*$ is function-free for any $L^*\in\lvs(T^*)$.
Thus, $T^*$ must be created from full rules only,
and any such rule is in $\dgclo(\Sigma)$.
Hence, $T^*$ is a chase tree from $N$ and $\dgclo(\Sigma)$ with $T^*\models Q^T$.
\end{proof}

%% file: figures/exmpl_chase_tree_guarded.tex
\begin{tikzpicture}
\tikzset{node distance=0.3cm,
	r-0/.pic={\node (r) {$R(c,d)$};},
	T0/.pic={\pic[trigger]{r-0};},
	Nl/.pic={\node (Nl) {$R(c,d),S(c,f_{1,1}(c,d))$};},
	Nr/.pic={\node (Nr) {$R(c,d),T(c,d)$};},
	T1/.pic={%
		\pic{r-0};
		\pic[below=of r,xshift=-1.5cm]{Nl};
		\pic[trigger,right=of Nl]{Nr};
		\path (r) edge (Nl)
		(r) edge (Nr);
	},
	Nrm/.pic={\node (Nrm) {$R(c,d),T(c,d),U(c,d,f_{2,1}(c,d))$};},
	T2/.pic={%
		\pic{r-0};
		\pic[below=of r,xshift=-1.5cm]{Nl};
		\pic[right=of Nl]{Nr};
		\pic[trigger,below=of Nr]{Nrm};
		\path (r) edge (Nl)
		(r) edge (Nr)
		(Nr) edge (Nrm);
	},
	Nrml/.pic={\node[align=center] (Nrml) {$R(c,d),T(c,d),$\\$U(c,d,f_{2,1}(c,d)),M(c)$};},
	Nrmr/.pic={\node[align=center] (Nrmr) {$R(c,d),T(c,d),$\\$U(c,d,f_{2,1}(c,d)),P(d)$};},
	T3/.pic={%
		\pic{r-0};
		\pic[trigger,below=of r,xshift=-1.7cm]{Nl};
		\pic[right=of Nl]{Nr};
		\pic[below=of Nr]{Nrm};
		\pic[below=of Nrm,xshift=-2cm]{Nrml};
		\pic[right=of Nrml]{Nrmr};
		\path (r) edge (Nl)
		(r) edge (Nr)
		(Nr) edge (Nrm)
		(Nrm) edge (Nrml)
		(Nrm) edge (Nrmr);
	},
	Nlm/.pic={\node[align=center] (Nlm) {$R(c,d),S(c,f_{1,1}(c,d)),$\\$M(f_{1,1}(c,d))$};},
	T4/.pic={%
		\pic{r-0};
		\pic[below=of r,xshift=-2cm]{Nl};
		\pic[below=of Nl]{Nlm};
		\pic[right=1.8cm of Nl]{Nr};
		\pic[below=of Nr]{Nrm};
		\pic[below=0.7cm of Nrm,xshift=-2.8cm]{Nrml};
		\pic[right=of Nrml]{Nrmr};
		\path (r) edge (Nl)
		(r) edge (Nr)
		(Nr) edge (Nrm)
		(Nrm) edge (Nrml)
		(Nrm) edge (Nrmr)
		(Nl) edge (Nlm);
	},
}
\node[scope] (T0) {
	\begin{tikzpicture}
		\pic{T0};
		\node[lnode,above=0.1cm of r]{$T_0$};
	\end{tikzpicture}
};
\node[scope,right=of T0] (T1) {
	\begin{tikzpicture}
		\pic{T1};
		\node[lnode,above=0.1cm of r]{$T_1$};
	\end{tikzpicture}
};
\node[scope,right=of T1] (T2) {
	\begin{tikzpicture}
		\pic{T2};
		\node[lnode,right=of r]{$T_2$};
	\end{tikzpicture}
};
\node[scope,below=1.2cm of T0,xshift=1.5cm] (T3) {
	\begin{tikzpicture}
		\pic{T3};
		\node[lnode,right=of r] {$T_3$};
	\end{tikzpicture}
};
\node[scope,right=of T3] (T4) {
	\begin{tikzpicture}
		\pic{T4};
		\node[lnode,right=of r] {$T_4$};
	\end{tikzpicture}
};
\pic{drawv=T0/T1/T2/T2};
\pic{drawv=T1/T2/T2/T2};
\pic{drawvr=T3/T4};
\pic{drawh=T2/T3/T3/T4};
\end{tikzpicture}

%% file: pages/dgtgd/dgtgd_guarded_complexity.tex
\subsection{Complexity Analysis}

Let us now examine the complexity bounds of our algorithm.
In the following, let~$n$ be the number of predicate symbols in $\shnf(\Sigma)$, 
$m$ be the number of Skolem functions in $\skol(\shnf(\Sigma))$,
$a$ be the maximum arity of any predicate in $\shnf(\Sigma)$, and $w\coloneqq\bwidth(\shnf(\Sigma))$.

\begin{lem}
The maximum arity of any function in $\skol(\shnf(\Sigma))$ is bounded by $w$.
\end{lem} 
\begin{proof}
Follows straight from \cref{defn:dtgd_Skolemisation}.
\end{proof} 

\begin{lem}\label{lem:gcloex_atom_number}
Any $\tau\in\dgcloex(\Sigma)$ contains at most $2nw^{aw}a^m$ atoms.
\end{lem} 
\begin{proof}
By \cref{prop:dgtgd_guarded_closure_resolvent_properties},
$\tau$ is a single-headed guarded simple rule in variable normal form with $\width(\tau)\leq w$.
For any fixed predicate symbol, we can thus choose at most $a$ function symbols
and $a\cdot w$ variables.
Further, any atom can occur in the body or the head of~$\tau$.
Hence, $\tau$ contains at most $2nw^{aw}m^a$ atoms.
\end{proof}

\begin{cor}\label{cor:dgtgd_guarded_closure_size}
$|\dgcloex(\Sigma)|\leq 2^{2nw^{aw}a^m}$.
\end{cor} 
\begin{proof}
A rule consists of a subset of all possible body and head atoms.
Now the claim follows by \cref{lem:gcloex_atom_number}.
\end{proof} 

\begin{cor}[Correctness]\label{cor:dgtgd_guarded_closure_atomic_rewriting}
The disjunctive guarded saturation $\dgclo(\Sigma)$ is an atomic rewriting of~$\Sigma$.
\end{cor} 
\begin{proof}
By \cref{cor:dgtgd_abstr_atomic_rewriting,cor:dgtgd_guarded_closure_soundness,prop:dgtgd_guarded_closure_evolve_complete}, $\dgclo(\Sigma)$ is an atomic rewriting of $\shnf(\Sigma)$, and hence, $\dgclo(\Sigma)$ is an atomic rewriting of $\Sigma$ by \cref{lem:dgtgd_shnf_equiv}.
\end{proof}

\begin{lem}\label{lem:dgtgd_guarded_closure_devolve_complexity}
Let $\tau,\tau'\in\dgcloex(\Sigma)$.
Then all \devolve-resolvents of ${\tau=\body\rightarrow\head}$, ${\tau'=\body'\rightarrow\head'}$
can be obtained in time doubly exponential in $a,m,w$.
\end{lem} 
\begin{proof}
Let $H\in\head$ be functional and $B'\in\body'$ be functional or a guard.
We check whether there is an mgu $\theta$ of $H$ and $B'$,
and if this is the case, we resolve the rules and normalise the result to
obtain a new rule.
Using \cref{lem:mgu_complexity,lem:guarded_simple_vnf_complexity}, these steps can be done in $\calO(\size(\tau)+\size(\tau'))$.
Since $\tau$ contains at most $s\coloneqq nw^{aw}a^m$ body atoms and $\tau'$ at most~$s$ head atoms, we have to repeat this step at most $s^s$ times.
Hence, the complexity is doubly exponential in $a,m,w$.
\end{proof} 

\begin{lem}\label{lem:dgtgd_guarded_closure_complexity}
The disjunctive guarded saturation algorithm terminates in time doubly exponential in $a,m,w$.
\end{lem} 
\begin{proof}
The initial Skolemisation and variable and single head normal from transformations can be done in linear time.
The algorithm then takes all pairs $\tau,\tau'$ in $\gcloex(\Sigma)$ and performs all \devolve-inferences.
By \cref{cor:dgtgd_guarded_closure_size}, this process takes at most~$c\coloneqq 2^{2nw^{aw}a^m}$ iterations,
and at any point, there are at most $\calO(c^2)$ pairs.
By \cref{lem:dgtgd_guarded_closure_devolve_complexity}, 
the complexity for each $\tau,\tau'$ is doubly exponential in $a,m,w$, which finishes the proof.
\end{proof} 

Finally, we can combine our results to obtain a decision procedure.
\begin{thm}\label{thm:dgtgd_guarded_dec_proc}
The disjunctive guarded saturation algorithm provides a decision procedure for the \qfqaprobdgtgds. 
The procedure terminates in 
$\exptime$ for fixed $\Sigma$
and $\texptime$ otherwise.
\end{thm} 
\begin{proof}
The size of $a,m,w$ is linear in $\size(\Sigma)$.
Now the claim follows from \cref{prop:dgtgd_atomic_rewriting_dec_proc,cor:dgtgd_guarded_closure_atomic_rewriting,lem:dgtgd_guarded_closure_complexity}.
\end{proof} 

Once more, our bounds already match for atomic queries:
\begin{thm}[\citep{gottlob2012complexity}]
Given a set of \dgtgds $\Sigma$,
the Atomic Query Answering Problem under $\Sigma$ is 
$\exptime$-complete for fixed~$\Sigma$,
and $\texptime$ in general.
\end{thm}

%% file: chapters/comparison.tex
\chapter{Comparison with \gf Resolution}\label{chap:comparison}

As we now have different decision procedures for the \qfqaprob under (disjunctive) \gtgds at our hand,
we naturally want to compare them with existing procedures from the literature.
Recall that our work's main motivation was to explain completeness of the resolution
procedures for the guarded fragment, as described by \cite{de1998resolution} and \cite{ganzinger1999superposition},
by using the tree-like model property.
We hence first examine which fragments of \gf we can cover with our chosen framework
and then elaborate how our algorithms compare with the \gf resolution procedures.

\input{pages/comparison/gf_vs_qfqaprob}
\input{pages/comparison/gf_res_vs_guarded_sat}

%% file: pages/comparison/gf_vs_qfqaprob.tex
\section{GF Satisfiability vs.\ the Quantifier-Free Query Answering Problem}\label{sec:gf_sat_vs_qfqap}

We have already shown in \cref{lem:gf_reduce} that the \qfqaprob under \dgtgds reduces
to the \gf satisfiability problem.
To identify which fragment of \gf we can capture with our framework, we start comparing our algorithms
with the \gf resolution procedures given in \citep{de1998resolution,ganzinger1999superposition}.
As mentioned before, both procedures there use a constant-free definition of \gf, the latter with equality.
Since we do not support equality in our procedures, we focus on \gf sentences without equality and constants.
Both \gf procedures apply an initial clausal normal form transformation to a given \gf sentence $\phi$
that creates clauses of the following form:
\begin{defn}[Guarded simple clause]
A clause $C$ is a \emph{guarded simple clause} if
\begin{propylist}
\item each literal $(\lnot)L$ is simple,
\item every functional subterm in $C$ contains all the variables of $C$ and no constants, and
\item if $C$ is non-ground, $C$ contains a non-functional negative literal that contains all the variables of $C$.
\end{propylist} 
\end{defn} 
\begin{rmk}
The definitions in referred articles are slightly more general 
as they also account for clauses that are derivable during the resolution process.
\end{rmk} 
Which fragment of guarded simple clauses can we capture with our framework consisting of a database $\db$,
\dgtgds $\Sigma$, and a quantifier-free \ubcq $Q$?
Fix a set~$S$ of guarded simple clauses, let $S_Q$ consist of all purely negative, ground $C=\{\lnot L_1,\dotsc, \lnot L_n\}\in S$,
and set $S^-\coloneqq (S\setminus S_Q)$.
Then $S$ is unsatisfiable \ifftext $S^-\models \lnot S_Q$,
where $\lnot S_Q$ is obtained by negating the formula corresponding to the clause set $S_Q$.
Note that $\lnot S_Q$ is a positive, ground formula in disjunctive normal form;
hence, $\lnot S_Q$ can be modelled by~$Q$.
Furthermore, any ground $C=\{\lnot L_1,\dotsc,\lnot L_n,L_1',\dotsc,L_m'\}\in S^-$ containing at least one positive literal and using constants $\vec c$ 
can be modelled by adding a full \dgtgd $R(\vec x)\land\bigwedge_{i=1}^nL_i(\vec x)\rightarrow\bigvee_{i=1}^mL_i'(\vec x)$
to $\Sigma$ and a fact $R(\vec c)$ to $\db$, 
where $R$ is a fresh relation symbol of arity $|\vec c\,|$.

What about non-ground clauses?
Recall that the Skolemisation of a set of single-headed \dgtgd results in a set 
of guarded simple rules with non-functional bodies.
Moreover, any introduced Skolem function cannot occur in more than one rule.
Consequently, we are able to model sets $U\subseteq S^-$ of non-ground clauses 
in which any clause $C$ does not contain a functional negative literal and any function occurs in at most one clause.
To see this, take a clause $C=\{\lnot L_1,\dotsc,\lnot L_n,L_1',\dotsc,L_m'\}\in U$ using functions $f_1,\dotsc,f_k$
and variables~$\vec x$.
Then $C$ corresponds to the rule $\bigwedge_{i=1}^nL_i(\vec x)\rightarrow\bigvee_{i=1}^mL_i'(\vec x)$,
which by \cref{thm:Skolem_equisat} is satisfiable \ifftext the rule
\begin{align}
&\bigwedge_{i=1}^nL_i(\vec x)\rightarrow\exists y_1,\dotsc,y_k\left(\bigvee_{i=1}^mL_i'(\vec x)[y_1/f_1(\vec x),\dotsc,y_k/f_k(\vec x)]\right)\nonumber\\
\equiv\quad&\bigwedge_{i=1}^nL_i(\vec x)\rightarrow\bigvee_{i=1}^m\exists y_1,\dotsc,y_k\,L_i'(\vec x)[y_1/f_1(\vec x),\dotsc,y_k/f_k(\vec x)]\label{eq:comp_deskol_dgtgd}
\end{align} 
is satisfiable, where $L_i'(\vec x)[y_1/f_1(\vec x),\dotsc,y_k/f_k(\vec x)]$ is the formula 
obtained by replacing any term $f_j(\vec x)$ in $L_i'$ by the variable $y_j$.
Now Sentence~\eqref{eq:comp_deskol_dgtgd} is a \dgtgd, and hence $U$ can be modelled by $\Sigma$.
Note that the disjunctive guarded saturation in \cref{sec:dgtgd_guarded_saturation} is defined for a larger set 
of clauses (for example, guarded simple clauses with functional bodies);
however, our completeness proof assumes that the initial set of rules corresponds to a set of \dgtgds.
Thus, any $S^-$ containing a functional negative literal
or a function that occurs in more than one clause is not reducible to our algorithm using our given results.

%% file: pages/comparison/gf_res_vs_guarded_sat.tex
\section{GF Resolution vs.\ the (Disjunctive) Guarded Saturation}\label{sec:gf_vs_guarded_sat}

In this section, we want to elaborate and demonstrate the differences between 
our derived algorithms and the known \gf resolution procedures when tackling the \qfqaprob under (disjunctive) \gtgds.

In the case of \dgtgds, the \gf resolution procedures proceed the same in both referred articles.
In fact, they apply the same selection and ordering strategy as we did for the disjunctive guarded saturation.
That is, a literal $L$ is selected in a clause $C$ \ifftext $L$ 
is functional and negative, $L$ is functional and $C$ contains no functional negative literal,
or $L$ is a guard and $C$ is non-functional.
This fact is, in some way, surprising as we started from a different, more specialised problem setting,
progressively developed a suitable selection and ordering strategy, but ended up with an identical result.
At the same time, this means that our work gives a natural account of \emph{why} chosen 
selection and ordering strategies work by considering the tree-like model property.

Since our disjunctive guarded saturation algorithm coincides with considered \gf resolution approaches,
we subsequently focus on the guarded saturation algorithm.
Regarding the runtime complexities, both the guarded saturation and the known \gf resolution procedures
are theoretically optimal in the general case and the case of bounded arities;
however, in case of fixed $\Sigma$, our algorithm 
stays theoretically optimal ($\ptime$) while the \gf resolution procedures' complexities
are exponential in the number of constants in~$\db$ in this case
(see \citep[Theorem~4.3]{ganzinger1999superposition} and \citep[Theorem~3.22]{de1998resolution}).

A second difference is that the guarded saturation does not require any Skolemisation step.
We demonstrate the differences between the guarded saturation and the \gf resolution procedures by means of a simple example.
This example is not meant to show any differences in performance.
An in-depth performance comparison of both approaches is open for future studies.
\begin{exmpl}
Consider the set of \gtgds
\begin{align*}
\tau_1&=R(x_1,x_2)\rightarrow\exists y_1,y_2\bigl( S(x_1,x_2,y_1,y_2)\land T(x_1,x_2,y_2)\bigr),\\
\tau_2&=S(x_1,x_2,x_3,x_4)\rightarrow U(x_4),\\
\tau_3&=T(z_1,z_2,z_3)\land U(z_3)\rightarrow P(z_1).
\end{align*} 
Then the guarded saturation first composes
$\tau_1,\tau_2$ to obtain 
\begin{equation*}
\tau_1^*\coloneqq R(x_1,x_2)\rightarrow\exists y_1,y_2\bigl( S(x_1,x_2,y_1,y_2)\land T(x_1,x_2,y_2)\land U(y_2)\bigr),
\end{equation*} 
and then composes $\tau_1^*,\tau_3$ to obtain $\tau_2^*\coloneqq R(x_1,x_2)\rightarrow P(x_1)$.

The \gf resolution procedures, on the other hand, first Skolemise $\tau_1$ to obtain 
\begin{equation*}
\tau_1'\coloneqq R(x_1,x_2)\rightarrow S\bigl(x_1,x_2,f_1(x_1,x_2),f_2(x_1,x_2)\bigr)\land T\bigl(x_1,x_2,f_2(x_1,x_2)\bigr),
\end{equation*} 
then resolve $\tau_1'$ and $\tau_2$ to obtain $\tau_2'\coloneqq R(x_1,x_2)\rightarrow U\bigl(f_2(x_1,x_2)\bigr)$,
then resolve $\tau_1'$ and $\tau_3$ to obtain 
\begin{equation*}
\tau_3'\coloneqq R(x_1,x_2)\land U(f_2(x_1,x_2))\rightarrow P(x_1),
\end{equation*} 
and finally resolve $\tau_2'$ and $\tau_3'$ to obtain $\tau_4'\coloneqq R(x_1,x_2)\rightarrow P(x_1)$.
\end{exmpl}

%% file: chapters/conclusion.tex
\chapter{Conclusion}\label{chap:conclusion}

\input{pages/conclusion/related_work}
\input{pages/conclusion/summary}

%% file: pages/conclusion/related_work.tex
\section{Related Work}

Our work was motivated by the work of \cite{de1998resolution} and \cite{ganzinger1999superposition} 
that gave a resolution-based decision procedure for the guarded fragment.
To the best of our knowledge, there has been no other attempt to relate
resolution-based procedures to the tree-like model property for any guarded logic.

There has been considerable work on rewritings for classes of \tgds.
\cite{barany2013rewriting} give a rewriting algorithm for
frontier-guarded \tgds, an extension of \gtgds. 
The algorithm is based on the tree-like model property, but works by generating
all guarded consequences of a certain shape, rather than using resolution.
This approach is not amenable to implementation.
\cite{gottlob2014expressiveness} present an algorithm for another
extension of \gtgds, namely weakly-guarded \tgds. 
The algorithm is similar to the one presented in this work, 
although again the approach will produce many unnecessary rules. 
No proofs are given in the conference paper
\cite{gottlob2014expressiveness}, and neither \cite{gottlob2014expressiveness}
nor \cite{barany2013rewriting} deal with disjunctive \tgds. 

A very recent rewriting approach for \dgtgds is given by \cite{ahmetaj2018rewriting}.
They use a game-theoretic correctness proof of their translation,
and most notably, they achieve a polynomial runtime in case of \dgtgds of bounded width.
Their translation, however, adds a vast number of rules and fresh predicate symbols, 
in particular in case of rules with unbounded width.

Moreover, many reasoning problems for description logics \citep{baader2003description} 
are special cases of the \gf satisfiability and \dgtgd Query Answering Problem \citep{cali2009datalog,hustadt2004reducing}.
These logics, however, are restricted to unary and binary relations,
whereas \gf and \dgtgds have no such limitation.

%% file: pages/conclusion/summary.tex
\section{Summary and Future Work}

Our main goal was to connect the resolution-based approaches for \gf satisfiability \citep{de1998resolution,ganzinger1999superposition}
with the tree-like model property.
We decided to work with the language of \dgtgds and its \qfqaprob as this
setting makes the connection with the tree-like model property clearer than
the \gf satisfiability problem.
After setting our problem statement, 
we introduced the chase and its disjunctive generalisation in \cref{sec:proof_machinery} and 
explained how the chase can be used to construct a universal tree-like model.

In \cref{chap:one_pass}, we gave a new approach to normalise the chase in the case 
of \gtgds to admit a one-pass behaviour.
We then generalised this result to the disjunctive chase.
This one-pass behaviour makes the connection between inferences 
and the tree structure more closely related,
and it was a key ingredient for our completeness proofs
that use the tree-like model property.
Moreover, it motivated the definition of algorithm-independent completeness properties 
for \gtgd rewritings in \cref{sec:gtgd_property_for_completeness} and \dgtgd rewritings in \cref{sec:dgtgd_property_for_completeness}, which can thus be used beyond this thesis.
As noted in \cref{sec:contribution}, the chase is a strong tool that made 
its way into numerous studies.
We believe that the one-pass property is not just a tool useful for our purposes,
but that it can give valuable insights for future investigations and new decision procedures, and that it 
can greatly facilitate proofs relying on the chase procedure.

We then derived two theoretically optimal\footnote{As stated in \cref{rmk:compl_issue_gtgd},
there is an open question whether our procedures are optimal in case of bounded arities and relation symbols, which requires further investigations.}
decision procedures for the case of \gtgds.
The first approach -- the simple saturation -- experiences an unpleasant unification issue, making it useless for practical purposes.
We then fixed this issue in our second procedure -- the guarded saturation.

We discovered that the guarded saturation can experience a better runtime 
in the case of fixed dependencies than the existing \gf resolution procedures.
Moreover the guarded saturation does not require any Skolemisation step. 
We highlighted the differences between our approach and the \gf resolution procedures
in \cref{sec:gf_vs_guarded_sat} by means of a simple example.
A practical implementation of our procedure and an extensive comparison 
with the resolution approaches for \gf is an open topic worth investigating.

We then extended our results and derived an efficient  procedure for the case of \mbox{\dgtgds\ -- }the disjunctive guarded saturation.
We then observed that the disjunctive guarded saturation coincides 
with the resolution approaches for \gf.
Our work thus gives a natural account of why the 
selection and ordering strategies of referred \gf resolution procedures work 
by considering the tree-like model property.

Last but not least, we investigated which fragment of \gf we can capture with our chosen framework.
We discovered that our approach's main restriction is the
limitation of functions to positive literals and single clauses.
We hence did not demystify the resolution approach for all of \gf, but we made a first step capturing a substantial fragment
 -- not to mention that our chosen framework has many practical applications on its own,
 as already said in \cref{sec:contribution}.

This clearly leaves the door open for further investigations that extend our results to
a more extensive fragment of \gf.
As the disjunctive guarded saturation coincides with the resolution approaches for \gf,
it is imaginable that a generalisation of our completeness proof already unveils 
such an extension.
Another step forward would be the integration of equality in our framework.
This could potentially be done with the use of so-called equality-generating 
dependencies and a generalised version of the chase, as described in \citep{deutsch2008chase}.

Lastly, just like the tree-like model property,
our approach must not be restricted to~\gf.
Many guarded logics would benefit from practical, effective decision procedures,
and we believe that ideas similar to ours can prove itself valuable for that matter.
In particular, we believe that extensions to frontier-guarded and weakly-guarded \tgds
\citep{baget2011walking} are a plausible next step,
and perhaps even the guarded negation fragment \citep{barany2011guarded} lies within the realms of possibility.

%% file: pages/appendix.tex
\chapter{Proofs for \cref{sec:gtgd_simple_saturation}}\label{apx:simple_saturation}

\simplesatresprops*
\begin{proof}
Clearly, all rules in $\vnf(\hnf(\Sigma))$ satisfy the conditions
and all derived rules will be in~variable normal form due to the final $\vnf(\cdot)$ application.

If we perform a \comp step to derive $\sigma$,
then $\sigma$ will only contain variables among $\{x_1,\dotsc,x_{\hwidth(\hnf(\Sigma))}\}$,
and hence $\width(\sigma)\leq w$.
Since the rules used to derive~$\sigma$ are full, $\sigma$ will also be full.

If we perform an \orig step to derive $\sigma=\body''\rightarrow\head''$,
the used rules $\tau,\tau'$ are of the form
\begin{align*}
	\body(\vec x)&\rightarrow \exists \vec y\,\head(\vec x,\vec y),\\
	\body'(\vec z)&\rightarrow \head'(\vec z).
\end{align*}
By definition, $\head''$ contains no $y_i$.
Moreover, $\vars\bigl(\body'\theta\setminus \head\theta\bigr)\subseteq\vec x\theta$ and $\vec x\theta\cap\vec y=\emptyset$.
So neither $\head''$ nor $\body''$ contains existentials, and thus $\sigma$ will be full.
As $\tau\in\hnf(\Sigma)$, we have ${|\vec x\theta|\leq w}$, and hence $\width(\sigma)\leq w$.
\end{proof}

\simplesatclosureprops*
\begin{proof}
The first property follows from the base case and final step of the algorithm.

For the second property, we know that $\tau,\tau'$ are of the form
\begin{align*}
	\body(\vec x)\rightarrow \head(\vec x),\\
	\body'(\vec z)\rightarrow \head'(\vec z).
\end{align*}
Without loss of generality, assume that $\consts(I_0)\subseteq\{c_1,\dotsc,c_{\hwidth(\hnf(\Sigma))}\}$,
and let $h,h'$ be the triggers used for $\tau,\tau'$.
If $h'(\body')\subseteq I_0$, setting $\tau^*\coloneqq\tau'$ does the job.
If, however, the chase step makes use of some facts in $I_1\setminus I_0$,
then define a unifier $\theta$ by
\begin{equation*}
	v\theta\coloneqq\begin{cases}
		x_i,&\text{if }v\in\vec x\text{ and }h(v)=c_i\\
		x_i,&\text{if }v\in\vec z\text{ and }h'(v)=c_i
	\end{cases}.
\end{equation*}
Note that $\vec x\theta\cup\vec z\theta\subseteq\{x_1,\dotsc,x_{\hwidth(\hnf(\Sigma))}\}$.
Further, for any atom $B'\in\body'$ with $h'(B')\in I_1\setminus I_0$,
there is an atom $H\in\head$ with $h(H)=h'(B')$, and hence,
by construction of~$\theta$, $H\theta=B'\theta$.
Now by \comp, we get the full ${\tau^*=\vnf\bigl(\body\theta\cup(\body'\theta\setminus \head\theta)\rightarrow \head'\theta\bigr)}$
in $\sclo(\Sigma)$.
Let $\rho$ be the renaming used by $\vnf(\cdot)$,
let $h''$ be the trigger defined by $h''(x_i)\coloneqq c_i$,
and set $h^*\coloneqq h''\circ\rho^{-1}$.
Then firing $\tau^*$ based on $h^*$ in $I_0$ creates $I_0\cup(I_2\setminus I_1)$.

Similarly, for the last property, we know that $\tau,\tau'$ are of the form
\begin{align*}
	\body(\vec x)&\rightarrow\exists \vec y\,\head(\vec x,\vec y)\\
	\body'(\vec z)&\rightarrow \head'(\vec z).
\end{align*}
Let $h,h'$ be the triggers used for $\tau,\tau'$, and without loss of generality, assume that $h(\body)$
only uses constants among $c_1,\dotsc,c_{\bwidth(\tau)}$.
Again, if $h(\body')\subseteq I_0$, setting $\tau^*\coloneqq\tau'$ does the job.
Otherwise, we define a unifier $\theta$ by
\begin{equation*}
	v\theta\coloneqq\begin{cases}
		x_i,&\text{if }v\in\vec x\text{ and }h(v)=c_i\\
		x_i,&\text{if }v\in\vec z\text{ and }h'(v)=c_i\\
		y_i,&\text{if }v\in\vec z\text{ and }h'(v)=h(y_i)
	\end{cases}.
\end{equation*}
Note that $\consts(h'(\body'))\subseteq\consts(h(\head))\subseteq\consts(h(\beta))\cup h(\vec y)$,
and hence \makebox{$\dom(\theta)=\vec x\cup\vec z$}.
Further, for any $B'\in\body'$ with $\vars(B'\theta)\cap\vec y\neq\emptyset$,
there must be $H\in\head$ with ${h(H)=h'(B')}$, and hence,
by construction of $\theta$, $H\theta=B'\theta$.
Consequently, $\vars(\beta'\theta\setminus\head\theta)\subseteq\vec x\theta$.
Moreover, $\vec{x}\theta \subseteq \vec{x}$, and hence $\vec{x}\theta\cap\vec{y}=\emptyset$.
Now set $\head''\coloneqq\{H'\in\head'\mid \vars(H'\theta)\cap\vec y=\emptyset\}$
and note that $\head''\neq\emptyset$ by \cref{assm:gtgd_simple_closure_props_3}.
Then by \orig, we get the full ${\tau^*=\vnf\bigl(\body\theta\cup(\body'\theta\setminus \head\theta) \rightarrow \head''\theta\bigr)}$
in $\sclo(\Sigma)$.
Let $\rho$ be the renaming used by $\vnf(\cdot)$,
let~$h''$ be the trigger defined by $h''(x_i)\coloneqq c_i$,
and set $h^*\coloneqq h''\circ\rho^{-1}$.
Then firing $\tau^*$ based on~$h^*$ in $I_0$ creates $F_r^2$.
\end{proof}

\begin{cor}\label{cor:gtgd_simple_closure_props_induct}
Let $I_0,\dotsc,I_n$ be a chase from $\sclo(\Sigma)$
and ${\tau_1,\dotsc,\tau_n\in\sclo(\Sigma)}$ be full such that
$|\consts(I_0)|\leq\hwidth(\hnf(\Sigma))$ and
firing $\tau_i$ in $I_{i-1}$ creates $I_i$.
Then there exists $\tau^*\in\sclo(\Sigma)$ such that firing $\tau^*$ in $I_0$ creates $I_0\cup(I_n\setminus I_{n-1})$.
\end{cor}
\begin{proof}
We prove the claim by induction on $n$.
For $n=1$, we just set $\tau^*\coloneqq\tau_1$.

For $n>1$, by applying \cref{lem:gtgd_simple_closure_props_2} on
$I_{n-2},I_{n-1},I_{n}$ and $\tau_{n-1},\tau_n$,
we get $\tau_{n-1}'\in\sclo(\Sigma)$ such that firing $\tau_{n-1}'$
in $I_{n-2}$ creates $I_{n-2}\cup(I_n\setminus I_{n-1})$.
Now we can apply the inductive hypothesis on $I_0,\dotsc,I_{n-2},I_{n-2}\cup(I_n\setminus I_{n-1})$ and $\tau_1,\dotsc,\tau_{n-2},\tau_{n-1}'$
to get $\tau^*\in\sclo(\Sigma)$ such that firing $\tau^*$ in $I_0$ creates
$I_0\cup\bigl((I_{n-2}\cup(I_n\setminus I_{n-1}))\setminus I_{n-2}\bigr)=I_0\cup(I_n\setminus I_{n-1})$.
\end{proof}

For the next proposition, recall the definition of evolve-completeness.
\gtgdabstrcompletenessprops*

\simplesatclosureevolvecomplete*
\begin{proof}
\cref{defn:gtgd_abstr_completeness_props_1} follows from \cref{lem:gtgd_simple_closure_props_1}.
For \cref{defn:gtgd_abstr_completeness_props_2},
we first use \cref{cor:gtgd_simple_closure_props_induct} to obtain $\tau'\in\sclo(\Sigma)$ firing $\tau'$ in $F_v^1$ creates $F_v^1\cup(F_v^{n+1}\setminus F_v^n)$.
Note that this is possible since $|\consts(F_v^1)|=\hwidth(\tau_r)\leq \hwidth(\hnf(\Sigma))$.
Let $h,h'$ be the triggers used by $\tau,\tau'$,
and let $F_r^n,I_1',I_2'$ be the chase obtained by firing~$\tau$ in $F_r^n$
based on $h$, creating a new node $v'$, and firing $\tau'$ in $F_{v'}^1\supseteq F_v^1$ based on $h'$.
Let $r'$ be the root node in this chase.
We have $I_2'\supseteq F_v^1\cup(F_v^{n+1}\setminus F_v^n)$,
and hence $F_{r'}^2\supseteq F_r^{n+1}$.
Now, by applying \cref{lem:gtgd_simple_closure_props_3} on $F_r^n,I_1',I_2'$ and $\tau,\tau'$,
we get~$\tau^*\in\sclo(\Sigma)$ such that firing $\tau^*$ in $F_r^n$ creates $F_{r'}^2$.
Thus, we have a chase from $F_r^n$ and $\sclo(\Sigma)$ of $F_r^{n+1}$.
\end{proof}

\simplesatclosuresize*
\begin{proof}
By \cref{lem:gtgd_simple_closure_resolvent_properties}, $\sclo(\Sigma)$ is a set of full \tgds in variable normal form
of width at most $w$.
Now the claim follows by \cref{lem:tgd_number_full_vnf}.
\end{proof}

\simplesatcomplexity*
\begin{proof}
The initial variable and head normal form transformations can be done in linear time by \cref{lem:tgd_vnf_complexity,def:tgd_hnf}.
The algorithm then takes all pairs in ${\hnf(\Sigma)\times\sclo(\Sigma)}$ and
$\sclo(\Sigma)^2$ to derive all possible resolvents of any pair.
By \cref{cor:gtgd_simple_closure_size}, this process takes at most~$c\coloneqq2^{2nw^a}$ iterations,
and at any point, there are at most $\calO(c^2+|\Sigma|c)$ pairs.

At any step, we have the possibility to unify any head and body atom pair of the two considered rules.
By \cref{lem:tgd_number_full_vnf,lem:gtgd_simple_closure_resolvent_properties}, there are at most $(nw^a)^{nw^a}$ such pairs.
Unfortunately, we do not only consider mgus, but all unifiers.
There are $w^{2w}$ potential unifiers.
If two atoms unify, we resolve the rules and normalise the result to obtain a new rule.
Using \cref{lem:mgu_complexity,lem:tgd_vnf_complexity}, these steps can be done in linear time.
Thus, we obtain the claimed complexity bounds.
\end{proof}

\simplesatdecproc*
\begin{proof}
Since $\sclo(\Sigma)$ is evolve-complete, it is complete by \cref{prop:gtgd_abstr_completeness}.
It is easy to show that the algorithm is sound.
Now the claim follows by \cref{prop:gtgd_atomic_rewriting_dec_proc,lem:gtgd_simple_closure_complexity}.
\end{proof}

Note that the matching lower bounds can be found in \cref{thm:lower_bounds_gtgd}.